\documentclass[fleqn,11pt,twoside]{article}

\usepackage{amsthm,amsthm,amssymb, color, xcolor,epsfig, graphics, subfigure}

\usepackage{amsmath, graphicx, latexsym, lscape }

\usepackage{threeparttable}
\usepackage{caption}
\usepackage{multirow}
\usepackage{epstopdf}

\makeatletter
\newcommand{\copyrightnote}[2]{{\renewcommand{\thefootnote}{}
 \footnotetext{\small\it
\begin{flushleft}
 \copyright \ #1   #2
\end{flushleft}}}}

\newcommand{\Name}[1]{\begin{flushleft}
                       \LARGE \bf #1
                       \end{flushleft}\vspace{-3mm}}

\newcommand{\Author}[1]{\begin{flushleft}
                       \it #1 \end{flushleft}}

\newcommand{\Address}[1]{\begin{flushleft}
                       \it #1 \end{flushleft}}

\newcommand{\Date}[1]{\begin{flushleft}
                      \small  \it #1 \end{flushleft}}

\newcommand{\evenhead}{Author \ name}
\newcommand{\oddhead}{Article \ name}

\renewcommand{\@evenhead}{
\hspace*{-3pt}\raisebox{-15pt}[\headheight][0pt]{\vbox{\hbox to \textwidth
{\thepage \hfil \evenhead}\vskip4pt \hrule}}}
\renewcommand{\@oddhead}{
\hspace*{-3pt}\raisebox{-15pt}[\headheight][0pt]{\vbox{\hbox to \textwidth
{\oddhead \hfil \thepage}\vskip4pt\hrule}}}
\renewcommand{\@evenfoot}{}
\renewcommand{\@oddfoot}{}

\long\def\@makecaption#1#2{%
  \vskip\abovecaptionskip
  \sbox\@tempboxa{\small \textbf{#1.}\ \ #2}%
  \ifdim \wd\@tempboxa >\hsize
    {\small \textbf{#1.}\ \ #2}\par
  \else
    \global \@minipagefalse
    \hb@xt@\hsize{\hfil\box\@tempboxa\hfil}%
  \fi
  \vskip\belowcaptionskip}

\newcommand{\JNMPnumberwithin}[3][\arabic]{%
  \@ifundefined{c@#2}{\@nocounterr{#2}}{%
    \@ifundefined{c@#3}{\@nocnterr{#3}}{%
      \@addtoreset{#2}{#3}%
      \@xp\xdef\csname the#2\endcsname{%
        \@xp\@nx\csname the#3\endcsname .\@nx#1{#2}}}}%
}

\renewenvironment{proof}[1][\proofname]{\par
  \normalfont
  \topsep6\p@\@plus6\p@ \trivlist
  \item[\hskip\labelsep\textbf{%
    #1\@addpunct{.}}]\ignorespaces
}{%
  \qed\endtrivlist
}

\newcommand{\resetfootnoterule} {
  \renewcommand\footnoterule{%
  \kern-3\p@
  \hrule\@width.4\columnwidth
  \kern2.6\p@}
}

\renewcommand{\footnoterule}{}

\makeatother

\theoremstyle{definition}

\newtheorem{theorem}{Theorem}
\newtheorem{proposition}{Proposition}
\newtheorem{lemma}{Lemma}

\def \h {\widehat}

\begin{document}

\renewcommand{\evenhead}{ {\LARGE\textcolor{blue!10!black!40!green}{{\sf \ \ \ ]ocnmp[}}}\strut\hfill Da-jun Zhang, Shi-min Liu and Xiao Deng }
\renewcommand{\oddhead}{ {\LARGE\textcolor{blue!10!black!40!green}{{\sf ]ocnmp[}}}\ \ \ \ \  The nonlinear Schr\"{o}dinger equation with nonzero backgrounds}

\thispagestyle{empty}
\newcommand{\FistPageHead}[3]{
\begin{flushleft}
\raisebox{8mm}[0pt][0pt]
{\footnotesize \sf
\parbox{150mm}{{Open Communications in Nonlinear Mathematical Physics}\ \  \ {\LARGE\textcolor{blue!10!black!40!green}{]ocnmp[}}
\ \ Vol.3 (2023) pp
#2\hfill {\sc #3}}}\vspace{-13mm}
\end{flushleft}}

\FistPageHead{1}{\pageref{firstpage}--\pageref{lastpage}}{ \ \ Article}

\strut\hfill

\strut\hfill

\copyrightnote{The author(s). Distributed under a Creative Commons Attribution 4.0 International License}

\Name{The solutions of classical and nonlocal nonlinear Schr\"{o}dinger equations with nonzero backgrounds:
Bilinearisation and reduction approach}

\Author{Da-jun Zhang$^{\,1}$, Shi-min Liu$^{\,1}$  and Xiao Deng$^{\,2}$}

\Address{$^{1}$ Department of Mathematics, Shanghai University, Shanghai 200444, P.R. China\\[2mm]
$^{2}$ Department of Applied Mathematics, Delft University of Technology, Delft 2628CD, The Netherlands.}

\Date{Received 13 September 2022; Accepted 31 January 2023}

\setcounter{equation}{0}

\begin{abstract}
\noindent
  In this paper we develop a bilinearisation-reduction approach to derive solutions
  to the classical and nonlocal nonlinear Schr\"{o}dinger (NLS) equations with nonzero backgrounds.
  We start from the second order Ablowitz-Kaup-Newell-Segur coupled equations
  as an unreduced system. With a pair of solutions $(q_0,r_0)$ we bilinearize the unreduced system
  and obtain solutions in terms of quasi double Wronskians.
  Then we implement reductions by introducing constraints on the column vectors of the Wronskians
  and finally obtain solutions to the reduced equations, including the classical NLS
  equation and the nonlocal NLS equations with reverse-space, reverse-time and reverse-space-time, respectively.
  With a set of plane wave solution $(q_0,r_0)$ as a background solution,
  we present explicit formulae
  for these column vectors.
  As examples, we analyze and illustrate solutions to the focusing NLS equation and the reverse-space nonlocal
  NLS equation. In particular, we present formulae for the rouge waves of arbitrary order for the
  focusing NLS equation.
\end{abstract}

\label{firstpage}


\section{Introduction}\label{sec-1}

It is common knowledge that the nonlinear Schr\"{o}dinger-type equations admit
carrier waves and  solitons,
and that breathers and other solutions (e.g. rogue waves) are the modulations of carrier waves.
Meanwhile, many (1+1)-dimensional soliton equations admit solitons with either zero
or nonzero asymptotic behaviours as $|x|\to \infty$.
As for the one of the most popular nonlinear integrable models,
the focusing nonlinear Schr\"{o}dinger (NLS) equation,
\begin{equation}
iq_t=q_{xx}+2|q|^2q,
\label{NLS-f}
\end{equation}
where $i$ is the imaginary unit, $|q|^2=qq^*$ and $q^*$ stands for the complex conjugate of $q$,
early investigations of the solutions of this equation with nonzero boundary conditions
were due to Kuznetsov \cite{Kuz-SP-1977},
Kawata and Inoue \cite{KawI-JPSJ-1977,KawI-JPSJ-1978} and Ma \cite{Ma-SAPM-1979}.
They solved the NLS equation \eqref{NLS-f}
with nonzero boundary conditions, i.e., $|q(x,t)|\rightarrow {\rm const.}$ as $x\rightarrow \pm \infty$,
by means of the inverse scattering transform.
Faddeev and Takhtajan have also done important work in this area (see for instance the monograph
Ref.\cite{Faddeev-1987} and references therein).
Besides, the NLS equation with different asymmetric nonzero boundary conditions has been studied in
\cite{Chen-Huang-1994,Huang-Chen-1993,Biondini-JMP-2014,Demontis-JMP-2014,vander-2021}.
The defocusing NLS equation,
\begin{equation}\label{NLS-df-nz}
iq_t+q_{xx}-2|q|^2q=0,
\end{equation}
has dark solitons with nonzero boundary condition ($|q|$ goes to a positive constant as $|x| \to \infty$).
Zakharov and Shabat are pioneers who studied the two NLS equations
using tools of integrability
\cite{ZakS-JETP-1972,ZakS-JETP-1973}.
Hirota derived bright soliton solution for the equation \eqref{NLS-f}
and dark soliton solution for \eqref{NLS-df-nz} by using bilinear method,
respectively in  \cite{Hirota-1973} and \cite{Hirota-Miura-book}.
For more references about the integrability of NLS equations one can refer to \cite{AblPT-book-2004}
and the references therein.

From the point of view of the Darboux transformation, for any seed solution $q_0$ of the NLS equation \eqref{NLS-f}
(i.e. $q_0$ satisfies \eqref{NLS-f}),
the envelope of the solution $q$ generated from the Darboux transformation has a form
(see equation (4.3.10) in \cite{MatS-book-1991})
\begin{equation}
|q|^2=|q_0|^2 + \partial_x^2 \ln f.
\end{equation}
In this context, when $q_0\neq 0$, we say the resulting solution $q$ is a solution with a nonzero background $q_0$.
Various methods for the systematic construction of  solutions of equation \eqref{NLS-f}
with a  plane wave background $q_0=\sqrt{\alpha}e^{-2i\alpha t}$
have been established, where $\alpha$ is a positive constant.
One  can  replace $q$ with $q e^{-2i\alpha t}$ in equation \eqref{NLS-f},
which leads it to the form
\begin{equation}
iq_t=q_{xx}+2(|q|^2-\alpha)q.
\label{NLS-f-nz}
\end{equation}
Mathematically, this implies, compared with \eqref{NLS-f},
that the envelope $|q|$ gains  a positive lift $\sqrt{\alpha}$
such that $|q|\to \sqrt{\alpha}$ as $|x| \to \infty$.
However, the plane wave background does bring interesting behavior of $|q|$ more than that.
The simplest solution (corresponding to one-soliton) of the equation \eqref{NLS-f-nz}
is a breather  \cite{KawI-JPSJ-1978,Kuz-SP-1977,Ma-SAPM-1979}, not the usual soliton.
In a special limit the breather yields a localized rational solution \cite{Per-JAMS-1983},
which is nowadays used to describe a rogue wave.
The rational solution was also derived by Matveev and Salle
via the Darboux transformation (see \S4.3 in \cite{MatS-book-1991},
where the rogue wave is called ``exulton'' solution).
The second order rational solution of equation \eqref{NLS-f-nz} was derived in 1985 in \cite{Akhmediev-JETP-1985},
using a similar way as in \cite{Per-JAMS-1983}.
The rational solution of arbitrary order of the NLS equation was first constructed
in 1986 in \cite{EleKK-SPD-1986}, where explicit formula of the solution was presented in an elegant way
and nonsingular property of the solution was proved as well  (see also  \cite{DubGKM-EJPST-2010}
for an alternative proof).
``Rogue waves'' is the name given by oceanographers to isolated large amplitude waves,
which occur more frequently than expected for normal, Gaussian distributed, statistical
events (cf.\cite{OnoRBMA-PR-2013}).
After rogue waves was observed in optic experiment in 2007 \cite{SolRKJ-Nat-2007},
it started to draw new attention and the research on rogue waves has become a hot topic.
One can refer to the review  \cite{OnoRBMA-PR-2013} for more references.
Mathematically, higher order rational solutions of the NLS equation \eqref{NLS-f-nz} can be obtained
using the Darboux transformation  via a special limit procedure \cite{GuoLL-PRE-2012,He-2013},
from a bilinear approach using reduction of the Kadomtsev-Petviashvili $\tau$ functions \cite{OhtY-PRSA-2012},
and from inverse scattering transform \cite{BM-2019}.
There are also some research of the NLS equation on the elliptic function background,
e.g. \cite{ChenJB-rspa-2018,LingLM-SAM-2019}.

In this paper, we will derive solutions with a plane wave background
for the NLS equation \eqref{NLS-f} and \eqref{NLS-df-nz} and their nonlocal versions
by using bilinear method but in a completely different way from \cite{OhtY-PRSA-2012}.

Our idea is to solve the  second order Ablowitz-Kaup-Newell-Segur (AKNS) coupled equations
\begin{subequations}\label{akns-2}
\begin{align}
iq _t=&  q_{xx}-2q^2r, \label{akns-2a}
\\
ir_t=&  -r_{xx}+2r^2q  \label{akns-2b}
\end{align}
\end{subequations}
as an unreduced system, which, for instance, yields the NLS equation \eqref{NLS-f} via reduction $r=-q^*$.
We can bilinearize this unreduced  system and present solutions of the bilinear equations in terms of
double Wronskians. Then, we impose constraints on the column vectors of the double Wronskians
so that the desired reduction holds and thus we get solutions to the reduced equation.
Such an idea was first proposed in \cite{ChenDLZ-2018,CheZ-AML-2018}
for obtaining solutions for the nonlocal integrable equations.
Nonlocal integrable systems were first systematically proposed by Ablowitz and Musslimani in 2013 \cite{AblM-PRL-2013}
and have drawn intensive attention
(e.g.\cite{Zhou-SAPM-2018,AblFLM-SAPM,YY-SAPM-2018,YJK-PRE-2018,AblM-JPA-2019,
Lou-SAPM-2019,BioW-SAPM-2019,AblLM-Nonl-2020,RaoCPMH-PD-2020,GurPKZ-PLA-2020,Lou-CTP-2020,
RybS-JDE-2021,RybS-CMP-2021}).
The bilinearisation-reduction approach has proved effective in deriving
solutions not only to the nonlocal systems but also to the classical equations
(e.g. \cite{Deng-AMC-2018,LWZ-2020,LWZ-2022,LWZ-SAPM-2022,ShiY-ND-2019,WW-CNSNS-2022,WWZ-2020}).
In this paper, we introduce transformation
\begin{equation}\label{qrfg}
q=q_0+\frac{g}{f},~~ r=r_0+\frac{h}{f}
\end{equation}
for the unreduced system \eqref{akns-2}.
Here $(q_0,r_0)$ are an arbitrary set of solution of \eqref{akns-2}.
It will be seen that for the NLS equation \eqref{NLS-f},
$|q_0|$ does act as a background of the envelope $|q|$,
see  equation (4.3.10) in \cite{MatS-book-1991} and equation \eqref{abs-q} in this paper).
In this context we also call $(q_0,r_0)$ a set of background solution of the system \eqref{akns-2}.
We will employ \eqref{qrfg} to bilinearize the unreduced system \eqref{akns-2}
and present (quasi) double Wronskian solutions to the bilinear equations.
Then we will implement reduction technique to obtain solutions to the reduced equations
listed in \eqref{cnls}-\eqref{tsnls}.

The paper is organized as follows.
In Sec.\ref{sec-2} we recall the classical and nonlocal reductions of the unreduced AKNS system \eqref{akns-2}.
In Sec.\ref{sec-3} we derive the bilinear form of \eqref{akns-2} with a set of background solution $(q_0,r_0)$
and derive (quasi) double Wronskian solutions to the bilinear equations.
In Sec.\ref{sec-4} the reduction technique is implemented
and explicit form of solutions with plane wave background solutions
are obtained for the reduced equations.
Then in Sec.\ref{sec-5} we investigate dynamics of some obtained solutions for
the classical NLS equation and the nonlocal NLS equations with nonzero backgrounds.
Finally, Sec.\ref{sec-6} serves for presenting conclusions.

\section{The second order AKNS system and its reductions} \label{sec-2}

The second order coupled AKNS  equations \eqref{akns-2},
where $q=q(x,t)$ and $r=r(x,t)$ are functions of $(x,t)\in \mathbb{R}^2$,
has been  studied as a classical coupled system in past decades.
Recently it was found that this system is related to
the cubic nonlinear Klein-Gordon equation, see \cite{AblM-JPA-2019}.
Its Lax pairs consist of the well known AKNS
(or  Zakharov-Shabat (ZS)-AKNS) spectral problem \cite{ZakS-JETP-1972,AKNS-PRL-1973,AKNS-SAM-1974},
\begin{eqnarray}\label{akns-x}
\left(
           \begin{array}{c}
             \Phi\\
             \Psi
           \end{array}
         \right)_x=M\left(
           \begin{array}{c}
             \Phi\\
             \Psi
           \end{array}
         \right),~~~~M=\left(
           \begin{array}{cc}
             \lambda & q\\
             r& -\lambda
           \end{array}
         \right),
\end{eqnarray}
and the corresponding time evolution part
\begin{eqnarray}\label{akns-t}
i\left(
           \begin{array}{c}
             \Phi\\
             \Psi
           \end{array}
         \right)_t=N\left(
           \begin{array}{c}
             \Phi\\
             \Psi
           \end{array}
         \right),~~~~N=\left(
           \begin{array}{cc}
             2\lambda^2-qr &2\lambda q+q_x\\
              2\lambda r-r_x &-2\lambda^2+qr
           \end{array}
         \right),
\end{eqnarray}
in which $\lambda$ is spectral parameter, $\lambda_t=0$, $\Phi$ and $\Psi$ are wave functions.

In the following we list possible one-component
equations reduced from equation \eqref{akns-2}. These equations will be considered in this paper.
Equation \eqref{akns-2} admits the following reductions (see \cite{AM-Nonl-2016} and reference therein)
\begin{eqnarray}
iq_t=q_{xx}-2\delta q^2q^*, &&  r=\delta q^*,\label{cnls}\\
iq_t=q_{xx}-2\delta q^2q^*(-x), &&  r=\delta q^*(-x),\label{snls}\\
iq_t=q_{xx}-2\delta q^2q(-t), &&  r=\delta q(-t),\label{tnls}\\
iq_t=q_{xx}-2\delta q^2q(-x,-t), &&  r=\delta q(-x,-t),\label{tsnls}
\end{eqnarray}
where  $\delta=\pm1$, $q(-x)=q(-x,t)$, $q(-t)=q(x,-t)$ and $q(-x,-t)$ indicate
the reverse-space, reverse-time and reverse-space-time, respectively.

\section{Bilinearisation and solutions of the AKNS system \eqref{akns-2}} \label{sec-3}

In this section, we develop the double Wronskian technique to construct exact solutions
of the  second order AKNS system \eqref{akns-2} with  nonzero background solution $(q_0, r_0)$.

\subsection{Bilinearisation}\label{sec-3-1}

Suppose that $(q_0,r_0)$ are a set of solution to the  second order AKNS system \eqref{akns-2}.
To introduce  nonzero backgrounds, we consider the dependent variable transformation (i.e. \eqref{qrfg})
\begin{eqnarray}\label{trans}
q=q_0+\frac{g}{f},~~r=r_0+\frac{h}{f},
\end{eqnarray}
with which the system \eqref{akns-2}
can be decoupled into the following bilinear form of $f, g$ and $h$,
\begin{subequations}\label{nls-b}
\begin{align}
& D^2_x f\cdot f=-2gh-2q_0hf-2r_0gf,   \label{nls-b1}\\
& (D_x^2-iD_t-2q_0r_0)g\cdot f + q_0 D_x^2 f\cdot f =0,   \label{nls-b2}\\
& (D_x^2+iD_t-2q_0r_0 )h\cdot f + r_0 D_x^2 f\cdot f=0, \label{nls-bl3}
\end{align}
\end{subequations}
where $D_x$ and $D_t$ are the well known Hirota bilinear operators defined as \cite{Hirota-1974}
\[
D_x^mD_t^nf\cdot g=(\partial_x-\partial_{x'})^m(\partial_t-\partial_{t'})^nf(x,t) g(x',t')|_{x'=x,t'=t}.
\]
Note that when $q_0=r_0=0$ the above bilinear form \eqref{nls-b}  degenerates to
the case of zero background (cf. equations (1.5.1)-(1.5.3) in \cite{CZB-SCS-2008}).

To have solutions of \eqref{nls-b}, we expand $f, g$ and $h$ as the series
\begin{subequations}\label{expand}
\begin{align}
& f(x,t)=1+f^{(1)}\varepsilon +f^{(2)}\varepsilon^2 +f^{(3)}\varepsilon^3
+\cdots+f^{(j)}\varepsilon^{j}+\cdots,\\
& g(x,t)=g^{(1)}\varepsilon+g^{(2)}\varepsilon^2+g^{(3)}\varepsilon^3+\cdots+g^{(j)}\varepsilon^{j}+\cdots,\\
& h(x,t)=h^{(1)}\varepsilon+h^{(2)}\varepsilon^2+h^{(3)}\varepsilon^3+\cdots+h^{(j)}\varepsilon^{j}+\cdots,
\end{align}
\end{subequations}
where $\varepsilon$ is an arbitrary number, $\{f^{(j)}, g^{(k)}, h^{(l)}\}$ are functions to be determined.
Consider a special case,
\begin{eqnarray}\label{q0r0-nss}
q_0=A_0e^{2iA_0^2t},~~r_0=A_0e^{-2iA_0^2t},
\end{eqnarray}
where $A_0$ is an arbitrary constant.
By calculation we can find out 1- and 2-soliton solutions, which agree with the following expressions
with $(\mu_1,\mu_2)=(1,0)$ and $(\mu_1,\mu_2)=(1,1)$:
\begin{subequations}
\begin{align}
  & f= 1+\mu_1\exp(\xi_1+a_1)+\mu_2\exp(\xi_2+a_2)+\mu_1\mu_2\exp(\xi_1+\xi_2+A_{12}),\\
& g=\biggl[\mu_1\exp(\xi_1+b_1)+\mu_2\exp(\xi_2+b_2)+\mu_1\mu_2\exp(\xi_1+\xi_2+B_{12})\biggr]e^{2iA_0^2t},\\
& h=\biggl[\mu_1\exp(\xi_1+c_1)+\mu_2\exp(\xi_2+c_2)+\mu_1\mu_2\exp(\xi_1+\xi_2+C_{12})\biggr]e^{-2iA_0^2t},
\end{align}
where (for $j=1,2$)
\begin{align}
& \xi_j = k_j x+\omega_j t+\xi_j^{(0)},~~
\omega_j=\pm\sqrt{k_j^4-4 A_0^2 k_j^2},~~\mu_j\in\{0,1\},\\
&e^{a_j}={-2A_0},~~ e^{b_j}={k_j^2- w_j},~~
e^{c_j}={k_j^2+w_j},\\
& e^{A_{12}}=-\frac{2k_1^2 k_2^2-2\omega_1\omega_2-4A_0^2 \left(k_i^2+k_j^2\right) }{(k_1+k_2)^2},\\
& e^{B_{12}}=-\frac{2A_0(k_1-k_2) \left(k_1^2-k_2^2-\omega_1+\omega_2\right)}{k_i+k_j},\\
& e^{C_{12}}=-\frac{2A_0(k_1-k_2) \left(k_1^2-k_2^2+\omega_1-\omega_2\right)}{k_i+k_j}.
\end{align}
\end{subequations}
We should mention that
the reduction $A_0=0$ of the above expression does not gives rise to
a significant solution for the unreduced system \eqref{akns-2} with a zero background.

Note that in \cite{LP-TMP-2007} the system \eqref{akns-2} was also bilinearized via the following transformation
\[q=\frac{G}{F},~~ r= \frac{H}{F},\]
and the bilinear form is
\begin{align*}
& (D_x^2-iD_t-\lambda)G\cdot F =0,    \\
& (D_x^2+iD_t-\lambda )H\cdot F=0, \\
& (D^2_x -\lambda)F \cdot F=-2GH,
\end{align*}
where $\lambda\in \mathbb{R}$.
One-soliton and two-soliton solutions they derived
(see equation (58) and (82) in \cite{LP-TMP-2007})
are shown to be associated with a more general plane wave background (see \eqref{q0r0-gen},
which degenerates to \eqref{q0r0-nss} when $a=0$ and $A_0=B_0$).

\subsection{Quasi double Wronskian solutions of the AKNS system \eqref{akns-2}}\label{sec-3-2}

We now derive double Wronskian solutions of the second order  AKNS system \eqref{akns-2}.

We extend the Lax pair \eqref{akns-x} and \eqref{akns-t} to the following matrix system
\begin{eqnarray}\label{pp-x}
\left(
           \begin{array}{c}
             \phi\\
             \psi
           \end{array}
         \right)_x=\mathcal{M}\left(
           \begin{array}{c}
             \phi\\
             \psi
           \end{array}
         \right),~~~~\mathcal{M}=\left(
           \begin{array}{cc}
             A & q_0I_{2m+2}\\
             r_0I_{2m+2}& -A
           \end{array}
         \right),
\end{eqnarray}
\begin{eqnarray}\label{pp-t}
i\left(
           \begin{array}{c}
             \phi\\
             \psi
           \end{array}
         \right)_t=\mathcal{N}\left(
           \begin{array}{c}
             \phi\\
             \psi
           \end{array}
         \right),~~~~\mathcal{N}=\left(
           \begin{array}{cc}
             2A^2-q_0r_0I_{2m+2} &2A q_0+q_{0,x}I_{2m+2}\\
              2A r_0-r_{0,x}I_{2m+2} &-2A^2+q_0r_0I_{2m+2}
           \end{array}
         \right),
\end{eqnarray}
where $A \in\mathbb{C}_{(2m+2)\times (2m+2)}$ is an arbitrary complex matrix,
$I_{2m+2}$ is the $(2m+2)$-th order identity matrix,
$\phi$ and $\psi$ are two $(2m+2)$-th column vectors.
Introduce vectors $\phi_{k}$ and $\psi_{k}$ by
\begin{eqnarray}\label{3.8}
\phi_{k}=A^k\phi, & \psi_{k}=(-A)^k\psi,
\end{eqnarray}
and define determinants\footnote{When $m=0$ we have
$g=2|\phi_{0}, \phi_1|, ~ h=-2|\psi_{0},\psi_1|$.}
\begin{eqnarray}\label{fgh}
f=|\h\phi_m;\h\psi_m|, & g=2|\h\phi_{m+1};\h\psi_{m-1}|, & h=-2|\h\phi_{m-1};\h\psi_{m+1}|,
\end{eqnarray}
where $\h\phi_m$ stands for the consecutive columns $(\phi_0,\phi_1,\cdots,\phi_m)$.
Then, solutions of the bilinear system \eqref{nls-b} are described as the following.

\begin{theorem}\label{th1}
The bilinear system \eqref{nls-b} has  solutions \eqref{fgh},
where $\phi$ and $\psi$ in \eqref{3.8} satisfy \eqref{pp-x} and \eqref{pp-t},
and $(q_0, r_0)$ are given solutions of the system \eqref{akns-2}.
Furthermore, matrix $A$ and any matrix that is similar to it lead to the same solution of
the AKNS system \eqref{akns-2} through the transformation \eqref{trans}.
\end{theorem}

The proof will be sketched in Appendix \ref{app-A}.
Later we only need to consider the canonical forms of $A$, i.e. $A$ being diagonal or a Jordan block.

Strictly speaking, the above $f,g,h$ in \eqref{fgh} are not double Wronskians that are defined by arranging columns
by increasing the order of derivatives of $\phi$ and $\psi$.
We may call them  quasi double Wronskians.
Note that when $A$ is diagonal the results in Theorem \ref{th1} are the same as
those derived from the Darboux transformation (cf. \S4.2 in \cite{MatS-book-1991}).
When $A$ is a Jordan block, the corresponding solutions can be obtained
using a limit procedure from those solutions which are derived from a diagonal matrix $A$
(e.g.  \S4.3 in \cite{MatS-book-1991}).
We also note that when the background solution $(q_0, r_0)$ is independent of $x$,
we may covert  $f,g,h$ given in \eqref{fgh} to double Wronskians.

\begin{theorem}\label{th2}
Suppose that the $(2m+2)\times (2m+2)$ double Wronskians\footnote{When $m=0$, $g$ and $h$ take the form
$g=-2|\phi,\partial_x \phi|, ~ h=2|\psi,\partial_x \psi|$.}
\begin{equation}\label{fgh-2}
f=(-1)^m|\h m;\h m|,  ~ g=-2|\h{m+1};\h{m-1}|, ~ h=2|\h{m-1};\h{m+1}|,
\end{equation}
where $|\widehat{n}; \widehat{m}|$  denotes a $(m+n+2)$ double Wronskian
defined as (see \cite{Nimmo-NLS-1983})
\begin{equation*}
|\widehat{m+j}; \widehat{m-j}|=|0,1,\cdots,m+j; 0,1,\cdots, m-j|
=|\phi,\partial_x \phi,\ldots,\partial_x^{m+j} \phi;
\psi,\partial_x \psi,\ldots,\partial_x^{m-j} \psi|,
\end{equation*}
$\phi$ and $ \psi $ are $(2m+2)$-th order column vectors.
When $\phi$ and $ \psi $ meet the condition
\begin{subequations}\label{wron-cond-x}
\begin{align}
& \phi_x=A\phi+q_0\psi,\quad i\phi_t=2\phi_{xx}-q_0r_0\phi-2q_0\psi_x,\label{wron-cond-x-a}\\
& \psi_x=r_0 \phi-A\psi,\quad  i\psi_t=-2\psi_{xx}+q_0r_0\psi+2r_0\phi_x,\label{wron-cond-x-b}
\end{align}
\end{subequations}
where $A$ is a $(2m+2)\times (2m+2)$ complex matrix,
$(q_0, r_0)$ satisfy \eqref{akns-2} but are independent of $x$,
then $f,g,h$ defined in \eqref{fgh-2} are solutions  to the bilinear equations \eqref{nls-b}.
Furthermore, matrix $A$ and any matrix that is similar to it lead to
the same solution to the AKNS system \eqref{akns-2}
through the transformation \eqref{trans}.
\end{theorem}

The proof will be given in Appendix \ref{app-B}.

\section{Reduction and solutions} \label{sec-4}

For convenience we call \eqref{akns-2} the unreduced system and \eqref{cnls}-\eqref{tsnls}
the reduced equations.
In the previous section we have already obtained solutions in terms of quasi double Wronskians
\eqref{fgh} (see Theorem \ref{th1} for the unreduced system \eqref{akns-2}).
In this section we implement reductions by imposing constraints on $A$ and $\psi$
so that \eqref{fgh} can provides solutions to the reduced equations  \eqref{cnls}-\eqref{tsnls}.
Such a reduction technique was first introduced in \cite{ChenDLZ-2018,CheZ-AML-2018}.

\subsection{Reduction of the Wronskian solution}\label{sec-4-1}

Let us directly present results and then prove them.

\begin{theorem}\label{th3}
Let  $A$ and $T$ be matrices in $\mathbb{C}_{(2m+2)\times (2m+2)}$.
Solutions of the reduced equations \eqref{cnls}-\eqref{tsnls} are given in the following, respectively.

\noindent
 (1) The classical NLS  equation \eqref{cnls} has solution
\begin{subequations}
\begin{align}
& q=q_0+\frac{g}{f},\\
&f=|\h\phi_m;T\h\phi_m^*|,~~ g=2|\h\phi_{m+1};T\h\phi_{m-1}^*|, \label{4.1b}
\end{align}
\end{subequations}
where $q_0$ is a  solution of  equation \eqref{cnls}
such that $r_0^*=\delta q_0$,
vector $\phi$ is a solution of matrix equations
\begin{subequations}\label{4.2}
\begin{align}
& \phi_x=A\phi+q_0T\phi^*,\label{4.2a}\\
& i\phi_t=(2A^2-\delta q_0q_0^*I_{2m+2})\phi+(2A q_0+q_{0x}I_{2m+2})T\phi^*,\label{4.2b}
\end{align}
\end{subequations}
and $A$ and $T$  obey the relation
\begin{equation}\label{cnls-AT}
AT+TA^*=0, ~~ TT^*=\delta I_{2m+2}, ~~ \delta=\pm 1.
\end{equation}

\noindent
(2) For the reverse-space nonlocal NLS  equation \eqref{snls}, its solution is given by
\begin{subequations}
\begin{align}
&q=q_0+\frac{g}{f},\\
& f= (-1)^{\frac{m(m+1)}{2}}|\h\phi_m;T\h\phi_m^*(-x)|,
~~ g=2(-1)^{\frac{m(m-1)}{2}}|\h\phi_{m+1};T\h\phi_{m-1}^*(-x)|,
\end{align}
\end{subequations}
where $q_0$ is a solution of equation \eqref{snls} such that
$r_0^*(x)=\delta q_0(-x)$, vector $\phi$ is a solution of matrix equations
\begin{subequations}\label{4.5}
\begin{align}
& \phi_x(x)=A\phi(x)+q_0(x)T\phi^*(-x),\label{4.5a}\\
& i\phi_t(x)=(2A^2-\delta q_0(x)q_0^*(-x)I_{2m+2})\phi(x)+(2A q_0+q_{0,x}(x)I_{2m+2})T\phi^*(-x),\label{4.5b}
\end{align}
\end{subequations}
and $A$ and $T$ obey the relation
\begin{equation}\label{snls-AT}
AT-TA^*=0, ~~ TT^*=-\delta I_{2m+2},~~ \delta=\pm 1.
\end{equation}

\noindent
(3) For the reverse-time nonlocal NLS  equation \eqref{tnls}, its solution is given by
\begin{subequations}
\begin{align}
&q=q_0+\frac{g}{f},\\
&f=|\h\phi_m;T\h\phi_m(-t)|, ~~g=2|\h\phi_{m+1};T\h\phi_{m-1}(-t)|,
\end{align}
\end{subequations}
where $q_0$ is a solution of equation \eqref{tnls} such that
$r_0(t)=\delta q_0(-t)$, $\phi$ is a solution of matrix equations
\begin{subequations}\label{4.8}
\begin{align}
&\phi_x(t)=A\phi(t)+q_0(t)T\phi(-t),\\
&i\phi_t(t)=(2A^2-\delta q_0(t)q_0(-t)I_{2m+2})\phi(t)+(2A q_0(t)+q_{0,x}(t)I_{2m+2})T\phi(-t),
\end{align}
\end{subequations}
and $A$ and $T$ obey the relation
\begin{equation}\label{tnls-AT}
AT+TA=0, ~~ T^2=\delta I_{2m+2}, ~~ \delta=\pm 1.
\end{equation}

\noindent
(4) For the reverse-space-time nonlocal NLS  equation \eqref{tsnls}, its solution is given by
\begin{subequations}
\begin{align}
&q=q_0+\frac{g}{f},\\
& f=(-1)^{\frac{m(m+1)}{2}}|\h\phi_m;T\h\phi_m(-x,-t)|,
~~g=2(-1)^{\frac{m(m-1)}{2}}|\h\phi_{m+1};T\h\phi_{m-1}(-x,-t)|,
\end{align}
\end{subequations}
where $q_0$ is a solution of equation \eqref{tsnls}
such that $r_0(x,t)=\delta q_0(-x,-t)$, vector $\phi$ is a solution of matrix equations
\begin{subequations}\label{4.11}
\begin{align}
& \phi_x(x,t)=A\phi(x,t)+q_0(x,t)T\phi(-x,-t),\\
&i\phi_t(x,t)=(2A^2-\delta q_0(x,t)q_0(-x,-t))\phi(x,t)+(2A q_0(x,t)+q_{0,x}(x,t))T\phi(-x,-t),
\end{align}
\end{subequations}
and $A$ and $T$ obey the relation
\begin{equation}\label{tsnls-AT}
AT-TA=0, ~~ T^2=-\delta I_{2m+2},~~ \delta=\pm 1.
\end{equation}
\end{theorem}

\begin{proof}

We employ the classical NLS  equation \eqref{cnls} as an illustrating example.
Introduce constraint on $\psi$,
\begin{eqnarray}\label{T1}
\psi=T\phi^*,
\end{eqnarray}
where $T$ is a certain matrix  in $\mathbb{C}_{(2m+2)\times (2m+2)}$.
First, it can be verified that when $r_0=\delta q_0^*$ and
 $A$ and $T$ satisfy \eqref{cnls-AT},
 the constraint \eqref{T1} reduces \eqref{pp-x} and \eqref{pp-t} to \eqref{4.2}.
In fact, taking \eqref{pp-x} as an example, under \eqref{T1} and  $r_0=\delta q_0^*$,
we rewrite \eqref{pp-x} as
\begin{subequations}\label{4.14}
\begin{align}
& \phi_x=A\phi+q_0 T \phi^*, \label{4.14a}\\
& T\phi_x^*=\delta q_0^* \phi-AT \phi^*, \label{4.14b}
\end{align}
\end{subequations}
where \eqref{4.14a} is nothing but \eqref{4.2a}.
Making use of \eqref{cnls-AT},
equation \eqref{4.14b} multiplied by $\delta T^*$ from the left gives rise to the complex conjugate of \eqref{4.14a}.
This indicates \eqref{pp-x} reduces to \eqref{4.2a}.
In a similar way one can find \eqref{pp-t} reduces to  \eqref{4.2b} in this case.

Next, with the constraint \eqref{T1}, we can rewrite the quasi double Wronskians \eqref{fgh} as
\begin{subequations}
\begin{eqnarray}
f&=&|\h\phi_m;\h\psi_m|=|\h\phi_m;T\h\phi_m^*|,\\
g&=&2|\h\phi_{m+1};\h\psi_{m-1}|=2|\h\phi_{m+1};T\h\phi_{m-1}^*|,\\
h&=&-2|\h\phi_{m-1};\h\psi_{m+1}|=-2|\h\phi_{m-1};T\h\phi_{m+1}^*|.
\end{eqnarray}
\end{subequations}
Making use of \eqref{cnls-AT} we find that
\[f =|\h\phi_m;T\h\phi_m^*|=|T||\delta T^*\h\phi_m;\h\phi_m^*|.\]
Then, switching the first $(m+1)$ columns and the last $(m+1)$ columns and picking the parameter $\delta$ out
yield
\[f= (-\delta)^{m+1}|T||\h\phi_m^*;T^*\h\phi_m|=(-\delta)^{m+1}|T|f^*.\]
In a similar way we can prove
\[h=-(-\delta)^m|T|g^*.\]
Thus we have
\begin{eqnarray*}
\frac{r}{q^*}=\frac{r_0+h/f}{q_0^*+g^*/f^*}=\frac{\delta q_0^*+\delta g^*/f^*}{q_0^*+g^*/f^*}=\delta,
\end{eqnarray*}
i.e. $r=\delta q^*$, which is the reduction by which we get \eqref{cnls} from  \eqref{akns-2}.

The proof of nonlocal cases is similar to the classical one.
For the reverse-space nonlocal NLS  equation \eqref{snls},
the reduction is implemented by taking
\begin{eqnarray}\label{T2}
\psi=T\phi^*(-x)
\end{eqnarray}
together with \eqref{snls-AT}.
Here and below we note that we do not write out independent variables
unless the inverse of them are involved.
Relations between Wronskians are
\begin{eqnarray*}
f=\delta^{m+1}|T|f^*(-x), &
h=\delta^m|T|g^*(-x),
\end{eqnarray*}
which yield
\begin{eqnarray*}
\frac{r}{q^*(-x)}=\frac{r_0+h/f}{q_0^*(-x)+g^*(-x)/f^*(-x)}=\frac{\delta q_0^*(-x)+\delta g^*(-x)/f^*(-x)}{q_0^*(-x)+g^*(-x)/f^*(-x)}=\delta,
\end{eqnarray*}
i.e. $r=\delta q^*(-x)$, which reduces the unreduced system \eqref{akns-2} to equation \eqref{snls}.

For the reverse-time nonlocal NLS  equation \eqref{tnls},
the reduction is implemented by taking
\begin{eqnarray}\label{T3}
\psi=T\phi(-t)
\end{eqnarray}
together with \eqref{tnls-AT}. Relations between Wronskians are
\begin{eqnarray*}
f=(-\delta)^{m+1}|T|f(-t),&
h=-(-\delta)^m|T|g(-t),
\end{eqnarray*}
which yield
\begin{eqnarray*}
\frac{r}{q(-t)}=\frac{r_0+h/f}{q_0(-t)+g(-t)/f(-t)}=\frac{\delta q_0(-t)+\delta g(-t)/f(-t)}{q_0(-t)+g(-t)/f(-t)}=\delta,
\end{eqnarray*}
i.e. $r=\delta q(-t)$, which reduces \eqref{akns-2} to \eqref{tnls}.

For the reverse-space-time nonlocal NLS  equation \eqref{tsnls},
we start from
\begin{eqnarray}\label{T4}
\psi=T\phi(-x,-t)
\end{eqnarray}
and \eqref{tsnls-AT}.
Relations between Wronskians are
\begin{eqnarray*}
f=\delta^{m+1}|T|f(-x,-t), &
h=\delta^m|T|g(-x,-t),
\end{eqnarray*}
which yield
\begin{eqnarray*}
\frac{r}{q(-x,-t)}=\frac{r_0+h/f}{q_0(-x,-t)+g(-x,-t)/f(-x,-t)}
=\frac{\delta q_0(-x,-t)+\delta g(-x,-t)/f(-x,-t)}{q_0(-x,-t)+g(-x,-t)/f(-x,-t)}=\delta,
\end{eqnarray*}
i.e. $r=\delta q(-x,-t)$, which reduces \eqref{akns-2} to  \eqref{tsnls}.

\end{proof}

\subsection{Matrices $A$ and $T$}\label{sec-4-2}

We look for explicit forms of  $A$ and $T$ in Theorem \ref{th3}.
Equations \eqref{cnls-AT} and \eqref{snls-AT} can be unified to be
\begin{equation}\label{con-AT-1}
 AT+\sigma TA^*=0,\quad TT^{*}=\sigma\delta I,~~\sigma,\delta=\pm 1,
\end{equation}
and equations \eqref{tnls-AT} and \eqref{tsnls-AT} can be unified to be
\begin{equation}\label{con-AT-2}
 AT+\sigma TA=0,\quad T^{2}=\sigma\delta I,~~\sigma,\delta=\pm 1.
\end{equation}
Consider special solutions to these matrix equations,
i.e. $A$ and $T$ are block matrices
\begin{align}\label{nls-diagA-T}
 A=\left(\begin{array}{cc}
                   K_1 & \mathbf{0} \\
                   \mathbf{0} & K_4
                 \end{array}
               \right),\quad
               T=\left(
     \begin{array}{cc}
       T_1 & T_2 \\
       T_3 & T_4
     \end{array}
   \right),
\end{align}
where $T_i$ and $K_i$ are $(m+1)\times (m+1)$  matrices.
Then solutions to equations \eqref{con-AT-1} and \eqref{con-AT-2}  can be listed out as in the Table \ref{table-5-1}
and \ref{table-5-2},
cf.\cite{ChenDLZ-2018}.

\begin{table}[htbp]
\centering
\captionsetup{font={small}}
\caption{$T$ and~$A$ for equation~\eqref{con-AT-1}}\label{table-5-1}
\begin{tabular}{|c|c|c|c|c|}
\hline
                              &            $(\sigma,\delta)$           &  $T$              &  $A$                                  \\
\hline
\multirow{2}{*}{\eqref{cnls-AT}} & $(1,1)$  &
$T_{1}=T_{4}=\mathbf{0}_{m+1},T_{2}=T_{3}=\mathbf{I}_{m+1}$  &
\multirow{2}{*}{$K_{1}=\mathbf{K}_{m+1}, K_{4}=-\mathbf{K}^*_{m+1}$}    \\
    \cline{2-3}
                                 & $(1,-1)$ & $T_{1}=T_{4}=\mathbf{0}_{m+1},T_{2}=-T_{3}=\mathbf{I}_{m+1}$  & \\
    \cline{1-4}
\multirow{2}{*}{\eqref{snls-AT}} & $(-1,1)$ &
$T_{1}=T_{4}=\mathbf{0}_{m+1},T_{2}=-T_{3}=\mathbf{I}_{m+1}$ &
\multirow{2}{*}{$K_{1}=\mathbf{K}_{m+1}, K_{4}=\mathbf{K}^*_{m+1}$} \\
    \cline{2-3}
                                 &  $(-1,-1)$ &  $T_{1}=T_{4}=\mathbf{0}_{m+1},T_{2}=T_{3}=\mathbf{I}_{m+1}$  &  \\
    \cline{1-4}
\hline
\end{tabular}
\end{table}

\begin{table}[htbp]
\centering
\captionsetup{font={small}}
\caption{$T$ and~$A$ for equation~\eqref{con-AT-2}}\label{table-5-2}
\begin{tabular}{|c|c|c|c|c|}
\hline
                              &            $(\sigma,\delta)$           &  $T$              &  $A$                                  \\
\hline
\multirow{2}{*}{\eqref{tnls-AT}} & $(1,1)$  &
$T_{1}=T_{4}=\mathbf{0}_{m+1},T_{2}=T_{3}=\mathbf{I}_{m+1}$  &
\multirow{2}{*}{$K_{1}=\mathbf{K}_{m+1}, K_{4}=-\mathbf{K}_{m+1}$}    \\
    \cline{2-3}
                                 & $(1,-1)$ & $T_{1}=T_{4}=\mathbf{0}_{m+1},T_{2}=-T_{3}=\mathbf{I}_{m+1}$  & \\
    \cline{1-4}
\multirow{2}{*}{\eqref{tsnls-AT}} & $(-1,1)$ &
$T_{1}=-T_{4}=i\mathbf{I}_{m+1},T_{2}=T_{3}=\mathbf{0}_{m+1}$ &
\multirow{2}{*}{$K_{1}=\mathbf{K}_{m+1}, K_{4}=\mathbf{H}_{m+1}$} \\
    \cline{2-3}
                                 &  $(-1,-1)$ & $T_{1}=-T_{4}=\mathbf{I}_{m+1},T_{2}=T_{3}=\mathbf{0}_{m+1}$  &  \\
    \cline{1-4}
\hline
\end{tabular}
\end{table}

In addition, equation \eqref{snls-AT} admits real solution in the form \eqref{nls-diagA-T} for the case $\delta=-1$,
where
\begin{subequations}\label{4.6-real}
\begin{align}
& K_{1}=\mathbf{K}_{m+1},
K_{4}=\mathbf{H}_{m+1}, ~\mathbf{K}_{m+1},\mathbf{H}_{m+1}\in \mathbb{R}_{(m+1)\times(m+1)},
\label{s-nls-f-A}\\
& T_{1}=-T_{4}=\mathbf{I}_{m+1},T_{2}=T_{3}=\mathbf{0}_{m+1},\label{s-nls-f-T1}\\
& \text{or} ~T_{1}=T_{4}=\mathbf{I}_{m+1},T_{2}=T_{3}=\mathbf{0}_{m+1}. \label{s-nls-f-T2}
\end{align}
\end{subequations}
Besides, equation \eqref{cnls-AT} with $\delta=1$ can have pure imaginary solution \eqref{nls-diagA-T}
and \eqref{4.6-real}
where in \eqref{s-nls-f-A} $\mathbf{K}_{m+1},\mathbf{H}_{m+1}\in i\mathbb{R}_{(m+1)\times(m+1)}$.

Due to the fact that $A$ and any matrix that is similar to it generate same solutions to the system \eqref{akns-2}
(see Theorem \ref{th1}), we only need to consider the canonical forms\footnote{
A general case for $\mathbf{K}_{m+1}$ is the block diagonal form
\begin{equation*}
\mathbf{K}_{m+1}=\mathrm{Diag}\left(J_{h_1}(k_1),J_{h_2}(k_2),\cdots,J_{h_s}(k_s),
\mathrm{Diag}(k_{s+1},\cdots, k_{s+n})\right),
\end{equation*}
where each $J_{h_j}(k_j)$ is an $h_j\times h_j$ Jordon block matrix defined as \eqref{Jordan},
$\mathrm{Diag}(k_{s+1},\cdots, k_{s+n})$ is an $n\times n$ diagonal matrix and $n+\sum_{j=1}^{s} h_j=m+1$.
In this case, $\phi$ is just composed accordingly since \eqref{pp-x} and \eqref{pp-t}
are linear system of $\phi$ and $\psi$.
Thus, we will only consider two limiting cases, \eqref{diag} and \eqref{Jordan}.
Note that  matrix $A$ corresponds to the eigenvalues of the AKNS
spectral problem \eqref{akns-x},
which means \eqref{diag} is  for the case of simple distinct eigenvalues
and \eqref{Jordan} is for the case of one eigenvalue of geometric multiplicity
one and algebraic multiplicity $m + 1$.}
of $A$.
That is, $\mathbf{K}_{m+1}$ can either be
\begin{equation}
\mathbf{K}_{m+1}=\mathrm{Diag}(k_1, k_2, \cdots, k_{m+1}),~~ k_i \in \mathbb{C},
\label{diag}
\end{equation}
or $\mathbf{K}_{m+1}=J_{m+1}(k),~ k\in \mathbb{C}$, where
\begin{equation}\label{Jordan}
J_{m+1}(k)=\begin{pmatrix}
k & 0 & 0 & \ldots & 0 & 0\\
1 & k & 0 & \ldots & 0 & 0\\
\ldots & \ldots & \ldots & \ldots & \ldots & \ldots\\
0 & 0 & 0 & \ldots & 1 & k
\end{pmatrix}_{(m+1)\times (m+1)}.
\end{equation}


\subsection{Plane wave background solution $q_0$}\label{sec-4-3}

Considering the expression \eqref{trans} (and also \eqref{abs-q}) we can call $q_0$ and $r_0$ to
be background solutions of
$q$ and $r$, respectively.
The unreduced  system \eqref{akns-2} admits a set of plane wave solutions
\begin{equation}\label{q0r0-gen}
q_0=A_0e^{i[ax+(a^2+2A_0B_0)t]}, ~~ r_0=B_0e^{-i[ax+(a^2+2A_0B_0)t]},
\end{equation}
where $A_0, B_0$ and $a$ are arbitrary complex constants.
It is easy to find that the reduced classical and nonlocal NLS equations \eqref{cnls}-\eqref{tsnls}
admit the following solutions, respectively,
\begin{subequations}\label{q0}
\begin{align}
& q_0=A_0e^{i[ax+(a^2+2\delta |A_0|^2)t]}, ~~~ (A_0\in \mathbb{C}, ~a\in \mathbb{R}), \label{q0-a}\\
& q_0=A_0e^{-ax+i(-a^2+2\delta |A_0|^2) t}, ~~~ (A_0\in \mathbb{C}, ~a\in \mathbb{R}),\label{q0-b} \\
& q_0=A_0e^{2i\delta A_0^2t}, ~~~ (A_0\in \mathbb{C}), \label{q0-c} \\
& q_0=A_0e^{i[ax+(a^2+2\delta A_0^2)t]}, ~~~ (A_0,a\in \mathbb{C}). \label{q0-d}
\end{align}
\end{subequations}
Our purpose is to write out explicit Wronskian vectors $\phi$ that respectively satisfy
the conditions \eqref{4.2}, \eqref{4.5}, \eqref{4.8} and \eqref{4.11} for given background solutions $q_0$.
We are going to consider the simple case where $q_0$ are given in \eqref{q0}.
If making use of some symmetries, we may start from a simpler background solution
\begin{equation}\label{q00}
q_0=e^{2i\delta t}
\end{equation}
instead of \eqref{q0}.

\begin{proposition}\label{Prop-1}
The classical and nonlocal NLS equations \eqref{cnls}-\eqref{tsnls} admit the following symmetries.

\noindent
(1) Classical NLS equation \eqref{cnls}:
\begin{itemize}
\item{ Galilean symmetry: if $q(x,t)$ solves the NLS equation \eqref{cnls} with the background
solution   $q_0$ given in \eqref{q0-a},
then $Q(X,Y)=q(x,t)e^{-iax-ia^2t}$,
where $X=x+2at$ and $Y=t$,   solves the NLS equation \eqref{cnls} with $Q(X,Y)$, i.e.
\begin{equation}\label{NLS-Q}
iQ_Y =Q_{XX}-2\delta Q^2Q^*,
\end{equation}
of which $Q_0(X,Y)=q_0 e^{-iax-ia^2t}= A_0e^{2i\delta |A_0|^2 Y}$ is a solution.
}
\item{ Scaling symmetry: if $q(x,t)$ solves the NLS equation \eqref{cnls} with a background solution
$q_0(x,t)=A_0e^{2i\delta |A_0|^2 t}$,
then $Q(X,Y)=\frac{1}{A_0}q(x,t)$, where $X=|A_0|x$ and $Y=|A_0|^2t$,
also solves the NLS equation \eqref{NLS-Q},
of which $Q_0(X,Y)=\frac{1}{A_0}q_0=e^{2i\delta Y}$ is a  solution.}
\end{itemize}

\noindent
(2) Reverse-space nonlocal NLS equation \eqref{snls}:
\begin{itemize}
\item{ Galilean symmetry: if $q(x,t)$ solves the reverse-space NLS equation \eqref{snls}
with the background solution $q_0(x,t)$ given in \eqref{q0-b},
then $Q(X,Y)=q(x,t)e^{ax+ia^2t}$, where $X=x+2iat$ and $Y=t$,
also solves the reverse-space NLS equation \eqref{snls} with $Q(X,Y)$, i.e.
\begin{equation}\label{NLS-Qx}
iQ_Y =Q_{XX}-2\delta Q^2Q^*(-X),
\end{equation}
of which $Q_0(X,Y)=q_0e^{ax+ia^2t}=A_0e^{2i\delta |A_0|^2 Y}$ is a  solution.
}
\item{Scaling symmetry: if $q(x,t)$ solves the reverse-space NLS equation \eqref{snls}
with a background solution $q_0(x,t)=A_0e^{2i\delta |A_0|^2 t}$,
then $Q(X,Y)=\frac{1}{A_0}q(x,t)$, where $X=|A_0|x,~ Y=|A_0|^2t$, also solves
the reverse-space NLS equation \eqref{NLS-Qx} of which
$Q_0(X,Y) = \frac{1}{A_0}q_0=e^{2i\delta Y}$ is a solution.}
\end{itemize}

\noindent
(3) Reverse-time nonlocal NLS equation \eqref{tnls}:
\begin{itemize}
\item{ Scaling symmetry: if $q(x,t)$ solves the reverse-time NLS equation \eqref{tnls}
with the background solution $q_0(x,t)$ given in \eqref{q0-c},
then $Q(X,Y)=\frac{1}{A_0}q(x,t)$, where $X=A_0x$ and $Y=A_0^2t$,
also solves the reverse-time NLS equation \eqref{tnls} with $Q(X,Y)$, i.e.
\begin{equation}\label{NLS-Qt}
iQ_Y =Q_{XX}-2\delta Q^2Q(-Y),
\end{equation}
of which  $Q_0(X,Y)=\frac{1}{A_0}q_0=e^{2i\delta Y}$ is a solution.
}
\end{itemize}

\noindent
(4) Reverse-space-time nonlocal NLS equation \eqref{tsnls}:
\begin{itemize}
\item{ Galilean symmetry: if $q(x,t)$ solves the reverse-space-time NLS equation \eqref{tsnls} with
the background solution $q_0(x,t)$ given in \eqref{q0-d},
then $Q(X,Y)=q(x,t)e^{-iax-ia^2t}$,
where $X=x+2at$ and $Y=t$, also solves the reverse-space-time NLS equation \eqref{tsnls}
with $Q(X,Y)$, i.e.
\begin{equation}\label{NLS-Qxt}
iQ_Y =Q_{XX}-2\delta Q^2Q(-X,-Y),
\end{equation}
of which $Q_0(X,Y)=q_0e^{-iax-ia^2t}=A_0e^{2i\delta A_0^2Y}$ is a solution.
}
\item{ Scaling symmetry: if $q(x,t)$ solves the reverse-space-time NLS equation \eqref{tsnls}
with a background solution $q_0(x,t)=A_0e^{2i\delta A_0^2t}$,
then $Q(X,Y)=\frac{1}{A_0}q(x,t)$,
where $X=A_0x$ and $Y=A_0^2t$, also solves the reverse-space-time NLS equation \eqref{NLS-Qxt},
of which $Q_0(X,Y)=\frac{1}{A_0}q_0=e^{2i\delta Y}$ is a solution.}
\end{itemize}

\end{proposition}

Based on these symmetries of equations \eqref{cnls}-\eqref{tsnls},
we only need to consider the unified background solution $q_0$ given in \eqref{q00}.

\subsection{Wronskian column vectors  $\phi$ and $\psi$}\label{sec-4-4}

\subsubsection{Vectors  $\phi$ and $\psi$ for the unreduced system \eqref{akns-2}}\label{sec-4-4-1}

We start with a pair of background solutions
\begin{equation}\label{q0r0}
q_0=e^{2i\delta t}, ~~r_0=\delta e^{-2i\delta t}, ~~ \delta=\pm 1
\end{equation}
of the unreduced system \eqref{akns-2}.
Note that the background solution $(q_0, r_0)$ agrees with the reductions used in
the equations \eqref{cnls}-\eqref{tsnls}.
Substituting $(q,r)=(q_0, r_0)$ into the matrix equations \eqref{akns-x} and \eqref{akns-t},
we find the following solutions of wave functions,
\begin{subequations}\label{sol-Lax pair}
\begin{align}
\Phi(\lambda,c,d)&=\delta\left[c\left(\lambda-\sqrt{\lambda^2+\delta}\right)
e^{-\sqrt{\lambda^2+\delta}(x-2i\lambda t)}
+d\left(\lambda+\sqrt{\lambda^2+\delta}\right)
e^{\sqrt{\lambda^2+\delta} (x-2i\lambda t)}\right] e^{i\delta t},\label{sol-Lax pair-Phi}\\
\Psi(\lambda,c,d)&=\left (ce^{-\sqrt{\lambda^2+\delta } (x-2i\lambda t)}\!
+\!de^{\sqrt{\lambda^2+\delta } (x-2i\lambda t)}\right)e^{-i\delta t},\label{sol-Lax pair-Psi}
\end{align}
\end{subequations}
in which $c$ and $d$ are  constants (or functions of $\lambda$).
Define
\begin{subequations}\label{PHi-PSi}
\begin{align}
& \phi= \bigl(\Phi(k_1,c_1,d_1), \Phi(k_2,c_2,d_2), \cdots, \Phi(k_{2m+2},c_{2m+2},d_{2m+2}) \bigr)^T,
\label{PHi} \\
& \psi= \bigl(\Psi(k_1,c_1,d_1), \Psi(k_2,c_2,d_2), \cdots, \Psi(k_{2m+2},c_{2m+2},d_{2m+2}) \bigr)^T.
\label{PSi}
\end{align}
\end{subequations}
Then, the quasi double Wronskians \eqref{fgh} composed by the above $\phi$ and $\psi$
provide solutions to the unreduced system \eqref{akns-2} via the transformation \eqref{trans}
where the background solutions take \eqref{q0r0}.
With regards to \eqref{PHi-PSi}, the matrix $A=\mathrm{Diag}(k_1, k_2, \cdots, k_{2m+2})$
in \eqref{pp-x} and \eqref{pp-t}.
One can also take
\begin{subequations}
\begin{align}
& \phi= \biggl(\Phi(k_1,c_1,d_1), \frac{\partial_{k_1}}{1!}\Phi(k_1,c_1,d_1), \cdots,
\frac{\partial_{k_1}^{2m+1}}{(2m+1)!}\Phi(k_1,c_1,d_1) \biggr)^T,
 \\
& \psi= \biggl(\Psi(k_1,c_1,d_1), \frac{\partial_{k_1}}{1!}\Psi(k_1,c_1,d_1), \cdots,
\frac{\partial_{k_1}^{2m+1}}{(2m+1)!}\Psi(k_1,c_1,d_1) \biggr)^T,
\end{align}
\end{subequations}
to get multiple pole solutions corresponding to $A=J_{2m+2}(k_1)$, defined as in \eqref{Jordan}.

\subsubsection{Vector  $\phi$  for the reduced equations}\label{sec-4-4-2}

To present vector $\phi$ for the reduced  equations \eqref{cnls}-\eqref{tsnls}, let us define
\begin{subequations}\label{phi-psi}
\begin{align}
& \phi=({\phi}^+,{\phi}^-)^T,~~
{\phi}^{\pm} = \bigl(\phi_{(1)}^{\pm},\phi_{(2)}^{\pm},\cdots,\phi_{(m+1)}^{\pm} \bigr)^T, \label{phi} \\
& \psi=({\psi}^+,{\psi}^-)^T,~~
{\psi}^{\pm} = \bigl(\psi_{(1)}^{\pm},\psi_{(2)}^{\pm},\cdots,\psi_{(m+1)}^{\pm} \bigr)^T, \label{psi}
\end{align}
\end{subequations}
where $\phi_{(j)}^{\pm}$ and $\psi_{(j)}^{\pm}$ are scalar functions.
Note that a general form for the constraints \eqref{T1}, \eqref{T2}, \eqref{T3} and \eqref{T4} is
\begin{equation}
\psi=TC^{\epsilon} \phi(\alpha x, \beta t),~~ \alpha, \beta=\pm 1,~ \epsilon=0,1,
\end{equation}
where $C$ stands for an operator for complex conjugation:
$C^\epsilon \phi=\phi^*$ when $\epsilon=1$ and
$C^\epsilon \phi=\phi$ when $\epsilon=0$.

There are only two types of $T$ in Sec.\ref{sec-4-2},
block skew-diagonal or block diagonal. We can present vector $\phi$ according to the type of $T$.

\noindent
{\textbf{Case 1: $T$ being block skew-diagonal:}}
In this case,
\begin{equation}\label{T-sd}
 T=\left(
     \begin{array}{cc}
       0 & I_{m+1} \\
       \gamma I_{m+1} & 0
     \end{array}
   \right), ~~~ \gamma=\pm 1,
\end{equation}
which is for equation \eqref{cnls}, \eqref{snls} and \eqref{tnls}, see Table \ref{table-5-1} and Table \ref{table-5-2}.
Vector $\phi$ takes the form
\begin{equation}\label{4.39}
\phi=(\phi^+, C^{\epsilon} \psi^{+}(\alpha x, \beta t))^T,
\end{equation}
where $(\epsilon,\alpha,\beta)$ takes $(1,1,1)$ for \eqref{cnls},
$(1,-1,1)$ for \eqref{snls} and $(0,1,-1)$ for \eqref{tnls}.
When $\mathbf{K}_{m+1}$ is diagonal as given in \eqref{diag},
we have
\begin{eqnarray}\label{diag-entry}
\phi_{(j)}^+=\Phi(k_j,c_j,d_j), & ~\psi_{(j)}^+=\Psi(k_j,c_j,d_j), ~ j=1,2,\cdots,m+1,
\end{eqnarray}
where $\Phi$ and $\Psi$ are defined in \eqref{sol-Lax pair}.
When $\mathbf{K}_{m+1}=J_{m+1}(k_1)$ is the Jordan matrix as given in \eqref{Jordan},
we take
\begin{eqnarray}\label{Jordan-entry}
\phi_{(j)}^+=\frac{\partial_{k_1}^{j-1}}{(j-1)!} \Phi(k_1,c,d), ~~
\psi_{(j)}^+=\frac{\partial_{k_1}^{j-1}}{(j-1)!} \Psi(k_1,c,d),
 ~j=1,2,\cdots,m+1.
\end{eqnarray}

\noindent{\textbf{Case 2: $T$ being block diagonal:}}
In this case,
\begin{equation}\label{T-dd}
 T=\left(
     \begin{array}{cc}
       I_{m+1} & 0 \\
       0& \gamma I_{m+1}
     \end{array}
   \right), ~~~ \gamma=\pm 1~\mathrm{or}~ \gamma=-i,
\end{equation}
which is associated with  equation \eqref{tsnls}, equation \eqref{snls} with $\delta=-1$
and equation \eqref{cnls} with $\delta=1$,
see Table \ref{table-5-2} and \eqref{4.6-real}.

To describe the relations between $c_j$ and $d_j$ in
a more succinct  and symmetric way,
we  rewrite \eqref{sol-Lax pair-Phi} and \eqref{sol-Lax pair-Psi}
as
\begin{subequations}\label{sol-Lax pair-st}
\begin{align}
&\Phi=\frac{(\hat{c}\lambda+\hat{c}\sqrt{\lambda^2+\delta}+ \hat{d})e^{\sqrt{\lambda^2+\delta} (x-2i\lambda t)}+(-\hat{c}\lambda+\hat{c}\sqrt{\lambda^2+\delta}- \hat{d})e^{-\sqrt{\lambda^2+\delta} (x-2i\lambda t)}}{2\sqrt{\lambda^2+\delta}}e^{\delta it},\label{sol-Lax pair-Phi-st}\\
&\Psi=\frac{(-\hat{d}\lambda+\hat{d}\sqrt{\lambda^2+\delta}+\delta \hat{c})e^{\sqrt{\lambda^2+\delta} (x-2i\lambda t)}+(\hat{d}\lambda+\hat{d}\sqrt{\lambda^2+\delta}-\delta \hat{c})e^{-\sqrt{\lambda^2+\delta} (x-2i\lambda t)}}{2\sqrt{\lambda^2+\delta}}e^{-\delta it},\label{sol-Lax pair-Psi-st}
\end{align}
\end{subequations}
where we have introduced $\hat{c}, \hat{d} \in \mathbb{C}$ and taken in \eqref{sol-Lax pair} that
\begin{equation*}
c=\frac{\hat{d}\lambda+\hat{d}\sqrt{\lambda^2+\delta}-\delta\hat{c}}{2\sqrt{\lambda^2+\delta}},~~
d=\frac{-\hat{d}\lambda+\hat{d}\sqrt{\lambda^2+\delta}+\delta \hat{c}}{2\sqrt{\lambda^2+\delta}}.
\end{equation*}

Then, consider diagonal case where
\begin{equation}
\mathbf{K}_{m+1}=\mathrm{Diag}(k_1, k_2, \cdots, k_{m+1}),~~
\mathbf{H}_{m+1}=\mathrm{Diag}(h_1, h_2, \cdots, h_{m+1}).
\label{diag-KH}
\end{equation}
For the reverse-space-time NLS  equation \eqref{tsnls},
we have
\begin{align}\label{entry-1}
\delta=1: &~~ \phi_{(j)}^+=\Phi(k_j,\hat{c}_j^+,i\hat{c}_j^+),~~
\phi_{(j)}^-=\Phi(h_j,\hat{c}_j^-,-i\hat{c}_j^-), \\
\delta=-1: &~~ \phi_{(j)}^+=\Phi(k_j,\hat{c}_j^+,\hat{c}_j^+),~~
\phi_{(j)}^-=\Phi(h_j,\hat{c}_j^-,-\hat{c}_j^-),  \label{entry-2}
\end{align}
where $\Phi$ is defined in \eqref{sol-Lax pair-Phi-st},
$k_j, h_j \in \mathbb{C}$, $\hat{c}_j^+$ and $\hat{c}_j^-$ can be
arbitrary complex functions of $k_j$ and $h_j$, respectively.
For the classical NLS  equation \eqref{cnls} with $\delta=1$ and the reverse-space NLS equation \eqref{snls}
with $\delta=-1$, we have
\begin{align}
&\phi_{(j)}^+=\Phi(k_j,\hat{c}_j^+,\hat{c}_j^{+*}),~~
\phi_{(j)}^-=\Phi(h_j,\hat{c}_j^-,-\hat{c}_j^{-*}), ~~
{\rm when}~T~{\rm is}~\eqref{s-nls-f-T1},\label{entry-snls+-T3}\\
&\phi_{(j)}^+=\Phi(k_j,\hat{c}_j^+,\hat{c}_j^{+*}),~~
\phi_{(j)}^-=\Phi(h_j,\hat{c}_j^-,\hat{c}_j^{-*}), ~~
{\rm when}~T~{\rm is}~\eqref{s-nls-f-T2},\label{entry-snls+-T4}
\end{align}
where $k_j, h_j\in i \mathbb{R}$ for equation \eqref{cnls} with $\delta=1$
and $k_j, h_j \in \mathbb{R}$ for equation \eqref{snls} with $\delta=-1$.

When
\[\mathbf{K}_{m+1}=J_{m+1}(k_1),~~ \mathbf{H}_{m+1}=J_{m+1}(h_1),
\]
we have
\begin{equation}
\phi_{(s)}^+=\frac{\partial_{k_1}^{s-1}}{(s-1)!} \phi_{(1)}^+, ~~
\phi_{(s)}^-=\frac{\partial_{h_1}^{s-1}}{(s-1)!} \phi_{(1)}^-,
 ~s=1,2,\cdots, m+1,
\end{equation}
where $\phi_{(1)}^{\pm}$ are defined as in \eqref{entry-1}, \eqref{entry-2},
\eqref{entry-snls+-T3} or \eqref{entry-snls+-T4},
varying with the equation considered.

\section{Examples of dynamics of solutions} \label{sec-5}

Nonzero background can bring new features for the classical and nonlocal NLS equations.
In this section we analyze some solutions and illustrate their dynamics.
The classical NLS equation \eqref{NLS-f} and reverse-space nonlocal NLS equation \eqref{snls} with $\delta=-1$
will serve as main models.

\subsection{The classical focusing NLS equation}\label{sec-5-1}

It follows from the transformation \eqref{trans} and the bilinear form  \eqref{nls-b} that
the envelope $|q|$ of the solution to the focusing NLS  equation \eqref{NLS-f} with
a background solution $q_0$ can be expressed as (also see \cite{BLM-2020,Tajiri-PRE-1998})
\begin{equation}\label{abs-q}
|q|^2
=|q_0|^2+\partial_x^2 \ln f,
\end{equation}
where
$f=|\h\phi_m;T\h\phi_m^*|$ is the quasi double Wronskian.
For the focusing NLS  equation \eqref{NLS-f},
$A=\mathrm{Diag}(\mathbf{K}_{m+1}, -\mathbf{K}^*_{m+1})$, $T$ is given by \eqref{T-sd} with $\gamma=-1$,
$\phi$ is given by \eqref{4.39} with $(\epsilon,\alpha,\beta)=(1,1,1)$.
In principle, solutions to equation \eqref{NLS-f} can be determined by the
eigenvalue structure of $\mathbf{K}_{m+1}$.
One can investigate these solutions according to the canonical form of $\mathbf{K}_{m+1}$.

\subsubsection{Breathers}\label{sec-5-1-1}

\noindent{\textbf{Case 1: $\mathbf{K}_{m+1}$ being a complex diagonal matrix}}

When $\mathbf{K}_{m+1}$ is a diagonal matrix  \eqref{diag},
following \eqref{4.39} we have $\phi=(\phi^+, \phi^-)^T$ where the entries in $\phi^{\pm}$ are
\[\phi_{(j)}^+=\Phi(k_j,c_j,d_j), ~~ \phi_{(j)}^-= \psi_{(j)}^{+*},
~~\psi_{(j)}^+=\Psi(k_j,c_j,d_j),
\]
$\Phi$ and $\Psi$ are given as \eqref{sol-Lax pair}.
When the background solution takes $q_0=e^{-2it}$, we have
\begin{subequations}\label{phij-psij}
\begin{align}
&\phi_{(j)}^+=\left(\left(-k_j+\sqrt{k_j^2-1}\right)e^{-\sqrt{k_j^2-1}(x-2ik_j t)-\xi_j^{(0)}}-\left(k_j+\sqrt{k_j^2-1}\right)
e^{\sqrt{k_j^2-1}(x-2ik_j t)+\xi_j^{(0)}}\right)e^{-it},\label{phij}\\
&\psi_{(j)}^+=\left(e^{-\sqrt{k_j^2-1}(x-2ik_j t)-\xi_j^{(0)}}
+e^{\sqrt{k_j^2-1}(x-2ik_j t)+\xi_j^{(0)}}\right)e^{it}\label{psij}
\end{align}
\end{subequations}
where we have taken $c_j=e^{-\xi_j^{(0)}}, d_j=e^{\xi_j^{(0)}}$
with $\xi_j^{(0)}$ being an arbitrary functions of $k_j$.

When $m=0$ we have from \eqref{4.1b} that
\begin{equation}
f=|\phi, T\phi^*|,
\end{equation}
where
\begin{equation}
\phi=\left(\begin{array}{c} \phi_{(1)}^+ \\ \psi_{(1)}^{+*} \end{array}\right),~~
T=\left(\begin{array}{cc} 0 & 1 \\ -1 & 0 \end{array}\right).
\end{equation}
Note that in this case we have
\[-f=\bigl|\phi_{(1)}^+\bigr|^2+\bigl|\psi_{(1)}^{+}\bigr|^2,\]
which is positive definite when  $\phi_{(1)}^+$ and $\psi_{(1)}^+$ are defined as in \eqref{phij-psij}.

By calculation we find
\begin{equation}\label{f}
f=-(A_1\cosh 2X_1(x,t)+A_2\sinh 2X_1(x,t)+A_3\cos 2X_2(x,t)-A_4\sin 2X_2(x,t)),
\end{equation}
where
\begin{align*}
&A_1=2(1+a_1^2+b_1^2+u_{11}^2+u_{12}^2),~A_2=4(a_1u_{11}+b_1u_{12}),\\
&A_3=2(1+a_1^2+b_1^2-u_{11}^2-u_{12}^2),~A_4=4(a_1u_{12}-b_1u_{11}),\\
&X_1(x,t)=u_{11} x+2B_1t+\xi_{1R}^{(0)},~B_1=a_1u_{12}+b_1u_{11},\\
&X_2(x,t)=u_{12} x+2B_2t+\xi_{1I}^{(0)},~B_2=b_1u_{12}-a_1u_{11},
\end{align*}
\[k_1 = a_1 + ib_1, ~\sqrt{k_1^2-1}=u_{11}+iu_{12}, ~\xi_1^{(0)}=\xi_{1R}^{(0)}+i\xi_{1I}^{(0)},\]
and $a_1,b_1, u_{11},u_{12},\xi_{1R}^{(0)},\xi_{1I}^{(0)}\in \mathbb{R}$.
Since $\sqrt{k_1^2-1}$ is a double-valued function of $k$, here we consider the branch
\begin{align*}
&u_{11}=\sqrt[4]{(a_1^2-b_1^2-1)^2+(2a_1b_1)^2}
\cos \left( \frac{1}{2}(\arg(a_1+1+ib_1)+\arg(a_1-1+ib_1)\right),\\
&u_{12}=\sqrt[4]{(a_1^2-b_1^2-1)^2+(2a_1b_1)^2}
\sin \left( \frac{1}{2}(\arg(a_1+1+ib_1)+\arg(a_1-1+ib_1)\right)
\end{align*}
without loss of generality.
Further we introduce
\begin{equation*}
\tan \theta=\frac{A_4}{A_3},
\end{equation*}
such that \eqref{f} is rewritten as
\begin{equation}\label{f-1}
f=-\left(A_1\cosh 2X_1(x,t)+A_2\sinh 2X_1(x,t)+\sqrt{A_3^2+A_4^2}\cos (2X_2(x,t)+\theta)\right).
\end{equation}
Noticing that $A_1>|A_2|>0$ for all $k_1\neq 0$,
from the above expression and \eqref{abs-q},
$|q|^2$ behaves like a wave traveling along the line $X_1=0$ and oscillating periodically
with a period determined by $2X_2+\theta=2j\pi$, $j\in \mathbb{Z}$.
Note that the case $a_1=0$ or $a_1=\pm1,b_1=0$ yields $|q|^2=1$, which is trivial and we do not consider.

To see more details, we rewrite  \eqref{f-1} in terms of the following coordinates,
\begin{equation}\label{x-z-coor}
\left(x,~z=t+\frac{u_{11}}{2B_1}x+\frac{\xi_{1R}^{(0)}}{2B_1}\right),
\end{equation}
which gives rise to
\begin{equation}\label{f-2}
\begin{aligned}
f=& -\Biggl\{A_1\cosh(4B_1z+2\xi_{1R}^{(0)})+A_2\sinh(4B_1z+2\xi_{1R}^{(0)})  \\
&\left.+\sqrt{A_3^2+A_4^2}\cos \left(4B_2\left(z+\frac{a_1(u_{11}^2+u_{12}^2)}{2B_1B_2}x
+\frac{\xi_{1I}^{(0)}}{2B_2}
-\frac{\xi_{1R}^{(0)}}{2B_1}\right)
+\theta \right)\right\}.
\end{aligned}
\end{equation}
In terms of \eqref{x-z-coor} we can see that \eqref{abs-q} with \eqref{f-2} provides a breather
traveling along the straight line $z=\mathrm{constant}$  and oscillating with a period  with respect to $x$,
\begin{equation}
P=\left|\frac{2\pi B_1}{a_1(u_{11}^2+u_{12}^2)}\right|.
\end{equation}
An illustration is given in Fig.\ref{fig-1}(a), which describes a moving breather.
Such a breather is also known as the Tajiri-Watanabe breather (see Fig.4 in \cite{Tajiri-PRE-1998}).
In 1998 Tajiri and Watanabe derived and studied breathers of the focusing NLS equation
using Hirota's bilinear method \cite{Tajiri-PRE-1998}.

Back to the expression \eqref{f-1}.
Stationary breathers appear when $b_1=0$.
More precisely, when $|a_1| > 1$ and $ b_1=0$, which leads to $u_{11}=\sqrt{a_1^2 -1}$ and $ u_{12}=0$,
we have $B_1=0$ and then $X_1(x,t)=u_{11} x+ \xi_{1R}^{(0)}$. In this case we can have a
breather stationary with respect to $x$, where
\begin{equation}\label{Ma-breather}
\begin{aligned}
f=&-\left(2a_1^2\cosh2\left(\sqrt{a_1^2 -1}x+\xi_{1R}^{(0)}\right)+4a_1\sqrt{a_1^2 -1}\sinh2\left(\sqrt{a_1^2 -1}x
+\xi_{1R}^{(0)}\right)\right.\\
&\left.+4\cos2\left(-2a_1\sqrt{a_1^2 -1}t+\xi_{1I}^{(0)}\right)\right).
\end{aligned}
\end{equation}
It follows that a stationary breather  oscillating in time with period $P_t=\frac{\pi}{|2a_1\sqrt{a_1^2 -1}|}$,
which is known as the Kuznetsov-Ma breather \cite{Kuz-SP-1977,Ma-SAPM-1979}.
It is described in Fig.\ref{fig-1}(b).
In another  case where $|a_1| < 1$ and $b_1=0$, which leads to $u_{11}=0$ and $u_{12}=\sqrt{1-a_1^2}$,
from  \eqref{f-1}  we have
\begin{equation}\label{Akhmediev-breather}
f=-\left(2\cosh2\left(2a_1\sqrt{1-a_1^2}t+\xi_{1R}^{(0)}\right)
+4a_1^2\cos\left(2(\sqrt{1-a_1^2}x+\xi_{1I}^{(0)})+\theta\right)\right)
\end{equation}
with
$\tan \theta=\frac{\sqrt{1-a_1^2}}{2a_1}$.
This will gives rise to a breather traveling along the line $t=-\frac{\xi_{1R}^{(0)}}{2a_1\sqrt{1-a_1^2}}$
and being periodic with respect to $x$ with the period $P_x=\frac{\pi}{\sqrt{1-a_1^2}}$.
Such a breather is known as the Akhmediev breather \cite{Akhmediev breather},
which was first studied by Akhmediev in \cite{Akhmediev breather} and then bear his name.
Stability of the Akhmediev and Kuznetsov-Ma breathers was studied recently \cite{Grinevich-Nonlinearity-2021,Haragus-JNS-2022}.
The Akhmediev breather is perpendicular to the Kuznetsov-Ma breather,
as  depicted in Fig.\ref{fig-1} (b) and (c).


\captionsetup[figure]{labelfont={bf},name={Fig.},labelsep=period}
\begin{figure}[htbp]
\centering
\subfigure[ ]{
\begin{minipage}[t]{0.32\linewidth}
\centering
\includegraphics[width=2.0in]{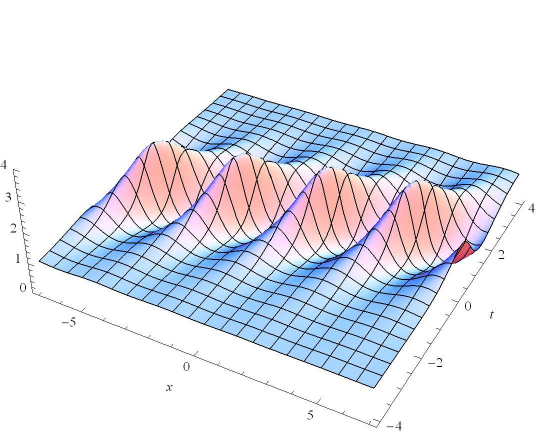}
\end{minipage}%
}%
\subfigure[ ]{
\begin{minipage}[t]{0.32\linewidth}
\centering
\includegraphics[width=2.0in]{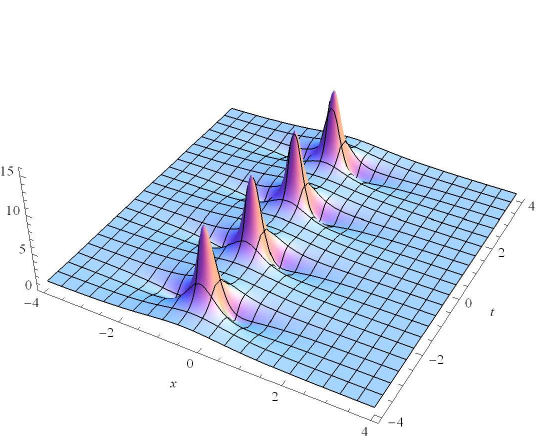}
\end{minipage}%
}%
\subfigure[ ]{
\begin{minipage}[t]{0.32\linewidth}
\centering
\includegraphics[width=2.0in]{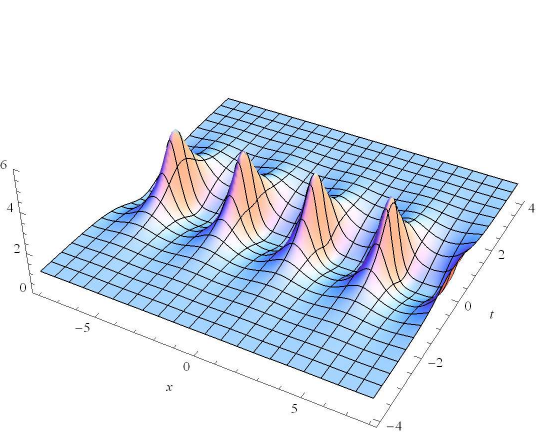}
\end{minipage}%
}
\caption{ Shape and motion of one breather solution of the focusing NLS equation \eqref{NLS-f}
with a background solution $q_0=e^{-2it}$.
The envelope is given by \eqref{abs-q} where $f$ is \eqref{f}.
(a) 3D-plot for a moving breather associated with \eqref{f-2}
for $a_1=0.3,b_1=-0.3,\xi_{1R}^{(0)}=\xi_{1I}^{(0)}=0$.
~(b) 3D-plot for the Kuznetsov-Ma breather associated with \eqref{Ma-breather}
for $a_1=1.2, \xi_{1R}^{(0)}=\xi_{1I}^{(0)}=0$.~
(c) 3D-plot for the Akhmediev breather associated with \eqref{Akhmediev-breather}
for $a_1=0.5, \xi_{1R}^{(0)}=\xi_{1I}^{(0)}=0$.
}
\label{fig-1}
\end{figure}


The envelope of two-breather solution can be obtained via \eqref{abs-q}
by taking $m=1$ in quasi double Wronskians \eqref{4.1b}, i.e.
\begin{subequations}\label{f-2bre}
\begin{equation}\label{2-sol}
 f=|\phi, \phi_1; \psi, \psi_1|
\end{equation}
with
\begin{equation}\label{2-sol-entry}
\phi=\left(\phi_{(1)}^+, \phi_{(2)}^+; \psi_{(1)}^{+*},  \psi_{(2)}^{+^*}\right)^T,~~
\psi=\left(\psi_{(1)}^+,\psi_{(2)}^+; -\phi_{(1)}^{+*}, -\phi_{(2)}^{+*}\right)^T,
\end{equation}
\end{subequations}
in which $\phi_{(j)}^+$ and $\psi_{(j)}^+$ are defined as in \eqref{phij-psij}.
There are various types of two-breather interactions.
As examples Fig.\ref{fig-2} illustrates interactions between two Tajiri-Watanabe breathers,
interaction of the Akhmediev breather and Kuznetsov-Ma breather and interaction of two
Akhmediev breathers, in Fig.\ref{fig-2} (a), (b) and (c), respectively.


\captionsetup[figure]{labelfont={bf},name={Fig.},labelsep=period}
\begin{figure}[ht]
\centering
\subfigure[ ]{
\begin{minipage}[t]{0.32\linewidth}
\centering
\includegraphics[width=2.0in]{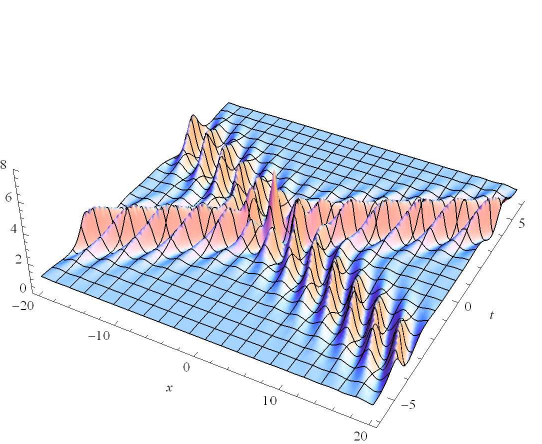}
\end{minipage}%
}%
\subfigure[ ]{
\begin{minipage}[t]{0.32\linewidth}
\centering
\includegraphics[width=2.0in]{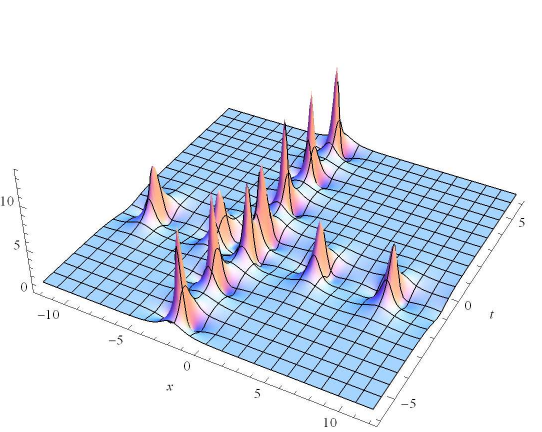}
\end{minipage}%
}%
\subfigure[ ]{
\begin{minipage}[t]{0.32\linewidth}
\centering
\includegraphics[width=2.0in]{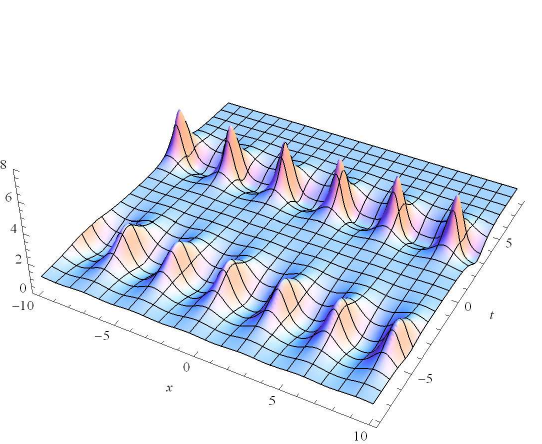}
\end{minipage}%
}
\caption{ Shape and motion of  two-breather interactions of the focusing NLS equation \eqref{NLS-f}
with a background solution $q_0=e^{-2it}$.
The envelope is given by \eqref{abs-q} where $f$ is \eqref{f-2bre}.
(a) 3D-plot for $a_1=0.3,b_1=0.5, a_2=0.3,b_2=-0.5,\xi_{1R}^{(0)}=\xi_{1I}^{(0)}=\xi_{2R}^{(0)}
=\xi_{2I}^{(0)}=0$.~
(b) 3D-plot for $a_1=1.2,b_1=0, a_2=0.8, b_2=0,\xi_{1R}^{(0)}=1,\xi_{1I}^{(0)}=0,\xi_{2R}^{(0)}=1,
\xi_{2I}^{(0)}=0$.~
(c) 3D-plot for $a_1=0.3,b_1=0,a_2=0.5,b_2=0,\xi_{1R}^{(0)}=2,\xi_{1I}^{(0)}=0,\xi_{2R}^{(0)}=-2,
\xi_{2I}^{(0)}=0$.~
}
\label{fig-2}
\end{figure}


\vskip 4pt

\noindent{\textbf{Case 2: $\mathbf{K}_{m+1}$ being a Jordan matrix}}

Let $\phi_{(1)}^+$ and $\psi_{(1)}^+$ be defined as in \eqref{phij-psij},
and we define
\begin{align*}
& \phi_{(j)}^+=\frac{\partial_{k_1}^{j-1}}{(j-1)!} \phi_{(1)}^+, ~~
 \psi_{(j)}^+=\frac{\partial_{k_1}^{j-1}}{(j-1)!} \psi_{(1)}^+, \\
& \phi_{(j)}^-=\psi_{(j)}^{+*}, ~~ \psi_{(j)}^-=- \phi_{(j)}^{+*}, ~~~ j=1,2,\cdots,m+1.
\end{align*}
The corresponding $f$ composed by the above elements yields breathers
when  $\mathbf{K}_{m+1}$ is the Jordan matrix $J_{m+1}(k_1)$  as given in \eqref{Jordan}.
For the simplest Jordan block solution of the focusing NLS equation \eqref{NLS-f}
with the background solution $q_0=e^{-2it}$, we have $m=1$ and $f$ composed by
\begin{equation}
\phi=\left(\phi_{(1)}^+, \partial_{k_1}\phi_{(1)}^+;  \psi_{(1)}^{+*},
\left(\partial_{k_1}\psi_{(1)}^+\right)^*\right)^T,
~~
\psi=\left(\psi_{(1)}^+,\partial_{k_1}\psi_{(1)}^+; - \phi_{(1)}^{+*},
-\left(\partial_{k_1} \phi_{(1)}^+\right)^*\right)^T.
\label{Jordan-phipsi}
\end{equation}
Such a breather is described in Fig.\ref{fig-3}.


\captionsetup[figure]{labelfont={bf},name={Fig.},labelsep=period}
\begin{figure}[ht]
\centering
\subfigure[ ]{
\begin{minipage}[t]{0.32\linewidth}
\centering
\includegraphics[width=2.0in]{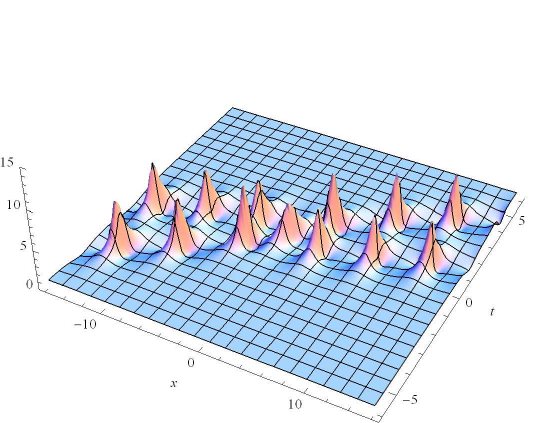}
\end{minipage}%
}%
\subfigure[ ]{
\begin{minipage}[t]{0.32\linewidth}
\centering
\includegraphics[width=2.0in]{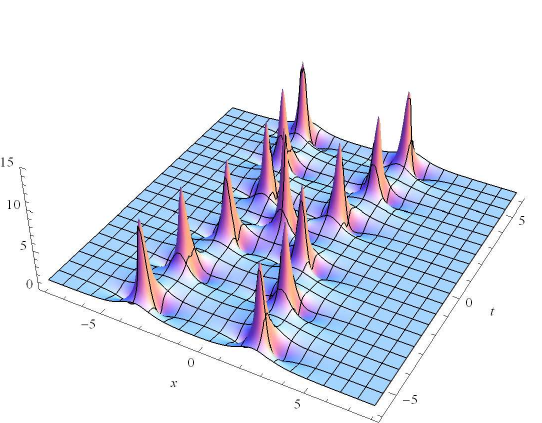}
\end{minipage}%
}%
\caption{ Shape and motion of Jordan block solution of the focusing NLS equation
\eqref{NLS-f} with a background solution $q_0=e^{-2it}$.
The envelope is given by \eqref{abs-q} where $f$ is composed by
\eqref{Jordan-phipsi}.
(a) 3D-plot for $a_1=0.8, b_1=-0.15, \xi_{1R}^{(0)}=\xi_{1I}^{(0)}=0$.~
(b) 3D-plot for $a_1=1.2, b_1=0, \xi_{1R}^{(0)}=\xi_{1I}^{(0)}=0$.
}
\label{fig-3}
\end{figure}


\subsubsection{Rational solutions and rogue waves}

Rational solutions can be obtained as a limit case of breathers when taking $k_j\rightarrow 1$.
This can be seen from the expression \eqref{phij-psij}.
Since the Akhmediev breathers and Kuznetsov-Ma breathers are generated when $b_j=0$,
rational solutions can also be understood as a limit of these two types of breathers.
In principle, for getting  rational solutions, in $A$ we should take $\mathbf{K}_{m+1}=J_{m+1}(1)$,
but the limit procedure needs to be elaborated.

Let us consider \eqref{sol-Lax pair} and rewrite them in the form
\begin{subequations}\label{ration-phipsi}
\begin{align}
&\Phi(\kappa,c,d)=\left(c(\kappa)\left(-\sqrt{\kappa^2+1}+\kappa\right)e^{-\kappa(x-2i\sqrt{\kappa^2+1} t)}
-d(\kappa)\left(\sqrt{\kappa^2+1}+\kappa\right)
e^{\kappa(x-2i\sqrt{\kappa^2+1} t)}\right)e^{-it},\label{ration-phi}\\
&\Psi(\kappa,c,d)=\left(c(\kappa) e^{-\kappa(x-2i\sqrt{\kappa^2+1} t)-\xi_j^{(0)}}
+ d(\kappa)e^{\kappa(x-2i\sqrt{\kappa^2+1} t)+\xi_j^{(0)}}\right)e^{it}\label{ration-psi},
\end{align}
\end{subequations}
where we have taken $\delta=-1$ and introduce $\kappa=\sqrt{k^2-1}$, and
we take $c(\kappa)$ and $d(\kappa)$ to be functions of $\kappa$.
Impose constraint $c(\kappa)=-d(-\kappa)$ and take formal expressions
\begin{equation}\label{cd-kappa}
c(\kappa)=\sum_{j=0}^{\infty}s_j\kappa^j,~d(\kappa)=\sum_{j=0}^{\infty}(-1)^{j+1}s_j\kappa^j,
\end{equation}
where $s_j$ are arbitrary complex parameters.
We denote the above $\Phi(\kappa,c,d)$ and $\Psi(\kappa,c,d)$ with \eqref{cd-kappa}
by $\Phi_{odd}$ and $\Psi_{odd}$ respectively.
Expand them in terms of $\kappa$ as
\begin{equation}\label{Taylor-nls}
\Phi_{odd}=\sum\limits_{j=0}^{\infty}R_{2j+1}\kappa^{2j+1},~~
\Psi_{odd}=\sum\limits_{j=0}^{\infty}S_{2j+1}\kappa^{2j+1},
\end{equation}
in which
\begin{gather}
R_{2j+1}={\frac{1}{(2j+1)!}\frac{\partial^{2j+1}}{\partial \kappa^{2j+1}}\Phi_{odd}}|_{\kappa=0},~~
S_{2j+1}={\frac{1}{(2j+1)!}\frac{\partial^{2j+1}}{\partial \kappa^{2j+1}}\Psi_{odd}}|_{\kappa=0},
~~j=0,1,2,\cdots.
\end{gather}
Define
\begin{subequations}\label{ration-entry-odd}
\begin{align}
& \phi^{odd}=(R_1,R_3,\cdots,R_{2m+1},S_1^*,S_3^*,\cdots,S_{2m+1}^*)^T, \label{phi-odd}\\
& \psi^{odd}=(S_1,S_3,\cdots,S_{2m+1},-R_1^*,-R_3^*,\cdots,-R_{2m+1}^*)^T=T(\phi^{odd})^*,
\end{align}
\end{subequations}
where $T$ takes the form \eqref{T-sd} with $\gamma=-1$.
It can be verified that $\phi^{odd}$ satisfies equation \eqref{4.2}
where $q_0=e^{-2it}$, $\delta=-1$, $T$ is given by \eqref{T-sd} with $\gamma=-1$,
and $A=\mathrm{Diag}(\mathbf{K}_{m+1},-\mathbf{K}_{m+1}^*)$ with
\begin{equation}\label{K}
\mathbf{K}_{m+1}=\left(
           \begin{array}{ccccc}
             \alpha_0     &        0         &       0       &       \cdots       &   0    \\
             \alpha_2     &   \alpha_0        &    0           &     \cdots         &  0   \\
             \alpha_4     &   \alpha_2        &   \alpha_0     &          \cdots         &  0    \\
             \vdots       &   \vdots          &   \ddots       &  \ddots      &    \vdots   \\
             \alpha_{2m}  &   \alpha_{2m-2}   &   \cdots       &  \alpha_2    &   \alpha_0
           \end{array}
         \right),
\end{equation}
in which
$\alpha_{2j}=\frac{1}{(2j)!}\partial_\kappa^{2j}\sqrt{\kappa^2+1}|_{\kappa=0},~(j=0,1,2,\cdots)$.
Note also that $A$ and $T$ satisfy \eqref{cnls-AT} with $\delta=-1$.
Thus, the  quasi double Wronskians
\begin{equation}\label{f-rat-m}
f=|\h\phi_m^{odd};\h\psi_m^{odd}|, ~~ g=2|\h\phi_{m+1}^{odd};\h\psi_{m-1}^{odd}|
\end{equation}
provide rational solutions to the  focusing NLS equation \eqref{NLS-f} via \eqref{trans}
and the envelope via \eqref{abs-q}.

The first order rational solution (for $m = 0$) is
\begin{equation}\label{1-ration-sol}
q=-\left(1+\frac{4it-1}{\widetilde{x}^2+4t^2+\frac{1}{4}}\right) e^{-2it},
\end{equation}
where $\widetilde{x}=x+\frac{s_0-2s_1}{2s_0}$ with $s_0,s_1$ being coefficients of $c(k)$.
Here we take $s_0,s_1\in \mathbb{R}$ for simplicity.
We refer to it as the Peregrine soliton since it was first derived by Peregrine in \cite{Per-JAMS-1983}.
Its envelope $|q|$ is localized in both space and time. It is also known as
a rogue wave of the focusing NLS equation.
The maximum value of $|q|$ is 3, occurring  at $(x,t)=(-\frac{s_0-2s_1}{2s_0},0)$,
three times hight of the  background $|q_0|=1$.
The envelope is depicted in Fig.\ref{fig-4}(a).

The general second order rational solution can be obtained from
\begin{subequations}\label{2-rat}
\begin{equation}\label{2-sol-r}
q=q_0+\frac{g}{f},~~g=2|\phi^{odd}, \phi^{odd}_1, \phi^{odd}_2; \psi^{odd}_1|,~~
f=|\phi^{odd}, \phi^{odd}_1; \psi^{odd}, \psi^{odd}_1|,
\end{equation}
where
\begin{equation}\label{entry}
\phi^{odd}=\left(R_1, R_2; S_1^*, S_2^*\right)^T,~~
\psi^{odd}=\left(S_1, S_2; -R_1^*, -R_2^*\right)^T.
\end{equation}
\end{subequations}
We skip explicit expression of $q$.
The envelope of a typical second order rational solution is shown in Fig.\ref{fig-4}(b) with a symmetric shape
and having a single maximum 5.
In general, the maximum amplitude of a $n$th-order rogue wave with one central main peak is $2n+1$
times of the height of the amplitude of the background plane wave \cite{Akhmediev-PRE-2009,Hejs-PLA-2017},
(also see \cite{Hejs-PLA-2017} where rogue wave with such pattern is called a ``fundamental rogue wave'').
The envelope of another typical second order rational solution has three peaks, as shown in Fig.\ref{fig-4}(c).


\captionsetup[figure]{labelfont={bf},name={Fig.},labelsep=period}
\begin{figure}[ht]
\centering
\subfigure[ ]{
\begin{minipage}[t]{0.32\linewidth}
\centering
\includegraphics[width=2.0in]{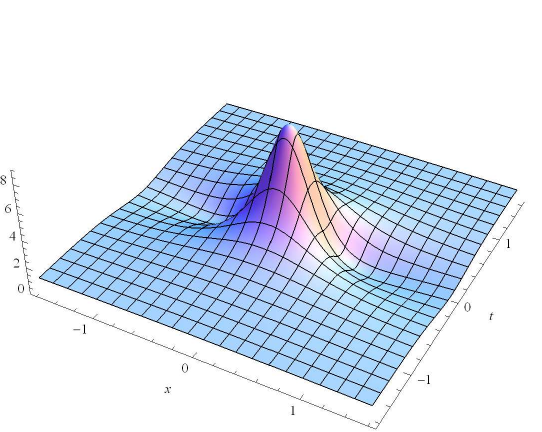}
\end{minipage}%
}%
\subfigure[ ]{
\begin{minipage}[t]{0.32\linewidth}
\centering
\includegraphics[width=2.0in]{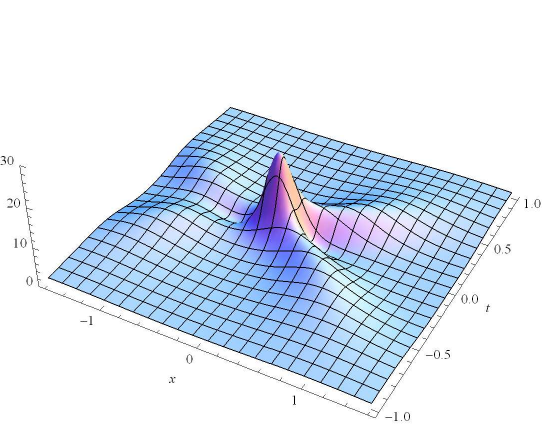}
\end{minipage}%
}%
\subfigure[ ]{
\begin{minipage}[t]{0.38\linewidth}
\centering
\includegraphics[width=2.0in]{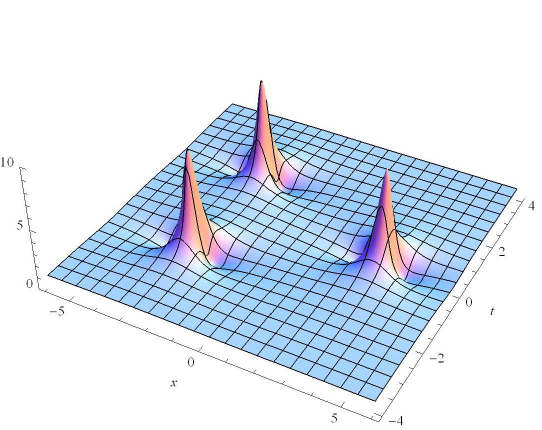}
\end{minipage}%
}
\caption{ Shape and motion of  rational solutions of the focusing NLS equation \eqref{NLS-f}.
(a) Envelope of the first order   rational solution given by \eqref{1-ration-sol} with $s_0=1, s_1=0.5$.~
(b) Envelope of the second order rational solution given by \eqref{2-rat}
with $s_0=2, s_1=0.5$.~
(c)  Envelope of the second order  rational solution given by \eqref{2-rat} with $s_0=1, s_1=0, s_2=10, s_3=-20$.
}
\label{fig-4}
\end{figure}


The third order rational solution is obtained by taking $m=2$ in \eqref{f-rat-m}.
Without presenting formulae, we depict
some different patterns of the envelope of these solutions in Fig.\ref{fig-5}.
Fig.\ref{fig-5} (a) shows the pattern where there is only one  central main peak,
Fig.\ref{fig-5} (d) and Fig.\ref{fig-5} (e) show the pattern
consisting basically of 6 well-separated fundamental part on a unit background,
which are located on a triangle and a  pentagon, respectively.
Another two interesting patterns  are shown in  Fig.\ref{fig-5} (b) and Fig.\ref{fig-5} (c).
Thus, it indicates that higher-order rogue waves contain richer structures.
Note that recently it was found the pattern
of rogue waves is related to the roots of Yablonskii-Vorob'ev polynomials \cite{YY-PD-2021a,YY-PD-2021b}.


\captionsetup[figure]{labelfont={bf},name={Fig.},labelsep=period}
\begin{figure}[ht]
\centering
\subfigure[ ]{
\begin{minipage}[t]{0.32\linewidth}
\centering
\includegraphics[width=2.0in]{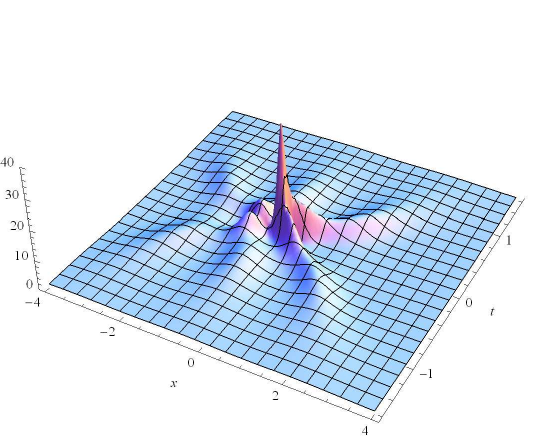}
\end{minipage}%
}%
\subfigure[ ]{
\begin{minipage}[t]{0.32\linewidth}
\centering
\includegraphics[width=2.0in]{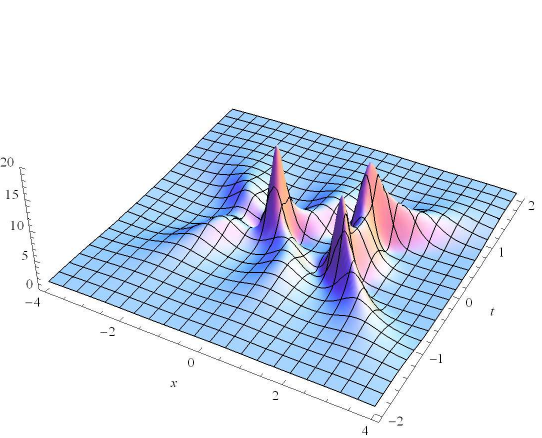}
\end{minipage}%
}%
\subfigure[ ]{
\begin{minipage}[t]{0.38\linewidth}
\centering
\includegraphics[width=2.0in]{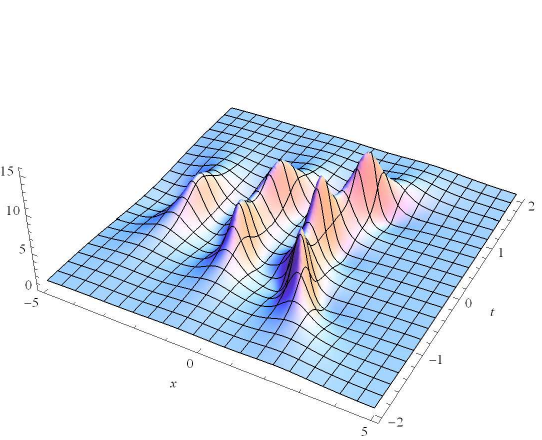}
\end{minipage}%
}
\\
\subfigure[ ]{
\begin{minipage}[t]{0.38\linewidth}
\centering
\includegraphics[width=2.0in]{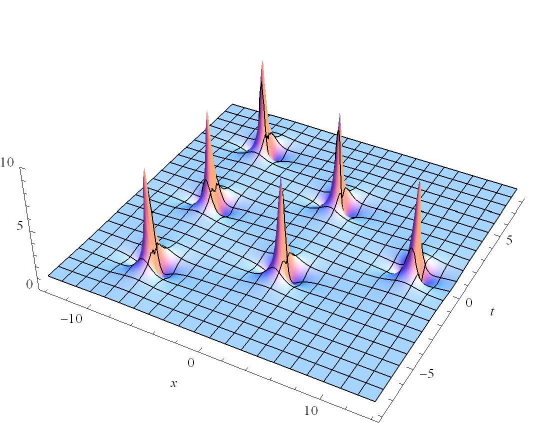}
\end{minipage}%
}%
\subfigure[ ]{
\begin{minipage}[t]{0.38\linewidth}
\centering
\includegraphics[width=2.0in]{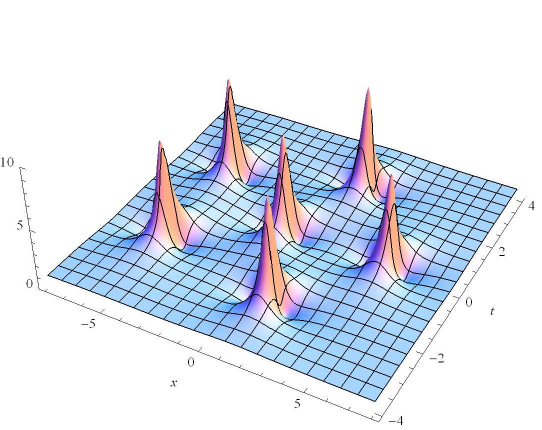}
\end{minipage}%
}
\caption{Shape and motion of the envelope of the third order rational solution of the focusing NLS equation \eqref{NLS-f}.
(a) 3D plot for $s_0=1, s_1=0, s_2=0, s_3=0, s_4=0, s_5=0$.~
(b) 3D plot for  $s_0=1, s_1=1, s_2=0, s_3=0, s_4=0, s_5=0$.~
(c) 3D plot for $s_0=1, s_1=0, s_2=0, s_3=1, s_4=0, s_5=0$.~
(d) 3D plot for $s_0=-1, s_1=0, s_2=0, s_3=100, s_4=-200, s_5=0$.~
(e) 3D plot for $s_0=1, s_1=0, s_2=100, s_3=1, s_4=0, s_5=200$.
}
\label{fig-5}
\end{figure}


Apart from \eqref{ration-entry-odd}, one can also introduce Wronskian entries by imposing
$c(\kappa)=d(-\kappa)$, i.e.,
\begin{equation}
c(\kappa)=\sum_{j=0}^{\infty}s_j\kappa^j,~d(\kappa)=\sum_{j=0}^{\infty}(-1)^{j}s_j\kappa^j,
\end{equation}
such that $\Phi(\kappa,c,d)$ and $\Psi(\kappa,c,d)$ given in \eqref{ration-phipsi}
(denoted by $\Phi_{evev}$ and $\Psi_{even}$, respectively)
can be expanded as
\begin{equation}
\Phi_{even}=\sum\limits_{j=0}^{\infty}R_{2j}\kappa^{2j}, ~~
\Psi_{even}=\sum\limits_{j=0}^{\infty}S_{2j}\kappa^{2j},
\end{equation}
where
\begin{gather}
R_{2j}={\frac{1}{(2j)!}\frac{\partial^{2j}}{\partial \kappa^{2j}}\Phi_{even}}|_{\kappa=0},~
S_{2j}={\frac{1}{(2j)!}\frac{\partial^{2j}}{\partial \kappa^{2j}}\Psi_{even}}|_{\kappa=0},~~j=0,1,2,\cdots.
\end{gather}
Then the vectors for the quasi double Wronskian are taken as
\begin{subequations}\label{ration-entry-even}
\begin{align}
& \phi^{even}=(R_0,R_2,\cdots,R_{2m},S_0^*,S_2^*,\cdots,S_{2m}^*)^T, \\
& \psi^{even}=(S_0,S_2,\cdots,S_{2m},-R_0^*,-R_2^*,\cdots,-R_{2m}^*)^T=T(\phi^{even})^*,
\end{align}
\end{subequations}
where $T$ is given by \eqref{T-sd} with $\gamma=-1$.
In this case $m=0$ does not lead to a nontrivial solution but
the solutions obtained by taking $m=1$ and $m=2$ correspond to
\eqref{1-ration-sol} and \eqref{2-sol-r},
which are the first order and second order rational solutions derived using $\phi^{odd}$ and $\psi^{odd}$.

One may conjecture that the $m$-th order rational solution derived using $\phi^{odd}$ and $\psi^{odd}$
corresponds to the $(m+1)$-th order rational solution derived using $\phi^{even}$ and $\psi^{even}$.
Similar connection is proved in the rational solutions of  the discrete KdV-type  equations,
see \cite{ZhangDD-SIGMA-2017}.
We also note that the parameters $\{s_j\}$ (or $c(\kappa)$) play the same roles as
the lower triangular Toeplitz matrices, cf.\cite{DJZ-2006,ZhangDJ-RMP-2014}.
An $(m+1)$-th order  lower triangular Toeplitz matrix  $\mathbf{P}_{m+1}$ is defined as
\begin{equation}\label{Toe}
\mathbf{P}_{m+1}=\mathcal{T}_{m+1}[t_j]_0^m=\left(
           \begin{array}{ccccc}
             t_0     &        0         &       0       &       \cdots       &   0    \\
             t_1     &   t_0        &    0           &     \cdots         &  0   \\
             t_2     &   t_1        &   t_0     &          \cdots         &  0    \\
             \vdots       &   \vdots          &   \ddots       &  \ddots      &    \vdots   \\
             t_{m}  &   t_{m-1}   &   \cdots       &  t_1    &   t_0
           \end{array}
         \right),~~ t_j\in \mathbb{C},
\end{equation}
which commutes with $\mathbf{K}_{m+1}$ defied in \eqref{K}.
For the block diagonal matrix $Q=\mathrm{Diag}(\mathbf{P}_{m+1},\mathbf{P}_{m+1}^*)$,
when $T$ is given by \eqref{T-sd} with $\gamma=-1$,
and $A=\mathrm{Diag}(\mathbf{K}_{m+1},-\mathbf{K}_{m+1}^*)$ with \eqref{K},
we have
\[AQ=QA, ~ QT=TQ^*.\]
This indicates that for any $\phi$ that satisfies \eqref{4.2} with the above $A$ and $T$,
$\widetilde\phi=Q\phi$ is also a solution of \eqref{4.2}.
Moreover, if $\phi^{odd}$ in \eqref{phi-odd} is derived with $c(\kappa)=1$,
then in $\widetilde\phi=Q\phi$, the parameters $\{t_j\}$ play the exactly same roles as $\{s_j\}$.
In \cite{DJZ-2006}, for the KdV equation, the relation between $\mathbf{P}_{m+1}$
and $c(\kappa)$ is described, see Sec.2 of \cite{DJZ-2006}.

\subsection{The defocusing reverse-space nonlocal NLS equation}\label{sec-5-2}

In this section we investigate solutions of the defocusing reverse-space nonlocal NLS equation
\begin{equation}\label{s-NLS-df}
iq_t=q_{xx}-2 q^2q^*(-x)
\end{equation}
with the background solution $q_0=e^{2it}$.
This is the equation \eqref{snls} with $\delta=1$.
Note that the reverse-space nonlocal NLS equation \eqref{snls}
is considered as a model with  parity-time symmetry (see \cite{AblM-PRL-2013}).
Efforts of finding physical applications of NLS type nonlocal integrable systems
can also be found in \cite{AblM-JPA-2019,Lou-CTP-2020,YJK-PRE-2018}, etc.

\subsubsection{Solitons and doubly periodic solutions}\label{sec-5-2-1}

Solution  to   equation \eqref{s-NLS-df} with a background solution $q_0$ is written as
\begin{equation}\label{DW}
q=q_0+\frac{g}{f}, ~~f=|\h\phi_m;\h\psi_m|, ~~ g=2|\h\phi_{m+1};\h\psi_{m-1}|,
\end{equation}
where we take $q_0=e^{2i t}$.
Consider the simplest case, $m=0$. From the results in Table \ref{table-5-1} and in Sec.\ref{sec-4-4-2},
we have
\begin{subequations}\label{fg-2ss}
\begin{equation}\label{fg-snls}
f=\left | \begin{smallmatrix}
            \phi_{(1)}^+ &~ \psi_{(1)}^+\\
            \psi_{(1)}^{+*}(-x) &~ -\phi_{(1)}^{+*}(-x)
            \end{smallmatrix}
            \right|,~~
g= 2 \left | \begin{smallmatrix}
            \phi_{(1)}^+ &~ k_1\phi_{(1)}^+\\
            \psi_{(1)}^{+*}(-x) &~ k_1^*\psi_{(1)}^{+*}(-x)
            \end{smallmatrix}
    \right|,
\end{equation}
where
\begin{align}
&\phi_{(1)}^+=\left (\alpha_1 e^{\sqrt{k^2_1+1} (x-2ik_1 t)}
+\beta_1 e^{-\sqrt{k^2_1+1} (x-2ik_1 t)}\right)e^{it},\label{5.30b}\\
&\psi_{(1)}^+=\left(\alpha_1 \left(\sqrt{k^2_1+1}-k_1\right)e^{\sqrt{k^2_1+1} (x-2ik_1 t)}
-\beta_1 \left(\sqrt{k^2_1+1}+k_1\right)
e^{-\sqrt{k^2_1+1} (x-2ik_1 t)}\right)e^{-it},\label{5.30c}
\end{align}
\end{subequations}
and in $\Phi(k,c,d)$ and $\Psi(k,c,d)$ defined in \eqref{sol-Lax pair} we have taken $\delta=1$,
\[c(k)=-\beta_1\left(\sqrt{k_1^2+1}+k_1\right), ~~ d(k)=\alpha_1\left(\sqrt{k_1^2+1}-k_1\right)\]
with $\alpha_1$ and $\beta_1$ as arbitrary functions of $k$.

The envelope $|q|$  of some solutions resulting from \eqref{fg-2ss} is depicted in Fig.\ref{fig-6},
which exhibits features of two-soliton interactions, although the solution is from the simplest case, $m=0$.
In the following we implement asymptotic analysis so as to understand such features.
To avoid singular and trivial solutions, we consider the special case where $k_1=ib$, $b\in \mathbb{R}$.
It turns out that the solution can be classified according to the sign of $1-b^2$.

\noindent{\textbf{Case 1: $|b|<1$}}

We write   solution $q$ in terms of the following coordinates,
\begin{equation*}
\left(X_1=x+2bt,~t\right).
\end{equation*}
This gives rise to
\[q=\frac{G}{F},\]
where
\begin{align*}
G=& e^{2it}\left\{1+2ib\left[(\sqrt{1-b^2}+ib)e^{4b\sqrt{1-b^2}t}
-\beta \beta^*(\sqrt{1-b^2}-ib)e^{-4b\sqrt{1-b^2}t}\right.\right. \\
&\left.\left. +\beta(\sqrt{1-b^2}+ib)e^{-2\sqrt{1-b^2}X_1+4b\sqrt{1-b^2}t}
-\beta^*(\sqrt{1-b^2}-ib)e^{2\sqrt{1-b^2}X_1-4b\sqrt{1-b^2}t}\right]\right\}, \\
F=& e^{4b\sqrt{1-b^2}t}+\beta \beta^*e^{-4b\sqrt{1-b^2}t}
+\beta b(b-i\sqrt{1-b^2})e^{-2\sqrt{1-b^2}X_1+4b\sqrt{1-b^2}t}  \\
& +\beta^* b(b+i\sqrt{1-b^2})e^{2\sqrt{1-b^2}X_1-4b\sqrt{1-b^2}t} ,
\end{align*}
and we have taken $\beta=\frac{\beta_1}{\alpha_1}$.
When keeping $X_1$ to be constant, we find
\begin{equation*}
|q|^2\rightarrow 1-\frac{2\sqrt{1-b^2}\beta_I}{\mathrm{sgn}[b]\left(|\beta|
\cosh(2\sqrt{1-b^2}X_1 - \ln(|\beta\, b^{\pm 1}|))
+\mathrm{sgn}[b](\sqrt{1-b^2}\beta_I+b\beta_R)\right)},
~bt\rightarrow \pm\infty.
\end{equation*}
Similarly, in terms of the coordinate
\begin{equation*}
\left(X_2=x-2bt,~t\right),
\end{equation*}
we get
\begin{equation*}
|q|^2\rightarrow 1+\frac{2((2b^2-1)\beta_I-2b\sqrt{1-b^2}\beta_R)\sqrt{1-b^2}}
{\mathrm{sgn}[b]\left(|\beta|\cosh(2\sqrt{1-b^2}X_2+ \ln(|\beta\, b^{\pm 1}|))
+\mathrm{sgn}[b](\sqrt{1-b^2}\beta_I+b\beta_R)\right)},~ bt\rightarrow \pm\infty.
\end{equation*}
Here we have taken $\beta=\beta_R+i\beta_I, \beta_R,\beta_I\in \mathbb{R}$.
Note that here we do not have formula \eqref{abs-q}, which indicates the background $|q_0|$ for the envelope $|q|$
in the classical case, however,
since the background solution is $q_0=e^{2it}$ which yields $|q_0|=1$,
the above asymptotic results indicate that $|q|$ has a background plane $|q|=1$, which is equal to  $|q_0|$.

For convenience let us call the above two solitons  $X_1$-soliton and $X_2$-soliton, respectively.
We further impose restriction  $\mathrm{sgn}[b](\sqrt{1-b^2}\beta_I+b\beta_R)>0$
so that the solution has no  singularity.
Then we have the following results on the asymptotic behaviors of $X_j$-solitons.

\begin{theorem}\label{Th5}
Assume that $\mathrm{sgn}[b](\sqrt{1-b^2}\beta_I+b\beta_R)>0$.
In  case $|b|<1$ and when $bt\rightarrow \pm\infty$, the envelope  of $X_1$-soliton
asymptotically travels on a background $|q|^2=1$ and with  characteristics
\begin{align*}
&{\rm trajectory}: ~x(t)=\frac{1}{2\sqrt{1-b^2}}\ln|\beta \, b^{\pm 1}|-2bt,\\
&{\rm velocity:} ~x^\prime(t)=-2b,\\
&{\rm amplitude:} ~1-\frac{2\sqrt{1-b^2}\beta_I}{\mathrm{sgn}[b]\left(|\beta|
+\mathrm{sgn}[b](\sqrt{1-b^2}\beta_I+b\beta_R)\right)};
\end{align*}
when
$bt\rightarrow \pm\infty$, the $|q|^2$ of $X_2$-soliton asymptotically travels
 on a background $|q|^2=1$ and with  characteristics
\begin{align*}
&{\rm top~ trace:}~x(t)=-\frac{1}{2\sqrt{1-b^2}}\ln |\beta \, b^{\pm 1}| +2bt,\\
&{\rm velocity:}~x^\prime(t)=2b,\\
&{\rm amplitude:}~1+\frac{2((2b^2-1)\beta_I-2b\sqrt{1-b^2}\beta_R)
\sqrt{1-b^2}}{\mathrm{sgn}[b]\left(|\beta|
+\mathrm{sgn}[b](\sqrt{1-b^2}\beta_I+b\beta_R)\right)}.
\end{align*}
Each soliton gets a phase shift $-2\ln|b|$ due to interaction.
\end{theorem}

The value of amplitude of each soliton can be either larger or less than the background $|q_0|=1$.
This indicates various types of interactions.
Fig.\ref{fig-6} exhibits three types of interactions.
It is also notable that the value of amplitude of each soliton can be even equal to  the background $|q_0|=1$,
which means the soliton can vanish on the background.
This happens when $\beta_I=0$ for the $X_1$-soliton
and when $(2b^2-1)\beta_I-2b\sqrt{1-b^2}\beta_R=0$ for the $X_2$-soliton.
Illustrations are given in Fig.\ref{fig-7}.
Note that such a behavior usually appear in some coupled system and known as ``ghost soliton'',
cf.\cite{Hie-2002}.


\captionsetup[figure]{labelfont={bf},name={Fig.},labelsep=period}
\begin{figure}[htbp]
\centering
\subfigure[ ]{
\begin{minipage}[t]{0.32\linewidth}
\centering
\includegraphics[width=2.0in]{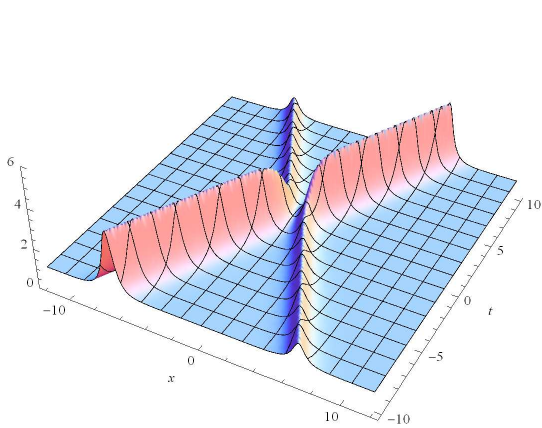}
\end{minipage}%
}%
\subfigure[ ]{
\begin{minipage}[t]{0.32\linewidth}
\centering
\includegraphics[width=2.0in]{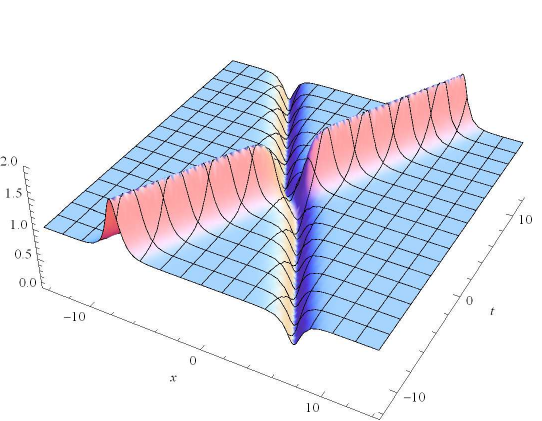}
\end{minipage}%
}%
\subfigure[ ]{
\begin{minipage}[t]{0.32\linewidth}
\centering
\includegraphics[width=2.0in]{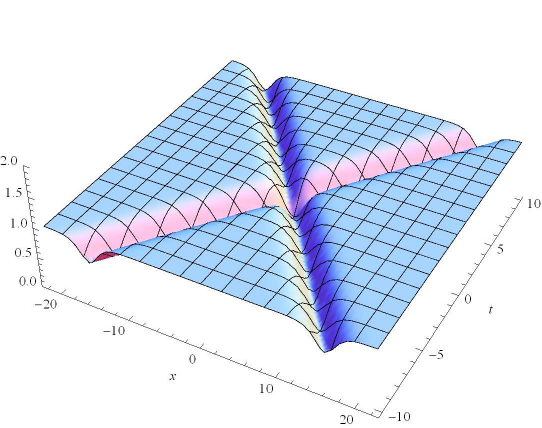}
\end{minipage}%
}
\caption{Interactions of the $X_1$-$X_2$ solitons for the defocusing reverse-space NLS equation \eqref{s-NLS-df}:
shape and motion of the envelope $|q|^2$ resulting from \eqref{fg-2ss}.
(a) 3D plot for $b=0.3, \beta=-1.6-0.4i$.~
(b) 3D plot for $b=0.3, \beta=-0.5+0.1i$.~
(c) 3D plot for $b=-0.8, \beta=1.4-i$.
In (a) and (b), $X_2$-soliton is the branch in up-right direction, and the other is $X_1$-soliton.
In (c) $X_1$-soliton is the branch in up-right direction, and the other is $X_2$-soliton.
}
\label{fig-6}
\end{figure}
\captionsetup[figure]{labelfont={bf},name={Fig.},labelsep=period}
\begin{figure}[ht!]
\centering
\subfigure[ ]{
\begin{minipage}[t]{0.4\linewidth}
\centering
\includegraphics[width=2.0in]{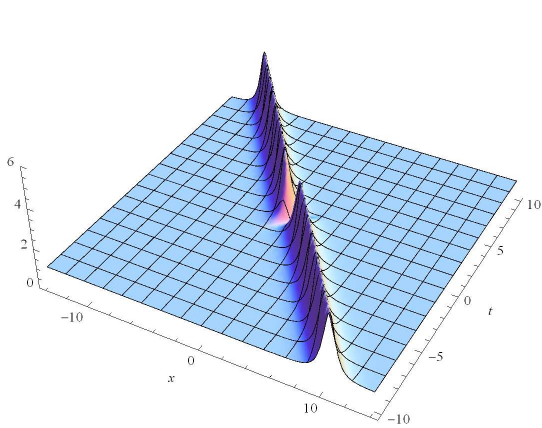}
\end{minipage}%
}%
\subfigure[ ]{
\begin{minipage}[t]{0.4\linewidth}
\centering
\includegraphics[width=2.0in]{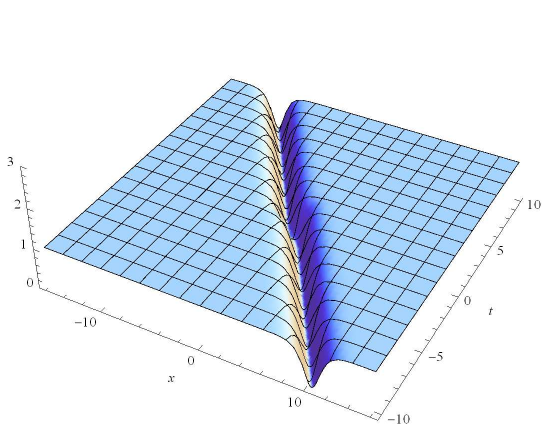}
\end{minipage}%
}%
\\
\subfigure[ ]{
\begin{minipage}[t]{0.4\linewidth}
\centering
\includegraphics[width=2.0in]{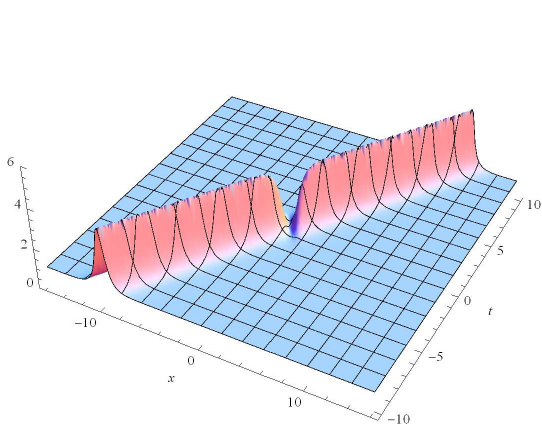}
\end{minipage}%
}%
\subfigure[ ]{
\begin{minipage}[t]{0.4\linewidth}
\centering
\includegraphics[width=2.0in]{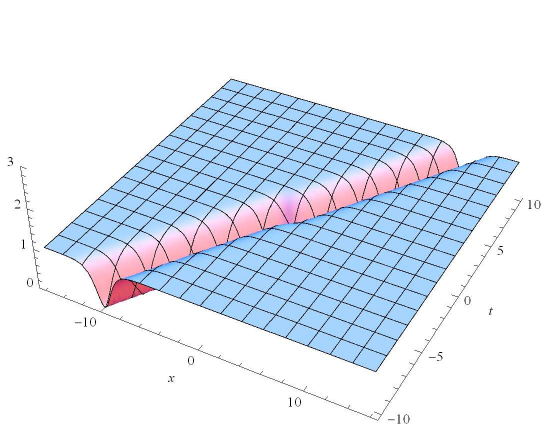}
\end{minipage}%
}%
\caption{
Interactions of  $X_1$-$X_2$ solitons with degeneration:
shape and motion of the envelope $|q|^2$ resulting from \eqref{fg-2ss}.
(a) 
3D plot for $b=0.5, \beta=0.5-\frac{\sqrt{3}}{2}i$.~
(b) 
3D plot for  $b=0.5, \beta=-0.5+\frac{\sqrt{3}}{2}i$.~
(c) 
3D plot for   $b=0.5, \beta=-0.5$.~
(d) 
3D plot for   $b=0.5, \beta=0.5$.
In (a) and (b),  $X_2$-soliton vanishes in the background $|q|^2=1$. In (c) and (d),  $X_1$-soliton
vanishes in the background $|q|^2=1$. }
\label{fig-7}
\end{figure}


\vskip 5pt
\noindent{\textbf{Case 2: $|b|>1$}}

In this case the  one-soliton solution of equation \eqref{s-NLS-df} resulting from \eqref{fg-2ss} can be written as
\begin{subequations}\label{sol-b>1}
\begin{equation}
q=\frac{G}{F},
\end{equation}
where
\begin{align}
G=& e^{2it}\left\{1+2b\left[((\sqrt{b^2-1}-b)-\beta\beta^*(\sqrt{b^2-1}+b))\cos(2\sqrt{b^2-1}x)\right.\right.
\nonumber\\
&+i((\sqrt{b^2-1}-b)+\beta\beta^*(\sqrt{b^2-1}+b))\sin(2\sqrt{b^2-1}x)\nonumber\\
& +(\beta(\sqrt{b^2-1}-b)-\beta^*(\sqrt{b^2-1}+b))\cos(4b\sqrt{b^2-1}t)\nonumber\\
&\left.\left. -i(\beta(\sqrt{b^2-1}-b)+\beta^*(\sqrt{b^2-1}+b))\sin(4b\sqrt{b^2-1}t)\right]\right\}, \\
F=& -b((\sqrt{b^2-1}-b)-\beta\beta^*(\sqrt{b^2-1}+b))\cos(2\sqrt{b^2-1}x) \nonumber\\
&\left.-ib((\sqrt{b^2-1}-b)+\beta\beta^*(\sqrt{b^2-1}+b))\sin(2\sqrt{b^2-1}x)\right.\nonumber \\
& +(\beta+\beta^*)\cos(4b\sqrt{b^2-1}t)+i(\beta^*-\beta)\sin(4b\sqrt{b^2-1}t).
\end{align}
\end{subequations}
Solution \eqref{sol-b>1} is   doubly periodic   with respect to both $x$ and $t$
and the periods are
\begin{equation}
P_x=\frac{\pi}{\sqrt{b^2-1}}, ~~P_t=\frac{\pi}{2b\sqrt{b^2-1}}.
\end{equation}
The solution is plotted in Fig.\ref{fig-8}.
Although there are some results on doubly periodic solutions,
which are constructed by Jacobi elliptic functions in general,
to our knowledge, the doubly periodic solution \eqref{sol-b>1} to
the defocusing reverse-space NLS  equation \eqref{s-NLS-df}  is not reported before.


\captionsetup[figure]{labelfont={bf},name={Fig.},labelsep=period}
\begin{figure}[ht]
\centering
\includegraphics[height=4cm,width=6cm]{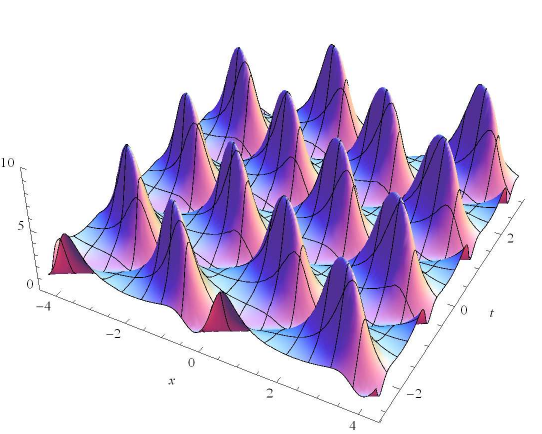}
\caption{ Envelope of doubly periodic solution \eqref{sol-b>1} for the defocusing reverse-space NLS
equation \eqref{s-NLS-df}, for $b=-1.25, \beta=1$.}
\label{fig-8}
\end{figure}


\subsubsection{Rational solutions}\label{sec-5-2-2}

Similar to the classical case, rational solutions can be seen as a limit case of breathers when taking $k_j\rightarrow i$, but the limit procedure should also be elaborated.

According to Table \ref{table-5-1}, for equation \eqref{s-NLS-df}, we have
$A=\mathrm{Diag}(\mathbf{K}_{m+1}, \mathbf{K}_{m+1}^*)$ and $T$ given by \eqref{T-sd} with
$\gamma=-1$.
Consider $\Phi(k,{c},{d})$ and $\Psi(k,{c},{d})$ defined in \eqref{sol-Lax pair} where we take $\lambda=k$ and $\delta=1$.
Introducing $\kappa=\sqrt{k^2+1}$, we have
\begin{subequations} \label{sol-Lax pair-st1}
  \begin{align}
  &\Phi(\kappa,c,d)=[c(\kappa)(\sqrt{\kappa^2-1}-\kappa)e^{-\kappa(x-2i\sqrt{\kappa^2-1}t)}+d(\kappa)(\sqrt{\kappa^2-1}+\kappa)e^{\kappa(x-2i\sqrt{\kappa^2-1}t)}]e^{it}\\
  &\Psi(\kappa,c,d)=[c(\kappa)e^{-\kappa(x-2i\sqrt{\kappa^2-1}t)}+d(\kappa)e^{\kappa(x-2i\sqrt{\kappa^2-1}t)}]e^{-it}
  \end{align}
\end{subequations}
Similarly, imposing $c(\kappa)=-d(-\kappa)$, we can have the same formal expressions  as eqs. \eqref{cd-kappa} and \eqref{Taylor-nls}.

Define
\begin{subequations}\label{ration-entry}
\begin{align}
&\phi^{odd}=(R_1,R_3,\cdots,R_{2m+1},S_1^*(-x),S_3^*(-x),\cdots,S_{2m+1}^*(-x))^T,\\
&\psi^{odd}=(S_1,S_3,\cdots,S_{2m+1},-R_1^*(-x),-R_3^*(-x),\cdots,-R_{2m+1}^*(-x))^T=T \phi^*(-x).
\end{align}
\end{subequations}
and $A=\mathrm{Diag}(\mathbf{K}_{m+1},-\mathbf{K}_{m+1}^*)$ with
\begin{equation}\label{K-t}
\mathbf{K}_{m+1}=\left(
           \begin{array}{ccccc}
             \alpha_0     &        0         &       0       &       \cdots       &   0    \\
             \alpha_2     &   \alpha_0        &    0           &     \cdots         &  0   \\
             \alpha_4     &   \alpha_2        &   \alpha_0     &          \cdots         &  0    \\
             \vdots       &   \vdots          &   \ddots       &  \ddots      &    \vdots   \\
             \alpha_{2m}  &   \alpha_{2m-2}   &   \cdots       &  \alpha_2    &   \alpha_0
           \end{array}
         \right),
\end{equation}
in which
$\alpha_{2j}=\frac{1}{(2j)!}\partial_\kappa^{2j}\sqrt{\kappa^2-1}|_{\kappa=0},~(j=0,1,2,\cdots)$.

It can be verified that $\phi^{odd}$ satisfies \eqref{4.5} with the above mentioned $A$, $T$, $\delta=1$ and $q_0=e^{2it}$.

This means, with such $\phi^{odd}$ and $\psi^{odd}$ as basic column vectors,
by the formula
\begin{equation}\label{q-snls}
q=q_0+\frac{g}{f}=e^{2it} +\frac{g}{f} ,
\end{equation}
the quasi double Wronskians
$$f=|\h\phi_m^{odd};\h\psi_m^{odd}|, ~ ~g=2|\h\phi_{m+1}^{odd};\h\psi_{m-1}^{odd}|$$
provide rational solutions to defocusing reverse-space NLS equation \eqref{s-NLS-df}.

The first order rational solution (for $m = 0$) is provided by
\begin{subequations}\label{fg-rat}
\begin{equation}
f=\Bigl | \begin{smallmatrix}
            R_1 &~ S_1\\
            S_1^*(-x) &~ -R_1^*(-x)
            \end{smallmatrix}
            \Bigr|,~~
g= 2 \Bigl | \begin{smallmatrix}
            R_1 &~ iR_1\\
            S_1^*(-x) &~ -iS_1^*(-x)
            \end{smallmatrix}
    \Bigr|
\end{equation}
with
\begin{equation}
R_1=2i(s_1-s_0(-i+2t+x))e^{it}, ~~ S_1=2(s_1-s_0(2t+x))e^{-it},
\end{equation}
\end{subequations}
where we have taken ${c}=s_0+s_1\kappa$.
Explicit form of the first order rational solution is given by
\begin{subequations}\label{q-snls-rat}
\begin{equation}
q=\frac{G}{F},
\end{equation}
where
\begin{align}
G=&-e^{2it}\left(2|s_1|^2+s_0s_1^*(3i-4t-2x)+s_0^*s_1(i-4t+2x)\right.\nonumber \\
&\left.+|s_0|^2(-1-8it+8t^2+2ix-2x^2)\right), \label{r-reversed-t}\\
F=& 2|s_1|^2+s_0s_1^*(i-4t-2x)+s_0^*s_1(-i-4t+2x)+|s_0|^2(1+8t^2+2ix-2x^2).
\end{align}
\end{subequations}

To understand the dynamics we investigate asymptotic behaviors of the above rational solution.
We introduce a new coordinate
\begin{equation*}
\left(X_1=x+2t,~t\right),
\end{equation*}
then rewrite \eqref{q-snls-rat} in this coordinate,
keep $X_1$ to be constant and let $t\rightarrow\pm \infty$.
It follows that
\begin{equation}\label{r-X1}
|q(X_1,t)|^2\rightarrow 1+\frac{8(|s_0|^2+\mathrm{Im}(s_0^*s_1))}{4|s_0|^2\left(X_1-\frac{\mathrm{Re}(s_0^*s_1)}{|s_0|^2}\right)^2+\frac{(|s_0|^2+2(\mathrm{Im}(s_0^*s_1))^2}{|s_0|^2}},
~~ t\rightarrow\pm \infty,
\end{equation}
for convenience, we call it $X_1$-soliton.
It indicates that, asymptotically,
this is a wave traveling on the background plane $|q_0|^2=1$,
along the line $x=-2t+\frac{\mathrm{Re}(s_0^*s_1)}{|s_0|^2}$, with amplitude $1+\frac{8(|s_0|^2+\mathrm{Im}(s_0^*s_1))|s_0|^2}{(|s_0|^2+2(\mathrm{Im}(s_0^*s_1))^2}$
and without  phase shift due to interaction.
The wave can be above the background plane when  $|s_0|^2+\mathrm{Im}(s_0^*s_1)>0$,
or below the background plane when  $|s_0|^2+\mathrm{Im}(s_0^*s_1)<0$, or vanishes in the  background plane when  $|s_0|^2+\mathrm{Im}(s_0^*s_1)=0$.
We further introduce a second coordinate frame
\begin{equation*}
\left(X_2=x-2t,~t\right),
\end{equation*}
in terms of which we rewrite \eqref{q-snls-rat}.
Then keeping $X_2$ to be constant and letting $t\rightarrow\pm \infty$, we find
\begin{equation}\label{r-X2}
|q(X_2,t)|^2\rightarrow 1+\frac{-8\mathrm{Im}(s_0^*s_1)}{4|s_0|^2\left(X_2+\frac{\mathrm{Re}(s_0^*s_1)}{|s_0|^2}\right)^2+\frac{(|s_0|^2+2(\mathrm{Im}(s_0^*s_1))^2}{|s_0|^2}},
~~ t\rightarrow\pm \infty.
\end{equation}
This implies that, when $t\rightarrow\pm \infty$, there is a wave
($X_2$-soliton for short)
traveling on the background plane $|q_0|^2=1$, along the line
$x=2t-\frac{\mathrm{Re}(s_0^*s_1)}{|s_0|^2}$,
with amplitude $1+\frac{-8\mathrm{Im}(s_0^*s_1)|s_0|^2}{(|s_0|^2+2(\mathrm{Im}(s_0^*s_1))^2}$
and without  phase shift due to interaction.
The wave can be either above or below or vanishes in  the background plane,
depending on  the sign of $\mathrm{Im}(s_0^*s_1)$.
We summarize these asymptotic behaviors in the follow theorem.

\begin{theorem}\label{Th6}
When $t\rightarrow \pm\infty$, the envelope $|q|^2$ of $X_1$-soliton
asymptotically travels on a background $|q_0|^2=1$ with  characteristics
\begin{align*}
&{\rm trajectory}: ~x(t)=-2t+\frac{\mathrm{Re}(s_0^*s_1)}{|s_0|^2},\\
&{\rm velocity:} ~x^\prime(t)=-2,\\
&{\rm amplitude:} ~1+\frac{8(|s_0|^2+\mathrm{Im}(s_0^*s_1))|s_0|^2}{(|s_0|^2+2(\mathrm{Im}(s_0^*s_1))^2},
\end{align*}
and the envelope $|q|^2$ of $X_2$-soliton asymptotically travels
 on a background $|q_0|^2=1$ with  characteristics
\begin{align*}
&{\rm top~ trace:}~x(t)=2t-\frac{\mathrm{Re}(s_0^*s_1)}{|s_0|^2},\\
&{\rm velocity:}~x^\prime(t)=2,\\
&{\rm amplitude:}~1+\frac{-8\mathrm{Im}(s_0^*s_1)|s_0|^2}{(|s_0|^2+2(\mathrm{Im}(s_0^*s_1))^2}.
\end{align*}
Asymptotically, no phase shift occurs for each soliton.
\end{theorem}

Various types of interactions are illustrated in Fig.\ref{fig-9},
which coincide with the above results of asymptotic analysis.
Note that, considering the signs of $|s_0|^2+\mathrm{Im}(s_0^*s_1)$ and $-8\mathrm{Im}(s_0^*s_1)$,
it is impossible to have both waves below  the background plane,
neither one wave below the background plane and another vanishing.


\captionsetup[figure]{labelfont={bf},name={Fig.},labelsep=period}
\begin{figure}[ht!]
\centering
\subfigure[ ]{
\begin{minipage}[t]{0.4\linewidth}
\centering
\includegraphics[width=2.0in]{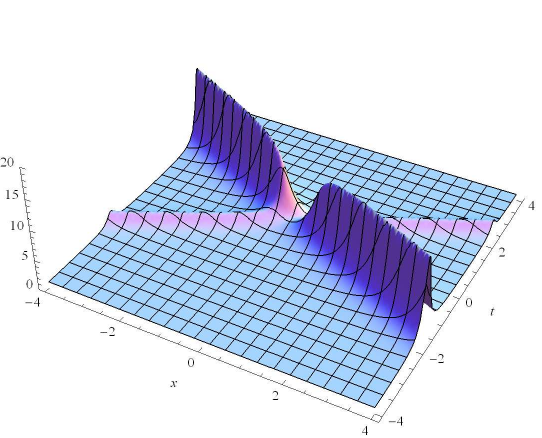}
\end{minipage}%
}%
\subfigure[ ]{
\begin{minipage}[t]{0.4\linewidth}
\centering
\includegraphics[width=2.0in]{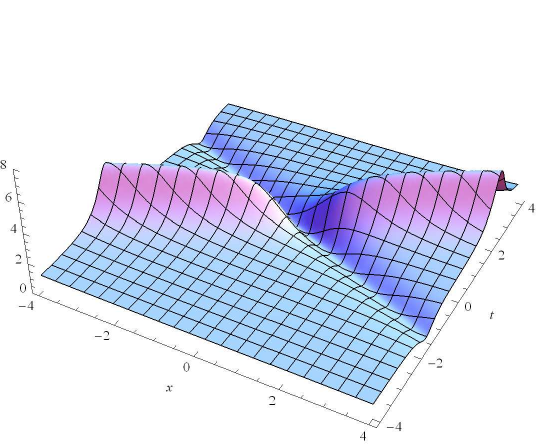}
\end{minipage}%
}%
\\
\subfigure[ ]{
\begin{minipage}[t]{0.4\linewidth}
\centering
\includegraphics[width=2.0in]{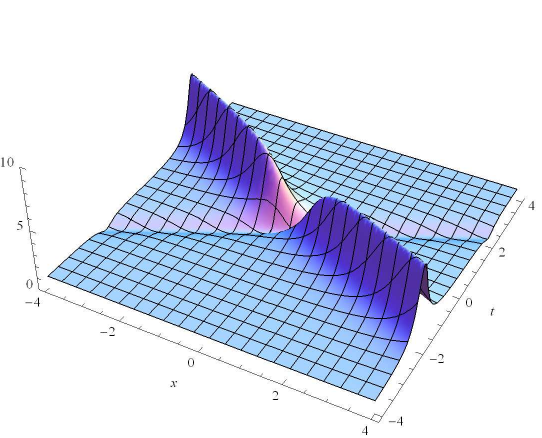}
\end{minipage}%
}%
\subfigure[ ]{
\begin{minipage}[t]{0.4\linewidth}
\centering
\includegraphics[width=2.0in]{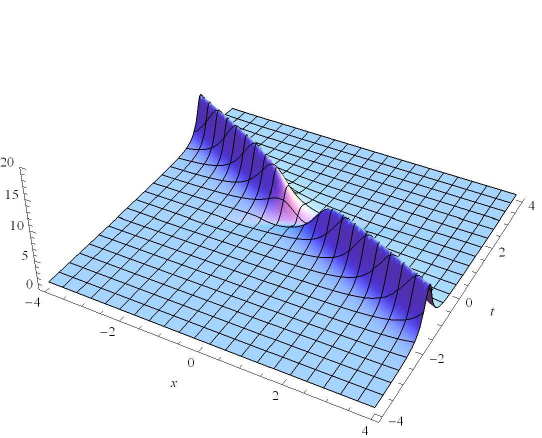}
\end{minipage}%
}%
\caption{Interactions of the $X_1$-$X_2$ solitons for the defocusing reverse-space NLS equation \eqref{s-NLS-df}:
shape and motion of the envelope $|q|^2$ resulting from \eqref{q-snls-rat}.
(a) 
3D plot for $s_0=1+1.4i,~s_1=1+i$.~
(b) 
3D plot for  $s_0=1+2i,~s_1=4+i$.~
(c) 
3D plot for  $s_0=1+2.8i,~s_1=i$.~
(d) 
3D plot for  $s_0=i,~s_1=i$.~
In (a), $X_1$-soliton and $X_2$-soliton are both bright solitons, while the amplitude of $X_1$-soliton is bigger than $X_2$-soliton;
in (b), the $X_1$-soliton is dark soliton while the $X_2$-soliton is bright soliton;
in (c), exactly the opposite conclusion;
in (d), $X_2$-soliton vanishes in the background $|q_0|^2=1$.
}
\label{fig-9}
\end{figure}



\subsection{The  defocusing reverse-time nonlocal  NLS equation}\label{sec-5-3}

For the defocusing reverse-time NLS equation
\begin{equation}\label{t-NLS-df}
iq_t=q_{xx}-2 q^2q(-t)
\end{equation}
with nonzero background $q_0=e^{2it}$,
we can analyze solutions resulting from $k_1=ib$, $b\in \mathbb{R}$,
as we have done in Sec.\ref{sec-5-2} for the reverse-space nonlocal NLS equation \eqref{s-NLS-df}.
However, it turns out that
the analysis procedure of these solutions and their dynamics are all similar to those in
Sec.\ref{sec-5-2} for   equation \eqref{s-NLS-df}.
Let us explain the statement below.

Rewrite \eqref{5.30b} and \eqref{5.30c} as
\begin{subequations}\label{5.42}
\begin{align}
&\phi_{(1)}^+(x,t,\alpha_1,\beta_1)=\left (\alpha_1  e^{\eta_1(x,t)}
+\beta_1 e^{-\eta_1(x,t)}\right)e^{it},\label{5.42a}\\
&\psi_{(1)}^+(x,t,\alpha_1,\beta_1)=\left( \alpha_1 \tilde{d}(k_1) e^{\eta_1(x,t)}
-\beta_1 \tilde{c}(k_1) e^{-\eta_1(x,t)}\right)e^{-it},\label{5.42b}
\end{align}
where
\begin{equation}
\eta_1(x,t)=\sqrt{k^2_1+1} (x-2ik_1 t),~~
\tilde{c}(k_1)=\sqrt{k^2_1+1}+k_1,~~\tilde{d}(k_1)=\sqrt{k^2_1+1}-k_1.
\end{equation}
\end{subequations}
According to Sec.\ref{sec-4-4-2}, for the reverse-space NLS equation \eqref{s-NLS-df}, the vectors
$\phi$ and $\psi$ in the $2\times 2 $ double Wronskians $f$ and $g$ are
\begin{equation}\label{5.43}
\phi=\phi_{[x]}(x,t,\alpha_1,\beta_1)=
\left(\begin{array}{c}
       \phi_{(1)}^+(x,t,\alpha_1,\beta_1)\\
       \psi_{(1)}^{+*}(-x,t,\alpha_1,\beta_1)
\end{array}\right),~~ \psi=\psi_{[x]}(x,t,\alpha_1,\beta_1)=T \phi_{[x]}^*(-x,t,\alpha_1,\beta_1),
\end{equation}
and for the reverse-time NLS equation \eqref{t-NLS-df}, we have
\begin{equation}
\phi=\phi_{[t]}(x,t,\alpha_1,\beta_1)=
\left(\begin{array}{c}
       \phi_{(1)}^+(x,t,\alpha_1,\beta_1)\\
       \psi_{(1)}^{+}(x,-t,\alpha_1,\beta_1)
\end{array}\right),~~ \psi=\psi_{[t]}(x,t,\alpha_1,\beta_1)=T \phi_{[t]}(x,-t,\alpha_1,\beta_1),
\end{equation}
where $T=\Big (\begin{smallmatrix}0& 1\\ \gamma & 0\end{smallmatrix}\Big )$
with $\gamma=\pm 1$,
the subscripts $[x]$ and $[t]$ stand for the reverse-space and reverse-time, respectively.
It can be verified that, when $k_1=ib$, $b\in \mathbb{R}$ (as we have taken in Sec.\ref{sec-5-2}),
we have
\begin{equation}
\eta_1(x,t)=\eta^*_1(x,t),~~
\eta_1(-x,t)=-\eta(x,-t),~~ \tilde{c}^*(k_1)=\tilde{d}(k_1).
\end{equation}
It then follows that
\begin{equation}
 \phi_{[x]}(x,t,\alpha_1,\beta_1) =B \phi_{[t]}(x,t,\beta_1^*,\alpha_1^*)
\end{equation}
and
\begin{equation}
\psi_{[x]}(x,t,\alpha_1,\beta_1)=T \phi_{[x]}^*(-x,t,\alpha_1,\beta_1)
=-B T \phi_{[t]}(x,-t,\beta_1^*,\alpha_1^*)
=-B\psi_{[t]}(x,t,\beta_1^*,\alpha_1^*),
\end{equation}
where
$B=\Big (\begin{smallmatrix}1& 0\\ 0 & -1\end{smallmatrix}\Big )$.
These relations indicate that
$|q|^2$ resulting from the above $\phi$ and $\psi$ for the reverse-time NLS equation \eqref{t-NLS-df}
are similar to those of the reverse-space NLS equation \eqref{s-NLS-df}.
We skip presenting illustrations.

\subsection{The defocusing reverse-space-time nonlocal NLS equation}\label{sec-5-4}

For the solutions of the nonlocal defocusing reverse-space-time NLS equation (i.e. \eqref{tsnls} with $\delta=1$),
\begin{equation}\label{ts-NLS-df}
iq_t=q_{xx}-2 q^2q(-x,-t)
\end{equation}
with the plane wave background $q_0=e^{2it}$,
the matrices $A$ and $T$ take the form (see Table \ref{table-5-2})
\begin{align}
 A=\left(\begin{array}{cc}
                   \mathbf{K}_{m+1} & \mathbf{0}_{m+1} \\
                   \mathbf{0}_{m+1} & \mathbf{H}_{m+1}
                 \end{array}
               \right),\quad
               T=\left(
     \begin{array}{cc}
       i\mathbf{I}_{m+1} & \mathbf{0}_{m+1} \\
       \mathbf{0}_{m+1} & -i\mathbf{I}_{m+1}
     \end{array}
   \right).
\end{align}
Solution  to equation \eqref{ts-NLS-df} with the background solution $q_0=e^{2it}$ is written as
\begin{equation*}
q=q_0+\frac{g}{f}.
\end{equation*}
For the case of $m=0$, (i.e. one-soliton case), there are
\begin{equation}\label{5.50}
f=\left | \begin{smallmatrix}
            \phi_{(1)}^+ &~ i\phi_{(1)}^{+}(-x,-t)\\
            \phi_{(1)}^{-} &~ -i\phi_{(1)}^{-}(-x,-t)
            \end{smallmatrix}
            \right|,~~
g=2 \left | \begin{smallmatrix}
            \phi_{(1)}^+ &~ k_1\phi_{(1)}^+\\
            \phi_{(1)}^{-} &~ h_1\phi_{(1)}^{-}
            \end{smallmatrix}
    \right|,
\end{equation}
where in light of \eqref{entry-1},
$$\phi_{(1)}^+=\Phi(k_1,\hat{c}_1^+,i\hat{c}_1^+),~\phi_{(1)}^-=\Phi(h_1,\hat{c}_1^-,-i\hat{c}_1^-),$$
$\Phi$ is defined in \eqref{sol-Lax pair-Phi-st},
$k_j, h_j \in \mathbb{C}$, $\hat{c}_j^+$ and $\hat{c}_j^-$ can be
arbitrary complex functions of $k_j$ and $h_j$, respectively.
For the case of $m=1$, (i.e., two-soliton case), we have
\begin{equation}\label{5.51}
f=\left| \begin{smallmatrix}
            \phi_{(1)}^+    &~k_1\phi_{(1)}^+      &~ i\phi_{(1)}^{+}(-x,-t)  &~ -i k_1\phi_{(1)}^{+}(-x,-t)\\
            \phi_{(2)}^+    &~k_2\phi_{(2)}^+     &~ i\phi_{(2)}^{+}(-x,-t)   &~ -i k_2\phi_{(2)}^{+}(-x,-t)\\
            \phi_{(1)}^{-}  &~h_1\phi_{(1)}^{-}      &~ -i\phi_{(1)}^{-}(-x,-t)  &~ i h_1\phi_{(1)}^{-}(-x,-t)\\
            \phi_{(2)}^{-}  &~h_2\phi_{(2)}^{-}      &~ -i\phi_{(2)}^{-}(-x,-t)  &~ i h_2\phi_{(2)}^{-}(-x,-t)
            \end{smallmatrix}
           \right|,~~
g=\left| \begin{smallmatrix}
            \phi_{(1)}^+    &~k_1\phi_{(1)}^+      &~ k_1^2\phi_{(1)}^{+}(-x,-t)  &~ i  \phi_{(1)}^{+}(-x,-t)\\
            \phi_{(2)}^+    &~k_2\phi_{(2)}^+     &~ k_2^2\phi_{(2)}^{+}(-x,-t)   &~ i \phi_{(2)}^{+}(-x,-t)\\
            \phi_{(1)}^{-}  &~h_1\phi_{(1)}^{-}      &~ h_1^2\phi_{(1)}^{-}(-x,-t)  &~ -i \phi_{(1)}^{-}(-x,-t)\\
            \phi_{(2)}^{-}  &~h_2\phi_{(2)}^{-}      &~ h_2^2\phi_{(2)}^{-}(-x,-t)  &~ -i  \phi_{(2)}^{-}(-x,-t)
            \end{smallmatrix}
            \right|,
\end{equation}
where $k_1, k_2, h_1,h_2\in \mathbb{C}$ and
$$\phi_{(j)}^+=\Phi(k_j,\hat{c}_j^+,i\hat{c}_j^+),~~\phi_{(j)}^-=\Phi(h_j,\hat{c}_j^-,-i\hat{c}_j^-),~j=1,2,$$
$\Phi$ is defined in \eqref{sol-Lax pair-Phi-st},
$k_j, h_j, \hat{c}_j^{\pm} \in \mathbb{C}$.

In the following, we list some solutions of reverse-space-time NLS equation \eqref{ts-NLS-df}
for the special form of $\mathbf{H}_{m+1}$.

When $\mathbf{H}_{m+1}=-\mathbf{K}_{m+1}$, one-soliton solution of equation \eqref{ts-NLS-df} reads
\begin{equation}\label{sol-reversed-x-t-1}
q=\frac{G}{F},
\end{equation}
where
\begin{align*}
G=& e^{2it}\left\{1-2k_1\left[\cosh(2\sqrt{k_1^2+1}x)+\sqrt{k_1^2+1}
\sin(4k_1\sqrt{k_1^2+1}t)+ik_1\cos(4k_1\sqrt{k_1^2+1}t)\right]\right\},\\
F=& k_1\cosh(2\sqrt{k_1^2+1}x)-i\cos(4k_1\sqrt{k_1^2+1}t),
\end{align*}
and we have taken $\hat{c}_1^+=\hat{c}_1^-=1$.
Nonsingular solution occurs when $k_1\in \mathbb{R}$.
In this case, the envelope of \eqref{sol-reversed-x-t-1} describes a wave
moving along the straight line $x=0$ and  oscillating in time
with a period $P_t=\frac{\pi}{2|k_1|\sqrt{k_1^2+1}}$.
Such a solution is depicted in Fig.\ref{fig-14}(a).
Note that the period tends to 0 when $|k_1|$ is large enough,
and such a change is illustrated in Fig.\ref{fig-14}(b).


\captionsetup[figure]{labelfont={bf},name={Fig.},labelsep=period}
\begin{figure}[h]
\centering
\subfigure[ ]{
\begin{minipage}[t]{0.4\linewidth}
\centering
\includegraphics[width=2.0in]{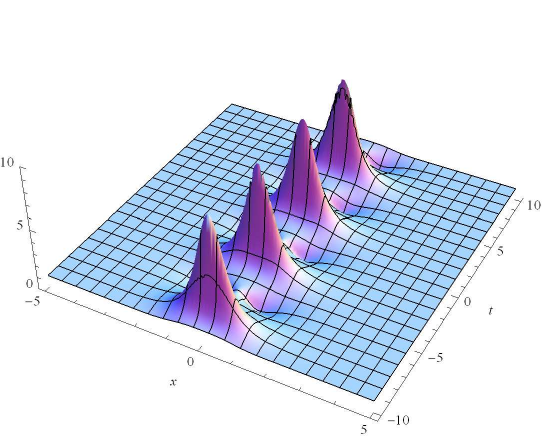}
\end{minipage}%
}%
\subfigure[ ]{
\begin{minipage}[t]{0.4\linewidth}
\centering
\includegraphics[width=2.0in]{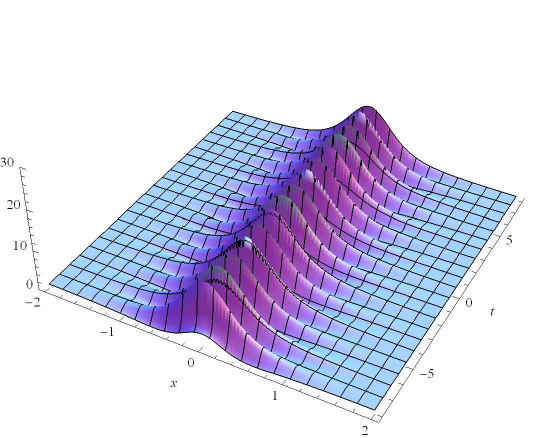}
\end{minipage}%
}%
\caption{Shape and motion of stationary breather of $|q|^2$ resulting from \eqref{sol-reversed-x-t-1}
for the nonlocal defocusing reverse-space-time NLS equation \eqref{ts-NLS-df}.
(a) 3D-plot for  $k_1=0.3$.~
(b) 3D-plot for  $k_1=1$.
}
\label{fig-14}
\end{figure}


When $\mathbf{H}_{m+1}=-\mathbf{K}^*_{m+1}$, by calculation, we get one-soliton solution which reads
\begin{equation}\label{sol-reversed-x-t-2}
q=e^{2it}+\frac{G}{F},
\end{equation}
where
\begin{equation*}
\begin{aligned}
G=&-2ia_1[(1-a_1^2-b_1^2+v_{11}^2+v_{12}^2-2ia_1)\cosh(2B_1)+2(i(b_1v_{11}-a_1v_{12})+v_{12})\sinh(2B_1)\\
&+(-1+a_1^2+b_1^2+v_{11}^2+v_{12}^2+2ia_1)\cos(2B_2)
+2i((a_1v_{11}+b_1v_{12})+iv_{11})\sin(2B_2)]e^{2it},\\
F=&(-1+a_1^2+b_1^2+v_{11}^2+v_{12}^2+2ia_1)\cosh(2B_1)
+(1-a_1^2-b_1^2+v_{11}^2+v_{12}^2-2ia_1)\cos(2B_2),
\end{aligned}
\end{equation*}
we take $\hat{c}_1^+=\hat{c}_1^-=1$ and denote
\begin{equation*}
\begin{aligned}
k_1 =& a_1 + ib_1, ~\sqrt{k_1^2+1}=v_{11}+iv_{12}, ~~a_1,b_1,v_{11},v_{12}\in  \mathbb{R},\\
B_1=&v_{11}x+2(a_1v_{12}+b_1v_{11})t,~B_2=v_{12}x+2(b_1v_{12}-a_1v_{11})t.
\end{aligned}
\end{equation*}
The envelope of  \eqref{sol-reversed-x-t-2} behaves like breather
 which travels along the line $x=-\frac{2(a_1v_{12}+b_1v_{11})t}{v_{11}}$.
An illustration is given in Fig.\ref{fig-15}(a).
Note that $a_1=0$ yields trivial solution and $b_1=0$ leads to the solution \eqref{sol-reversed-x-t-1}.
We can also calculate two-soliton solution from \eqref{5.51}
of this case, where we take $h_j=-k{_j^*},j=1,2$.
Its envelope describes  a head-on collision of two breathers, as shown in Fig.\ref{fig-15}(b).


\captionsetup[figure]{labelfont={bf},name={Fig.},labelsep=period}
\begin{figure}[htbp]
\centering
\subfigure[ ]{
\begin{minipage}[t]{0.4\linewidth}
\centering
\includegraphics[width=2.0in]{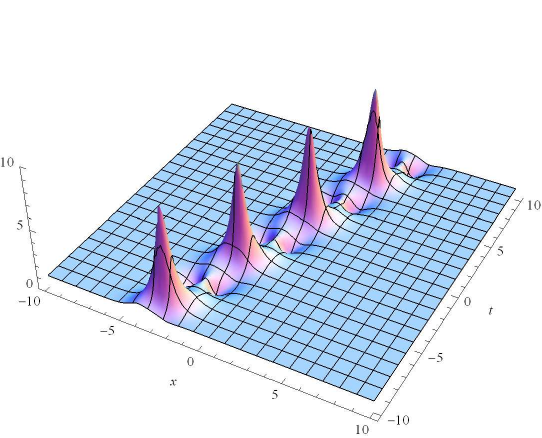}
\end{minipage}%
}%
\subfigure[ ]{
\begin{minipage}[t]{0.4\linewidth}
\centering
\includegraphics[width=2.0in]{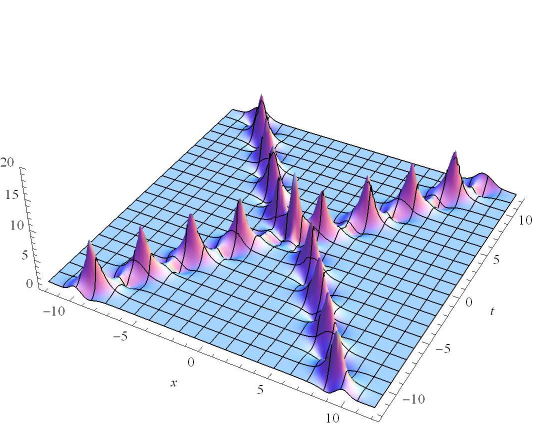}
\end{minipage}%
}%
\caption{Shape and motion of the envelope $|q|^2$ of the
nonlocal defocusing reverse-space-time NLS equation \eqref{ts-NLS-df}.
(a) 3D-plot for the one breather resulting \eqref{sol-reversed-x-t-2} with $k_1=0.3-0.15i$.~
(b) 3D-plot for two breathers interaction resulting from \eqref{5.51} with $k_1=0.5+0.3i,k_2=0.5-0.3i,\hat{c}_1^{\pm}=\hat{c}_2^{\pm}=1$.
}
\label{fig-15}
\end{figure}


Finally, $\mathbf{H}_{m+1}=\mathbf{K}^*_{m+1}$ with $k_i\in \mathbb{C}$,
we claim that dynamics of solutions are similar to the reverse-space case.
The explanation is similar to what we have done in Sec.\ref{sec-5-3}.
In fact, write the vectors $\phi$ and $\psi$ with $h_1=k_1^*$ in the double Wronskians \eqref{5.50}
as
\begin{eqnarray}
          \phi=\phi_{[x,t]}=\left(
           \begin{array}{c}
             \phi_{(1)}^+(x,t,k_1,\hat{c}_1,i\hat{c}_1)\\
             \phi_{(1)}^-(x,t,k_1^*,\hat{c}_1,-i\hat{c}_1)
           \end{array}
         \right),
         \psi=\psi_{[x,t]}=\left(
           \begin{array}{c}
            i\phi_{(1)}^+(-x,-t,k_1,\hat{c}_1,i\hat{c}_1)\\
            -i\phi_{(1)}^-(-x,-t,k_1^*,\hat{c}_1,-i\hat{c}_1)
           \end{array}
         \right),
\end{eqnarray}
where the subscript  $[x,t]$ stands for the reverse-space-time.
We rewrite the vectors in \eqref{5.43} (for the reverse-space NLS equation \eqref{s-NLS-df}) as
\begin{eqnarray}
          \phi=\phi_{[x]}=\left(
           \begin{array}{c}
             \phi_{(1)}^+(x,t,k_1,c,d)\\
             \psi_{(1)}^{+*}(-x,t,k_1,c,d)
           \end{array}
         \right),
         ~\psi=\psi_{[x]}=\left(
           \begin{array}{c}
            \psi_{(1)}^+(x,t,k_1,c,d)\\
            -\phi_{(1)}^{+*}(-x,t,k_1,c,d)
           \end{array}
         \right).
\end{eqnarray}
In the special case $c_1=1, d_1=i, \hat{c}_1=1$, the following relations hold:
\begin{equation}\label{raletion-phi-psi}
\phi_{[x]}=C\phi_{[x,t]},~~\psi_{[x]}=C\psi_{[x,t]},
\end{equation}
where $C=\Big(\begin{smallmatrix} 1 & 0 \\ 0 & -i \end{smallmatrix}\Big)$.
Besides, the construction of $f$ and $g$, the same
$A=\Big(\begin{smallmatrix}  k_1 & 0 \\ 0 & k_1^* \end{smallmatrix}\Big)$
is used.
This indicates that in the case $\mathbf{H}_{m+1}=\mathbf{K}^*_{m+1}$,
 the analysis of dynamics of solutions for the revere-space-time NLS equation
will be similar to those of the revere-space  NLS equation that has been investigated in Sec.\ref{sec-5-2}.
We skip it.

\section{Conclusions} \label{sec-6}

In this paper, by means of the bilinearisation-reduction approach,
solutions for the classical and nonlocal NLS equations with nonzero backgrounds were constructed
in a systematical way. Solutions are presented in terms of quasi double Wronskians.
Asymptotic analysis and illustrations were provided to understand dynamics of solutions,
in particular, breathers and rogue waves of the classical focusing NLS equation \eqref{cnls}
and solitons and rational solutions of the reverse-space nonlocal NLS equation \eqref{snls}.
One can see that the nonzero backgrounds bring more interesting behaviors in the dynamics of solutions.
In addition, although the solutions are given in terms of quasi double Wronskians
(not standard double Wronskians), the reduction technique is still effective.
In light of Theorem \ref{th2}
one can also use the double Wronskians given in Theorem \ref{th2} if $q_0$ is independent of $x$.
This bilinearisation-reduction technique can also be extended to the other integrable equations
with nonzero backgrounds, which will be investigated elsewhere.

\subsection*{Acknowledgments}

The authors are grateful to the referees for their expertise and invaluable comments.
This project is  supported by the NSF of China (Nos. 11875040, 12271334, 12126352, 12126343).

\begin{appendix}

\section{Proof of Theorem \ref{th1}} \label{app-A}

To prove Theorem \ref{th1}, we first recall the following lemmas.

\begin{lemma}\label{lemma 1}\cite{FN-PLA-1983}
Suppose that $\mathbf{D}$ is an arbitrary $s\times (s-2)$ matrix, and $\mathbf{a}$, $\mathbf{b}$,
$\mathbf{c}$ and $\mathbf{d}$ are column vectors of order $s$, then
\begin{equation}\label{id-1}
|\mathbf{D}, \mathbf{a},\mathbf{b}||\mathbf{D},\mathbf{c},\mathbf{d}|
-|\mathbf{D}, \mathbf{a},\mathbf{c}||\mathbf{D},\mathbf{b},\mathbf{d}|
+|\mathbf{D}, \mathbf{a},\mathbf{d}||\mathbf{D},\mathbf{b}, \mathbf{c}|=0.
\end{equation}
\end{lemma}

\begin{lemma}\label{lemma 2}\cite{DJZ-2006,ZhangDJ-RMP-2014}
Suppose that  $\Xi=(a_{js})_{M\times M}$ is an $M\times M$ matrix with
column vector set $\{\alpha_j\}$ and  row vector set $\{\beta_j\}$.
${\cal{P}}=(P_{js})_{M\times M}$ is an $M\times M$ operator matrix where
each $P_{js}$ is an operator. Then we have
\begin{equation}\label{id-D}
\sum^M_{s=1} |\alpha_1, \cdots,\alpha_{s-1},~C_s\circ\alpha_s, ~\alpha_{s+1},\cdots,\alpha_M|
 =\sum^M_{j=1}\left|~\begin{matrix}
 \beta_1\\
 \vdots\\
 \beta_{j-1}\\
 R_j\circ \beta_j\\
 \beta_{j+1}\\
 \vdots\\
 \beta_M
\end{matrix}~\right|,
\end{equation}
where $\{C_s\}$ and $\{R_j\}$ are respectively the column and row vector sets of  ${\cal{P}}$,
and ``$\circ$'' denotes $C_s \circ\alpha_s = \left(P_{1s} a_{1s}, ~P_{2s} a_{2s},\cdots, P_{Ms} a_{Ms}\right)^T$
and $R_j\circ\beta_j=\left(P_{j1} a_{j1}, ~P_{j2} a_{j2},\cdots, P_{jM} a_{jM}\right)$.
\end{lemma}

\noindent
\textit{Proof of Theorem \ref{th1}}

Direct calculation yields
\begin{align*}
f_x=&|\h\phi_{m-1},\phi_{m+1};\h\psi_m|+|\h\phi_{m};\h\psi_{m-1},\psi_{m+1}|,\\
f_{xx}=&|\h\phi_{m-2},\phi_m,\phi_{m+1};\h\psi_m|+|\h\phi_{m-1},\phi_{m+2};\h\psi_m|+
2|\h\phi_{m-1},\phi_{m+1};\h\psi_{m-1},\psi_{m+1}|\\
&+|\h\phi_{m};\h\psi_{m-2},\psi_m,\psi_{m+1}|+|\h\phi_{m};\h\psi_{m-1},\psi_{m+2}|\\
&+2q_0|\h\phi_{m-1};\h\psi_{m+1}|-2r_0|\h\phi_{m+1};\h\psi_{m-1}|,
\\
if_t=&-2|\h\phi_{m-2},\phi_m,\phi_{m+1};\h\psi_m|+2|\h\phi_{m-1},\phi_{m+2};\h\psi_m|\\
&+2|\h\phi_{m};\h\psi_{m-2},\psi_m,\psi_{m+1}|-2|\h\phi_{m};\h\psi_{m-1},\psi_{m+2}|\\
&+2q_0|\h\phi_{m-1};\h\psi_{m+1}|+2r_0|\h\phi_{m+1};\h\psi_{m-1}|,
\\
ig_{t}/2=&-2|\h\phi_{m-1},\phi_{m+1},\phi_{m+2};\h\psi_{m-1}|+2|\h\phi_{m},\phi_{m+3};\h\psi_{m-1}|\\
&+2|\h\phi_{m+1};\h\psi_{m-3},\psi_{m-1},\psi_{m}|-2|\h\phi_{m+1};\h\psi_{m-2},\psi_{m+1}|\\
&+2q_0(|\h\phi_{m};\h\psi_{m-1},\psi_{m+2}|+|\h\phi_{m-1},\phi_{m+1};\h\psi_{m-1},\psi_{m+1}|
+|\h\phi_{m-2},\phi_{m},\phi_{m+1};\h\psi_{m}|)\\
&-q_0r_0g-q_{0x}f_x,
\\
g_x/2=&|\h\phi_{m},\phi_{m+2};\h\psi_{m-1}|+|\h\phi_{m+1};\h\psi_{m-2},\psi_{m}|-q_0f_x,\\
g_{xx}/2=&|\h\phi_{m-1},\phi_{m+1},\phi_{m+2};\h\psi_{m-1}|+|\h\phi_{m},\phi_{m+3};\h\psi_{m-1}|+
2|\h\phi_{m},\phi_{m+2};\h\psi_{m-2},\psi_{m}|\\
&+|\h\phi_{m+1};\h\psi_{m-3},\psi_{m-1},\psi_{m}|+|\h\phi_{m+1};\h\psi_{m-2},\psi_{m+1}|\\
&+q_0(|\h\phi_{m};\h\psi_{m-1},\psi_{m+2}|-|\h\phi_{m-1},\phi_{m+2};\h\psi_{m}|\\
&-|\h\phi_{m};\h\psi_{m-2},\psi_{m},\psi_{m+1}|+|\h\phi_{m-2},\phi_{m},\phi_{m+1};\h\psi_{m}|)
-q_0f_{xx}-q_{0x}f_x.
\end{align*}
Next, using Lemma \ref{lemma 2}
we derive some relations of quasi double Wronskians.
Taking $\Xi=|\h\phi_{m}; \h\psi_{m}|$ and (for $1\leq j \leq 2m+2$)
\[ P_{js}=\left\{
\begin{array}{ll}
\partial_x-(-1)^mq_0\psi_m \circ\phi_m^{-1}\circ,~ & 0\leq s  \leq m,\\
-\partial_x+(-1)^mr_0\phi_m\circ\psi_m^{-1}\circ,~ & m+1\leq s \leq 2m+2,
\end{array}
\right.
\]
where the ``$\circ$'' is defined as in Lemma \ref{lemma 2}.
One can find from the relation \eqref{id-D} that
\begin{equation*}
(\textrm{Tr}A)f=|\h\phi_{m-1},\phi_{m+1};\h\psi_m|+|\h\phi_{m};\h\psi_{m-1},\psi_{m+1}|,
\end{equation*}
where $\textrm{Tr}A$ stands for the trace of matrix $A$.
Similarly, we can get
\begin{equation*}
\begin{aligned}
(\textrm{Tr}A)^2f=&|\h\phi_{m-2},\phi_m,\phi_{m+1};\h\psi_m|+|\h\phi_{m-1},\phi_{m+2};\h\psi_m|
-2|\h\phi_{m-1},\phi_{m+1};\h\psi_{m-1},\psi_{m+1}|\\
&+|\h\phi_{m};\h\psi_{m-2},\psi_m,\psi_{m+1}|+|\h\phi_{m};\h\psi_{m-1},\psi_{m+2}|.
\end{aligned}
\end{equation*}
Substituting them into the left hand side of \eqref{nls-b1} one obtains
\begin{eqnarray*}
ff_{xx}-f_xf_x&=&4|\h\phi_{m-1},\phi_{m+1};\h\psi_{m-1},\psi_{m+1}|f-4|\h\phi_{m-1},\phi_{m+1};\h\psi_m|
|\h\phi_{m};\h\psi_{m-1},\psi_{m+1}|\\
&&+2q_0|\h\phi_{m-1};\h\psi_{m+1}|f-2r_0|\h\phi_{m+1};\h\psi_{m-1}|f\\
&=&4|\h\phi_{m-1};\h\psi_{m+1}||\h\phi_{m+1};\h\psi_{m-1}|
+2q_0|\h\phi_{m-1};\h\psi_{m+1}|f-2r_0|\h\phi_{m+1};\h\psi_{m-1}|f\\
&=&-gh-q_0hf-r_0gf,
\end{eqnarray*}
where we have made use of
$$f[(\textrm{Tr}A)f]=[(\textrm{Tr}A)f]^2$$
and the identity
\begin{eqnarray*}
&&|\h\phi_{m-1},\phi_{m+1};\h\psi_{m-1},\psi_{m+1}|f-4|\h\phi_{m-1},\phi_{m+1};\h\psi_m|
|\h\phi_{m};\h\psi_{m-1},\psi_{m+1}|\\
&&+|\h\phi_{m-1};\h\psi_{m+1}||\h\phi_{m+1};\h\psi_{m-1}|=0,
\end{eqnarray*}
which is derived from Lemma \ref{lemma 1}.
Thus, \eqref{nls-b1} is proved.

For \eqref{nls-b2}, we first derive the following relations using Lemma \ref{lemma 2},
\begin{align*}
(\mathrm{Tr}A)|\h\phi_{m+1};\h\psi_{m-1}|=&|\h\phi_{m},\phi_{m+2};\h\psi_{m-1}|
-|\h\phi_{m+1};\h\psi_{m-2},\psi_{m}|,\\
(\mathrm{Tr}A)^2|\h\phi_{m+1};\h\psi_{m-1}|=&|\h\phi_{m-1},\phi_{m+1},\phi_{m+2};\h\psi_{m-1}|
+|\h\phi_{m},\phi_{m+3};\h\psi_{m-1}|
-2|\h\phi_{m},\phi_{m+2};\h\psi_{m-2},\psi_{m}|\\
&+|\h\phi_{m+1};\h\psi_{m-3},\psi_{m-1},\psi_{m}|+|\h\phi_{m+1};\h\psi_{m-2},\psi_{m+1}|.
\end{align*}
Substituting the derivatives of $f$ and $g$ into equation \eqref{nls-b2} we have
\begin{align*}
&(f_{xx}g+fg_{xx}-2f_xg_x-ig_tf+if_tg)/2\\
=&(-|\h\phi_{m-2},\phi_m,\phi_{m+1};\h\psi_m|+3|\h\phi_{m-1},\phi_{m+2};\h\psi_m|
+2|\h\phi_{m-1},\phi_{m+1};\h\psi_{m-1},\psi_{m+1}|\\
&+3|\h\phi_{m};\h\psi_{m-2},\psi_m,\psi_{m+1}|
-|\h\phi_{m};\h\psi_{m-1},\psi_{m+2}|)|\h\phi_{m+1};\h\psi_{m-1}|
+4q_0|\h\phi_{m-1};\h\psi_{m+1}||\h\phi_{m+1};\h\psi_{m-1}|\\
&+(3|\h\phi_{m-1},\phi_{m+1},\phi_{m+2};\h\psi_{m-1}|-|\h\phi_{m},\phi_{m+3};\h\psi_{m-1}|
+2|\h\phi_{m},\phi_{m+2};\h\psi_{m-2},\psi_{m}|\\
&-|\h\phi_{m+1};\h\psi_{m-3},\psi_{m-1},\psi_{m}|+3|\h\phi_{m+1};\h\psi_{m-2},\psi_{m+1}|)f\\
&-q_0(|\h\phi_{m-2},\phi_m,\phi_{m+1};\h\psi_m|+|\h\phi_{m-1},\phi_{m+2};\h\psi_m|
+2|\h\phi_{m-1},\phi_{m+1};\h\psi_{m-1},\psi_{m+1}|\\
&+|\h\phi_{m};\h\psi_{m-2},\psi_m,\psi_{m+1}|+|\h\phi_{m};\h\psi_{m-1},\psi_{m+2}|)f
-q_0f_{xx}f+2q_0r_0|\h\phi_{m+1};\h\psi_{m-1}|f\\
&-2(|\h\phi_{m-1},\phi_{m+1};\h\psi_m|+|\h\phi_{m};\h\psi_{m-1},\psi_{m+1}|)
(|\h\phi_{m},\phi_{m+2};\h\psi_{m-1}|+|\h\phi_{m+1};\h\psi_{m-2},\psi_{m}|-q_0f_x)\\
=&-4(|\h\phi_{m-1},\phi_{m+1};\h\psi_m||\h\phi_{m},\phi_{m+2};\h\psi_{m-1}|
+|\h\phi_{m};\h\psi_{m-1},\psi_{m+1}||\h\phi_{m+1};\h\psi_{m-2},\psi_{m}|)\\
&+4(|\h\phi_{m-1},\phi_{m+1},\phi_{m+2};\h\psi_{m-1}|+|\h\phi_{m+1};\h\psi_{m-2},\psi_{m+1}|)f\\
&+4(|\h\phi_{m-1},\phi_{m+2};\h\psi_m|+|\h\phi_{m};\h\psi_{m-2},\psi_m,\psi_{m+1}|)
|\h\phi_{m+1};\h\psi_{m-1}|\\
&-2q_0(f_{xx}f-f_x^2)+2q_0^2|\h\phi_{m-1};\h\psi_{m+1}|f+4q_0|\h\phi_{m-1};\h\psi_{m+1}|
|\h\phi_{m+1};\h\psi_{m-1}|\\
=&-2q_0(f_{xx}f-f_x^2)-q_0^2hf-q_0gh\\
=&-q_0(f_{xx}f-f_x^2)+q_0r_0gf,
\end{align*}
in which the identity
$$f[(\textrm{Tr}A)^2g]=g[(\textrm{Tr}A)^2f]=[(\textrm{Tr}A)f][(\textrm{Tr}A)g]$$
and relations
\begin{align*}
&|\h\phi_{m-1},\phi_{m+2};\h\psi_m|g/2-|\h\phi_{m-1},\phi_{m+1};\h\psi_m|
|\h\phi_{m},\phi_{m+2};\h\psi_{m-1}|
+|\h\phi_{m-1},\phi_{m+1},\phi_{m+2};\h\psi_{m-1}|f=0,\\
&|\h\phi_{m};\h\psi_{m-2},\psi_m,\psi_{m+1}|g/2-|\h\phi_{m};\h\psi_{m-1},\psi_{m+1}|
|\h\phi_{m+1};\h\psi_{m-2},\psi_{m}|
+|\h\phi_{m+1};\h\psi_{m-2},\psi_{m+1}|f=0
\end{align*}
have also been used. Thus, \eqref{nls-b2} has been proved. The third equation \label{nls-b3}
can be proved in a similar way.

Suppose that $A=P^{-1}JP$, i.e. $A$ is similar to $J$.
We introduce $\phi'=P\phi,~ \psi'=P\psi$,
which satisfy \eqref{pp-x} and \eqref{pp-t} with $J$ in place of $A$.
Then, for the quasi double Wronskians, we have $f(\phi',\psi')=|P|f(\phi,\psi)$, $g(\phi',\psi')=|P|g(\phi,\psi)$
and $h(\phi',\psi')=|P|h(\phi,\psi)$,
which means $A$ and any matrix that is similar to it can generate same $q$ and $r$.

\section{Proof of Theorem \ref{th2}}\label{app-B}

For simplicity we denote
$$F=|\h {m};\h {m}|,~~ s=|\h {m+1};\h {m-1}|,~~ H=|\h {m-1};\h {m+1}|.$$
Direct calculation yields
\begin{align*}
F_{x} =&|\h {m-1},m+1;\h m|+|\h m;\h{m-1},{m+1}|,\\
F_{xx} =&|\h{m-2},m,{m+1};\h m|+|\h{m-1},{m+2};\h m|+2|\h{m-1},{m+1};\h {m-1},{m+1}|\\
&+|\h{m};\h{m-2},m,{m+1}|+|\h{m};\h{m-1},{m+2}|,\\
s_x=&|\h{m},{m+2};\h{m-1}|+|\h{m+1};\h{m-2},{m}|,\\
s_{xx}=&|\h{m-1},{m+1},{m+2};\h{m-1}|+|\h{m},{m+3};\h{m-1}|+2|\h{m},{m+2};\h{m-2},{m}|\\
&+|\h{m+1};\h{m-3},{m-1},{m}|+|\h{m+1};\h{m-2},{m+1}|,\\
iF_t=&-2|\h{m-2},m,{m+1};\h m|+2|\h{m-1},{m+2};\h m|+2|\h{m};\h{m-2},m,{m+1}|\\
&-2|\h{m};\h{m-1},{m+2}|+2q_0(-1)^m|\h{m-1};\h (m+1)|+2r_0(-1)^m|\h{m+1};\h(m-1)|,\\
is_{t}=&-2q_0r_0|\h{m+1};\h(m-1)|-2|\h{m-1},{m+1},{m+2};\h{m-1}|+2|\h{m},{m+3};\h{m-1}|\\
&+2|\h{m+1};\h{m-3},{m-1},{m}|-2|\h{m+1};\h{m-2},{m+1}|\\
&+2q_0(-1)^m(-|\h{m-2},m,{m+1};\h m|+|\h{m-1},{m+1};\h {m-1},{m+1}|\\
& -|\h{m+1};\h{m-2},{m+1}|).
\end{align*}
Taking $\Xi=|\h{m}; \h{m}|$ and (for $1\leq j \leq 2m+2$)
\[ P_{js}=\left\{
\begin{array}{ll}
\partial_x-q_0(\partial_x^{s} \psi)\circ(\partial_x^{s} \phi)^{-1}\circ,~ & 0\leq s \leq m,\\
-\partial_x+r_0(\partial_x^{s} \phi)\circ(\partial_x^{s} \psi)^{-1}\circ,~ & m+1\leq s \leq 2m+2,
\end{array}
\right.
\]
using  \eqref{id-D} we can have
\begin{equation*}
(\textrm{Tr}A)F=|\h {m-1},m+1;\h m|-|\h m;\h{m-1},{m+1}|,
\end{equation*}
and
\begin{align*}
(\textrm{Tr}A)s=&|\h{m},{m+2};\h{m-1}|-|\h{m+1};\h{m-2},{m}|\\
&+q_0(-1)^m(|\h {m-1},m+1;\h m|-|\h m;\h{m-1},{m+1}|),\\
(\textrm{Tr}A)^2F=&|\h{m-2},m,{m+1};\h m|+|\h{m-1},{m+2};\h m|-2|\h{m-1},{m+1};\h {m-1},{m+1}|\\
&+|\h{m};\h{m-2},m,{m+1}|+|\h{m};\h{m-1},{m+2}|+2q_0(-1)^m|\h{m-1};\h{m+1}|\\
&-2r_0(-1)^m|\h{m+1};\h{m-1}|,\\
(\textrm{Tr}A)^2s=&|\h{m-1},{m+1},{m+2};\h{m-1}|+|\h{m},{m+3};\h{m-1}|-2|\h{m},{m+2};\h{m-2},{m}|\\
&+|\h{m+1};\h{m-3},{m-1},{m}|+|\h{m+1};\h{m-2},{m+1}|\\
&+2q_0(-1)^m(|\h{m-1},{m+2};\h m|+|\h{m};\h{m-2},m,{m+1}|-|\h{m-1},{m+1};\h {m-1},{m+1}|)\\
&+2q_0(q_0|\h{m-1};\h{m+1}|-r_0|\h{m+1};\h{m-1}|).
\end{align*}
Substituting the  derivatives of $f$ into the left-hand side of \eqref{nls-b1} one obtains
\begin{eqnarray}\label{6aa}
&&ff_{xx}-f_x^2\\
&=&FF_{xx}-F_x^2\\
&=&4(|\h{m};\h{m}||\h{m-1},m+1;\h{m-1},m+1|-|\h{m-1},{m+1};\h{m}||\h{m};\h{m-1},{m+1}|)\nonumber \\
&&-2q_0(-1)^m|\h{m};\h{m}||\h{m-1};\h{m+1}|+2r_0(-1)^m|\h{m};\h{m}||\h{m+1};\h{m-1}|\nonumber\\
&=&4|\h{m-1};\h{m+1}||\h{m+1};\h{m-1}|-2q_0(-1)^m|\h{m};\h{m}||\h{m-1};\h{m+1}|\nonumber\\
&&+2r_0(-1)^m|\h{m};\h{m}||\h{m+1};\h{m-1}|\\
&=&4Hs-2q_0Hf+2r_0sf,
\end{eqnarray}
where we have made use of the equality
$$F[(\textrm{Tr}A)F]=[(\textrm{Tr}A)F]^2$$
and the relation
\begin{eqnarray*}
&&|\h{m-1};\h{m+1}||\h{m+1};\h{m-1}|-|\h{m};\h{m}||\h{m-1},m+1;\h{m-1},m+1|\\
&&+|\h{m-1},{m+1};\h{m}||\h{m};\h{m-1},{m+1}|=0.
\end{eqnarray*}
Meanwhile, a direct calculation of the right-hand side of \eqref{nls-b1} gives rise to
$4(2Hs-q_0Hf+r_0sf)$. Thus, equation \eqref{nls-b1} is proved.

For equation \eqref{nls-b2}, let us first consider $(D_x^2-iD_t)s\cdot F$. We have
\begin{eqnarray*}
&&s_{xx}F-2s_xF_x+sF_{xx}-i(s_tF-sF_t)\\
&=&(F(s_{xx}-is_t)+s(F_{xx}+iF_t)-2s_xF_x))\\
&=&(3|\h{m-1},{m+1},{m+2};\h{m-1}|-|\h{m},{m+3};\h{m-1}|+2|\h{m},{m+2};\h{m-2},{m}|\\
&&-|\h{m+1};\h{m-3},{m-1},{m}|+3|\h{m+1};\h{m-2},{m+1}|)F+2q_0r_0sF\\
&&+2q_0(-1)^m(|\h{m-2},m,{m+1};\h m|-|\h{m-1},{m+1};\h {m-1},{m+1}|+|\h{m};\h{m-1},{m+2}|)F\\
&&+(-|\h{m-2},m,{m+1};\h m|+3|\h{m-1},{m+2};\h m|+2|\h{m-1},{m+1};\h {m-1},{m+1}|\\
&&+3|\h{m};\h{m-2},m,{m+1}|-|\h{m};\h{m-1},{m+2}|)s+2q_0(-1)^m|\h{m-1};\h{m+1}|s+2r_0(-1)^ms^2\\
&&-2(|\h {m-1},m+1;\h m|+|\h m;\h{m-1},{m+1}|)(|\h{m},{m+2};\h{m-1}|+|\h{m+1};\h{m-2},{m}|).
\end{eqnarray*}
Utilizing identity
\begin{equation*}
F[(\textrm{Tr}A)^2s]=s[(\textrm{Tr}A)^2F]=[(\textrm{Tr}A)F][(\textrm{Tr}A)s],
\end{equation*}
\eqref{nls-b2} gives rise to
\begin{eqnarray*}
&&-4(|\h {m-1},m+1;\h m||\h{m},{m+2};\h{m-1}|+|\h m;\h{m-1},{m+1}||\h{m+1};\h{m-2},{m}|)\\
&&+4F(|\h{m-1},{m+1},{m+2};\h{m-1}|+|\h{m+1};\h{m-2},{m+1}|)\\
&&+4s(|\h{m-1},{m+2};\h m|+|\h{m};\h{m-2},m,{m+1}|)\\
&&-2q_0^2HF+4q_0r_0sf+4q_0(-1)^mHs\\
&=&-2q_0^2HF+4q_0r_0sf+4q_0(-1)^mHs.
\end{eqnarray*}
Then we have
\begin{eqnarray*}
&&(D_x^2-iD_t)s\cdot f\\
&=&(-1)^m(-2q_0^2HF+4q_0r_0sf+4q_0(-1)^mHs)\\
&=&-q_0^2hf-2q_0r_0gf-q_0gh,
\end{eqnarray*}
which proves \eqref{nls-b2}. Equation \eqref{nls-bl3} can be verified similarly.

\end{appendix}

\label{lastpage}

\begin{thebibliography}{99}






\bibitem{AblFLM-SAPM} Ablowitz M J, Feng B F, Luo X D and Musslimani Z H,
         Reverse space-time nonlocal sine-Gordon/sinh-Gordon equations with nonzero boundary conditions,
         {\it Stud. Appl. Math.}, 2018, V.141, 267--307.

\bibitem{AKNS-PRL-1973} Ablowitz M J, Kaup D J, Newell A C and Segur H,
         Nonlinear evolution equations of physical significance,
         {\it Phys. Rev. Lett.}, 1973, V.31, 125--127.


\bibitem{AKNS-SAM-1974} Ablowitz M J, Kaup D J, Newell A C and Segur H,
          The inverse scattering transform-Fourier analysis for nonlinear problems,
          {\it Stud. Appl. Math.}, 1974, V.54, 249--315.

\bibitem{AblLM-Nonl-2020} Ablowitz M J, Luo X D and Musslimani Z H,
          Discrete nonlocal nonlinear Schr\"odinger systems: Integrability, inverse scattering and solitons,
          {\it Nonlinearity}, 2020, V.33, 3653--3707.

\bibitem{AblM-PRL-2013} Ablowitz M J and Musslimani Z H,
         Integrable nonlocal nonlinear Schr\"odinger equation,
         {\it Phys. Rev. Lett.}, 2013, V.110, 064105 (5pp).

\bibitem{AM-Nonl-2016} Ablowitz M J and Musslimani Z H,
         Inverse scattering transform for the integrable nonlocal nonlinear Schr\"odinger equation,
         {\it Nonlinearity}, 2016, V.29, 915--946.


\bibitem{AblM-JPA-2019} Ablowitz M J and Musslimani Z H,
        Integrable nonlocal asymptotic reductions of physically significant nonlinear equations,
        {\it J. Phys. A: Math. Theor.}, 2019, V.52, 15LT02 (8pp)

\bibitem{AblPT-book-2004} Ablowitz M J, Prinari B and Trubatch A D,
          Discrete and Continuous Nonlinear Schr\"odinger Systems,
          Camb. Univ. Press, Cambridge, 2004.

\bibitem{Akhmediev-PRE-2009} Akhmediev N N, Ankiewicz A and Soto-Crespo J M,
           Rogue waves and rational solutions of the nonlinear Schr\"{o}dinger equation,
           {\it Phys. Rev. E}, 2009, V.80, 026601 (9pp).

\bibitem{Akhmediev-JETP-1985} Akhmediev N N, Eleonskii V M and Kulagin N E,
          Generation of periodic trains of picosecond pulses in an optical fiber: Exact solutions,
          {\it Sov. Phys. JETP}, 1985, V.89, 894--899.

\bibitem{Akhmediev breather} Akhmediev N N, Eleonskii V M and Kulagin N E,
           Exact first-order solutions of the nonlinear Schr\"odinger equation,
           {\it Theor. Math. Phys.}, 1987, V.72, 809--818.

\bibitem{BLM-2020} Bilman D,  Ling L M and Miller P D,
          Extreme superposition: Rogue waves of infinite order and the Painlev\'e-III hierarchy,
          {\it Duke Math. J.}, 2020, V.169, 671--760.

\bibitem{BM-2019} Bilman D and Miller P D,
          A robust inverse scattering transform for the focusing nonlinear Schr\"odinger equation,
          {\it Commun. Pure  Appl. Math.}, 2019, V.72, 1722--1805.

\bibitem{Biondini-JMP-2014} Biondini G and Kova\v{c}i\v{c} G,
          Inverse scattering transform for the focusing nonlinear Schr\"odinger equation with nonzero boundary conditions,
          J. Math. Phys., 2014, V.55, 031506 (22pp).

\bibitem{BioW-SAPM-2019} Biondini G and Wang Q,
          Discrete and continuous coupled nonlinear integrable systems via the dressing method,
           {\it Stud. Appl. Math.}, 2019, V.142, 139--161.

\bibitem{CZB-SCS-2008} Chen D Y, Zhang D J and Bi J B,
         New double Wronskian solutions of the AKNS equation,
         {\it Sci.  China Ser. A: Math.}, 2008, V.51, 55--69.

\bibitem{ChenJB-rspa-2018} Chen J and Pelinovsky D E,
           Rogue periodic waves of the focusing nonlinear Schr\"odinger equation,
           {\it Proc. R. Soc. A},  2018, V.474, 20170814 (18pp).

\bibitem{ChenDLZ-2018} Chen K, Deng X, Lou S Y and Zhang D J,
         Solutions of nonlocal equations reduced from the AKNS hierarchy,
         {\it Stud. Appl. Math.}, 2018, V.141, 113--141.

\bibitem{CheZ-AML-2018} Chen K and Zhang D J,
         Solutions of the nonlocal nonlinear Schr\"odinger hierarchy via reduction,
         {\it Appl. Math. Lett.}, 2018, V.75, 82--88.

\bibitem{Chen-Huang-1994} Chen Z Y, Li Z G and Huang N N,
         General soliton solutions of the NLS+ equation under nonvanishing boundary condition,
         {\it Acta. Phys. Sin. 
         }, 1994, V.3, 1--17.

\bibitem{Demontis-JMP-2014} Demontis F, Prinari B, van der Mee C and Vitale F,
          The inverse scattering transform for the focusing nonlinear Schr\"odinger equations
          with asymmetric boundary conditions,
          {\it J. Math. Phys.}, 2014, V.55, 101505 (40pp).

\bibitem{Deng-AMC-2018} Deng X, Lou S Y and Zhang D J,
         Bilinearisation-reduction approach to the nonlocal discrete nonlinear Schr\"odinger equations,
         {\it Appl. Math. Comput.}, 2018, V.332, 477--483.

\bibitem{DubGKM-EJPST-2010} Dubard P, Gaillard P, Klein C and Matveev V B,
         On multi-rogue wave solutions of the NLS equation and positon solutions of the KdV equation,
         {\it Eur. Phys. J. Special Topics}, 2010, V.185, 247--258.

\bibitem{EleKK-SPD-1986} Eleonskii V M, Krichever I M and Kulagin N E,
         Rational multisoliton solutions to the nonlinear Schr\"odinger equation,
         {\it Sov. Phys. Dokl.}, 1986, V.287, 226--228.

\bibitem{Faddeev-1987} Faddeev L D and Takhtajan L A,
          Hamiltonian Methods in the Theory of Solitons,
          Springer, Berlin, 1987.

\bibitem{LingLM-SAM-2019} Feng B F, Ling L M and Takahashi D A,
          Multi-breathers and high order rogue waves for the nonlinear Schr\"odinger equation
          on the elliptic function background,
          {\it Stud. App. Math.}, 2020, V.144, 46--101.

\bibitem{FN-PLA-1983} Freeman N C and Nimmo J J C,
         Soliton solutions of the Korteweg-de Vries and Kadomtsev-Petviashvili equations: The Wronskian technique,
         {\it Phys. Lett. A}, 1983, V.95, 1--3.

\bibitem{Grinevich-Nonlinearity-2021} Grinevich P G and Santini P M, The linear and nonlinear instability of the Akhmediev breather,
         {\it Nonlinearity}, 2021, V.34, 8331--8358.

\bibitem{GuoLL-PRE-2012} Guo B L, Ling L M and Liu Q P,
         Nonlinear Schr\"odinger equation: Generalized Darboux transformation and rogue wave solutions,
         {\it Phys. Rev. E}, 2012, V.85, 026607 (9pp).

\bibitem{GurPKZ-PLA-2020} G\"urses M, Pekcan A and Zheltukhin K,
           Discrete symmetries and nonlocal reductions,
           {\it Phys. Lett. A}, 384 (2020) 126065 (5pp).

\bibitem{Haragus-JNS-2022} Haragus M and Pelinovsky D E, Linear instability of breathers for the focusing nonlinear Schr\"odinger equation,
         {\it J. Nonlinear Sci.}, 2022, V.32, 66 (40pp).

\bibitem{He-2013} He J S, Zhang H R, Wang L H, Porsezian K and Fokas A S,
           Generating mechanism for higher-order rogue waves,
           {\it Phys. Rev. E},  2013, V.87, 052914 (10pp).

\bibitem{Hie-2002} Hietarinta J,
         Scattering of solitons and dromions,
         in: R. Pike, P. Sabatier (Eds.), Scattering: Scattering and Inverse Scattering in Pure and Applied Science,
         Academic Press, London, 2002, pp.1773--1791.

\bibitem{Hirota-1973} Hirota R,
         Exact envelope-soliton solutions of a nonlinear wave equation,
         {\it J. Math. Phys.}, 1973, V.14, 805--809.

\bibitem{Hirota-1974} Hirota R,
          A new form of B\"acklund transformations and its relation to the inverse scattering problem,
          {\it Prog. Theor. Phys.}, 1974, V.52, 1498--1512.

\bibitem{Hirota-Miura-book} Hirota R, Direct method of finding exact solutions of nonlinear evolution equations,
         in: R.M. Miura (Ed.),  B\"acklund Transformations, the Inverse Scattering Method, Solitons, and Their Applications,
         Springer-Verlag, Berlin, 1976, pp.40--68.

\bibitem{Huang-Chen-1993} Huang N N and Chen Z Y,
          Zakharov-Shabat equations for dark doliton dolutions to the NLS equation,
          {\it Commun. Theor. Phys.}, 1993, V.20, 187--194.



\bibitem{KawI-JPSJ-1977} Kawata R and Inoue H,
          Eigen value problem with nonvanishing potentials,
          {\it J. Phys. Soc. Jpn.}, 1977, V.43, 361--362.


\bibitem{KawI-JPSJ-1978} Kawata R and Inoue H,
         Inverse scattering method for the nonlinear evolution equations under nonvanishing conditions,
         {\it J. Phys. Soc. Jpn.}, 1978, V.44, 1722--1729.


\bibitem{Kuz-SP-1977} Kuznetsov E A,
          Solitons in a parametrically unstable plasma.
          {\it Dokl. Akad. Nauk. SSSR.}, 1977, V.236, 575--577.

\bibitem{LP-TMP-2007} Lee J H and Pashaev O K,
          Solitons of the resonant nonlinear Schr\"odinger equation with nontrivial boundary conditions:
          Hirota bilinear method,
          {\it Theor. Math. Phys.}, 2007, V.152, 991--1003.

\bibitem{LWZ-2020} Liu S M, Wu H and Zhang D J,
         New results on the classical and nonlocal Gross-Pitaevskii equation with a parabolic potential,
         {\it Reports Math. Phys.}, 2020, V.86, 271--292.

\bibitem{LWZ-2022} Liu S M, Wang J and Zhang D J,
          Solutions to integrable space-time shifted nonlocal equations,
          {\it Reports Math. Phys.}, 2022, V.89, 199--220.

\bibitem{LWZ-SAPM-2022} Liu S Z, Wang J and Zhang D J,
         The Fokas-Lenells equations: Bilinear approach,
         {\it Stud. Appl. Math.}, 2022, V.148, 651--688.


\bibitem{Lou-SAPM-2019} Lou S Y,
         Prohibitions caused by nonlocality for nonlocal Boussinesq-KdV type systems,
         {\it Stud. Appl. Math.}, 2019, V.143, 123--138.

\bibitem{Lou-CTP-2020} Lou S Y,
          Multi-place physics and multi-place nonlocal systems,
          {\it Commun. Theor. Phys.}, 2020, V.72, 057001 (13pp).

\bibitem{Ma-SAPM-1979} Ma Y C,
         The perturbed plane-wave solutions of the cubic Schr\"odinger equation,
         {\it Stud. Appl. Math.}, 1979, V.60, 43--58.

\bibitem{MatS-book-1991} Matveev V B and Salle M A,
         Darboux Transformations and Solitons, 
         Springer Verlag, Berlin, 1991.

\bibitem{Nimmo-NLS-1983} Nimmo J J C,
        A bilinear B\"{a}cklund transformation for the nonlinear Schr\"{o}dinger equation,
        {\it Phys. Lett. A}, 1983, V.99, 279--280.

\bibitem{OhtY-PRSA-2012} Ohta Y and Yang J K,
         General high-order rogue waves and their dynamics in the nonlinear Schr\"odinger equation,
         {\it Proc. R. Soc. A}, 2012, V.468, 1716--1740.

\bibitem{OnoRBMA-PR-2013} Onorato M, Residori S, Bortolozzo U, Montina A and Arecchi F T,
         Rogue waves and their generating mechanisms in different physical contexts,
         {\it Phys. Reports}, 2013, V.528, 47--89.


\bibitem{Per-JAMS-1983} Peregrine D H,
         Water waves, nonlinear Schr\"odinger equations and their solutions,
         {\it J. Aust. Math. Soc. B}, 1983, V.25, 16--43.

\bibitem{RaoCPMH-PD-2020} Rao D H, Cheng Y, Porsezian K, Mihalache D and He J S,
          PT-symmetric nonlocal Davey-Stewartson I equation: Soliton solutions with nonzero background,
          {\it Physica D}, 2020, V.401, 132180 (28pp).

\bibitem{RybS-JDE-2021} Rybalko Y and Shepelsky D,
       Long-time asymptotics for the nonlocal nonlinear Schr\"odinger equation with step-like initial data,
      {\it J. Diff. Equ.}, 2021, V.270, 694--724.

\bibitem{RybS-CMP-2021} Rybalko Y and Shepelsky D,
       Long-time asymptotics for the integrable nonlocal focusing nonlinear Schr\"odinger equation
       for a family of step-like initial data,
       {\it Commun. Math. Phys.}, 2021, V.382, 87--121.

\bibitem{ShiY-ND-2019} Shi Y, Shen S F and Zhao S L,
           Solutions and connections of nonlocal derivative nonlinear Schr\"odinger equations,
           {\it Nonlinear Dyn.}, 2019, V.95, 1257--1267.

\bibitem{SolRKJ-Nat-2007} Solli D, Ropers C, Koonath P and Jalali B,
         Optical rogue waves,
        {\it Nature}, 2007, V.450, 1054--1057.

\bibitem{Tajiri-PRE-1998} Tajiri M and Watanabe Y,
       Breather solutions to the focusing nonlinear Schr\"{o}dinger equation,
       {\it Phys. Rev. E}, 1998, V.57, 3510--3519.


\bibitem{vander-2021} van der Mee C,
            Focusing NLS equations with nonzero boundary conditions: Triangular representations and direct scattering,
            {\it J. Nonlin. Math. Phys.}, 2021, V.28, 68--89.



\bibitem{WW-CNSNS-2022} Wang J and Wu H,
          On $(2+1)$-dimensional mixed AKNS hierarchy,
          {\it Commun. Nonlin. Sci. Numer. Simul.}, 2022, V.104, 106052  (13pp).

\bibitem{WWZ-2020}	Wang J, Wu H and Zhang D J,
         Solutions of the nonlocal (2+1)-D breaking solitons hierarchy and the negative order AKNS hierarchy,
         {\it Commun. Theor. Phys.}, 2020, V.72, 045002 (12pp).

\bibitem{Hejs-PLA-2017} Wang L H, Yang C H, Wang J and He J S,
          The height of an nth-order fundamental rogue wave for the nonlinear Schr\"{o}dinger equation,
          {\it Phys. Lett. A}, 2017, V.381, 1714--1718.

\bibitem{YY-SAPM-2018} Yang B and Yang J K,
        Transformations between nonlocal and local integrable equations,
        {\it Stud. Appl. Math.}, 2018, V.140, 178--201.

\bibitem{YY-PD-2021a} Yang B and Yang J K,
        Rogue wave patterns in the nonlinear Schr\"odinger equation,
        {\it Physica D}, 2021, V.419, 132850 (14pp).

\bibitem{YJK-PRE-2018} Yang J K,
        Physically significant nonlocal nonlinear Schr\"odinger equation and its soliton solutions,
        {\it Phys. Rev. E}, 2018, V.98, 042202 (13pp).

\bibitem{YY-PD-2021b} Yang B and Yang J K,
        Universal rogue wave patterns associated with the Yablonskii-Vorob'ev polynomial hierarchy,
        {\it Physica D}, 2021, V.425, 132958 (24pp).

\bibitem{ZakS-JETP-1972} Zakharov V E and Shabat A B,
          Exact theory of two-dimensional self-focusing and one-dimensional self-modulation of waves in nolinear media,
          {\it Sov. Phys. JETP}, 1972, V.34, 62--69.

\bibitem{ZakS-JETP-1973} Zakharov V E and Shabat A B,
           Interaction between soliton in a stable medium,
           {\it Sov. Phys. JETP}, 1973, V.37, 823--828.

\bibitem{ZhangDD-SIGMA-2017}  Zhang D D and Zhang D J,
           Rational solutions to the ABS list: Transformation approach,
           {\it SIGMA}, 2017, V.13, 078 (24pp).

\bibitem{DJZ-2006} Zhang D J,
        Notes on solutions in Wronskian form to soliton equations: KdV-type,
        {\it arXiv: nlin.SI/0603008}, (2006) preprint.

\bibitem{ZhangDJ-RMP-2014} Zhang D J, Zhao S L, Sun Y Y and Zhou J,
       Solutions to the modified Korteweg-de Vries equation,
       {\it Rev. Math. Phys.}, 2014, V.26, 1430006 (42pp).

\bibitem{Zhou-SAPM-2018} Zhou Z X,
          Darboux transformations and global explicit solutions for nonlocal Davey-Stewartson I equation,
         {\it Stud. Appl. Math.}, 2018, V.141, 186--204.


\end{thebibliography}
\end{document}